\documentclass[acmsmall]{acmart}

\newcommand{\nop}[1]{}

\usepackage[linesnumbered,ruled,vlined]{algorithm2e}
\usepackage{algpseudocode}
\usepackage{color}
\usepackage{amsmath}
\usepackage{multirow}
\usepackage{subfigure}
\usepackage{soul}
\usepackage{balance} 
\usepackage[normalem]{ulem}
\usepackage{setspace}
\usepackage{flushend}
\usepackage{bm}
\usepackage{mathbbol}
\usepackage{pifont}

\usepackage{tcolorbox}
\usepackage{xcolor}
\usepackage{subcaption}

\usepackage{geometry}
\usepackage{placeins}
\usepackage[dvipsnames]{xcolor}

\newtheorem{example}{Example}
\newtheorem{definition}{Definition}

\definecolor{Xiang}{rgb}{1,0,0}

\usepackage{wrapfig}

\AtBeginDocument{%
  }

\setcopyright{acmlicensed}

\acmJournal{PACMMOD}
\acmYear{2025} \acmVolume{3} \acmNumber{6 (SIGMOD)} \acmArticle{309} \acmMonth{12}\acmDOI{10.1145/3769774}

\begin{document}

\title{Continuous Subgraph Matching via Cost-Model-based Dynamic Vertex Dominance Embeddings (Technical Report)}


\author{Yutong Ye}
\affiliation{
  \institution{East China Normal University}
  \country{China}}
\affiliation{
  \institution{Kent State University}
  \country{USA}}
\email{ytye@ecnu.edu.cn}
\email{yyt8@kent.edu}

\author{Xiang Lian}
\affiliation{
  \institution{Kent State University}
  \country{USA}
}
\email{xlian@kent.edu}

\author{Nan Zhang}
\affiliation{
  \institution{East China Normal University}
  \country{China}
}
\email{51255902058@stu.ecnu.edu.cn}

\author{Mingsong Chen}
\affiliation{
  \institution{East China Normal University}
  \country{China}
}
\email{mschen@sei.ecnu.edu.cn}
\renewcommand{\shortauthors}{Yutong Ye et al.}

\begin{abstract}
  In many real-world applications such as social network analysis, knowledge graph discovery, biological network analytics, and so on, graph data management has become increasingly important and has drawn much attention from the database community. While many graphs (e.g., Twitter, Wikipedia, etc.) are usually evolving over time, it is of great importance to study the \textit{continuous subgraph matching} (CSM) problem, a fundamental, yet challenging, graph operator, which continuously monitors subgraph matching results over dynamic graphs with a stream of edge updates. To efficiently tackle the CSM problem, we carefully design a general CSM processing framework, based on novel \textit{\underline{D}ynam\underline{I}c \underline{V}ertex Dom\underline{IN}ance \underline{E}mbedding} (DIVINE), which maps vertex neighborhoods into an embedding space to enable efficient subgraph matching and incremental maintenance under dynamic updates. Inspired by low pruning power for high-degree vertices, we propose a new \textit{degree grouping} technique to decompose high-degree star patterns into groups of lower-degree star substructures, and devise \textit{degree-aware star substructure synopses} (DAS$^3$) over embeddings of star substructure groups. We develop efficient algorithms to incrementally maintain dynamic graphs and answer CSM queries by traversing DAS$^3$ synopses and applying our designed \textit{vertex dominance} and \textit{range pruning strategies}. Through extensive experiments, we confirm the efficiency of our proposed DIVINE approach over both real and synthetic graphs.
\end{abstract}



\received{April 2025}
\received[revised]{July 2025}
\received[accepted]{August 2025}

\maketitle

\section{Introduction}
For the past few decades, the \textit{subgraph matching} problem has been extensively studied as a fundamental operator in the graph data management \cite{orogat2022smartbench,wasserman1994social,karlebach2008modelling,szklarczyk2015string,chen2009monitoring,zhang2022relative} for many real-world applications such as social network analysis, knowledge graph discovery, and pattern matching in biological networks. Given a large-scale data graph $G$, a subgraph matching query finds all the subgraphs of $G$ that are isomorphic to a given query graph $q$. 

While many previous works \cite{he2008graphs,shang2008taming,bonnici2013subgraph,bi2016efficient,juttner2018vf2++,han2019efficient,bhattarai2019ceci,sun2020memory} usually considered the subgraph matching over static graphs, real-world graph data are often dynamic and changing over time. For example, in social networks, friend relationships among users may be subject to changes (e.g., adding or breaking up with friends); similarly, in collaboration networks, people may start to collaborate on some project and then suddenly stop the collaboration for years. In these scenarios, it is important, yet challenging, to conduct the subgraph matching over such a large-scale, dynamic data graph $G_D$, upon updates (e.g., edge insertions/deletions). In other words, we need to continuously monitor subgraphs (e.g., user communities or collaborative teams) in $G_D$ that follow query graph patterns $q$ in real-world applications (such as online advertising to communities on social networks or finding collaboration teams from bibliographic networks).

\begin{figure}[t]
    \centering
    \subfigure[query graph $q$]{
        \includegraphics[height=3cm]{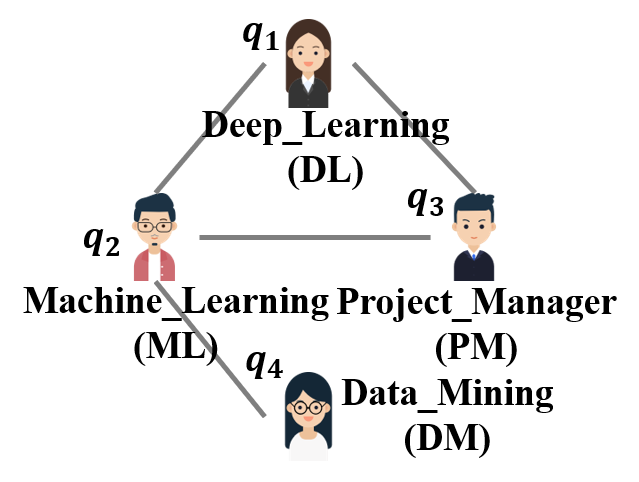}
        \label{subfig:query_graph}
    }
    \qquad
    \subfigure[dynamic collaboration network $G_D$]{
        \includegraphics[height=3.5cm]{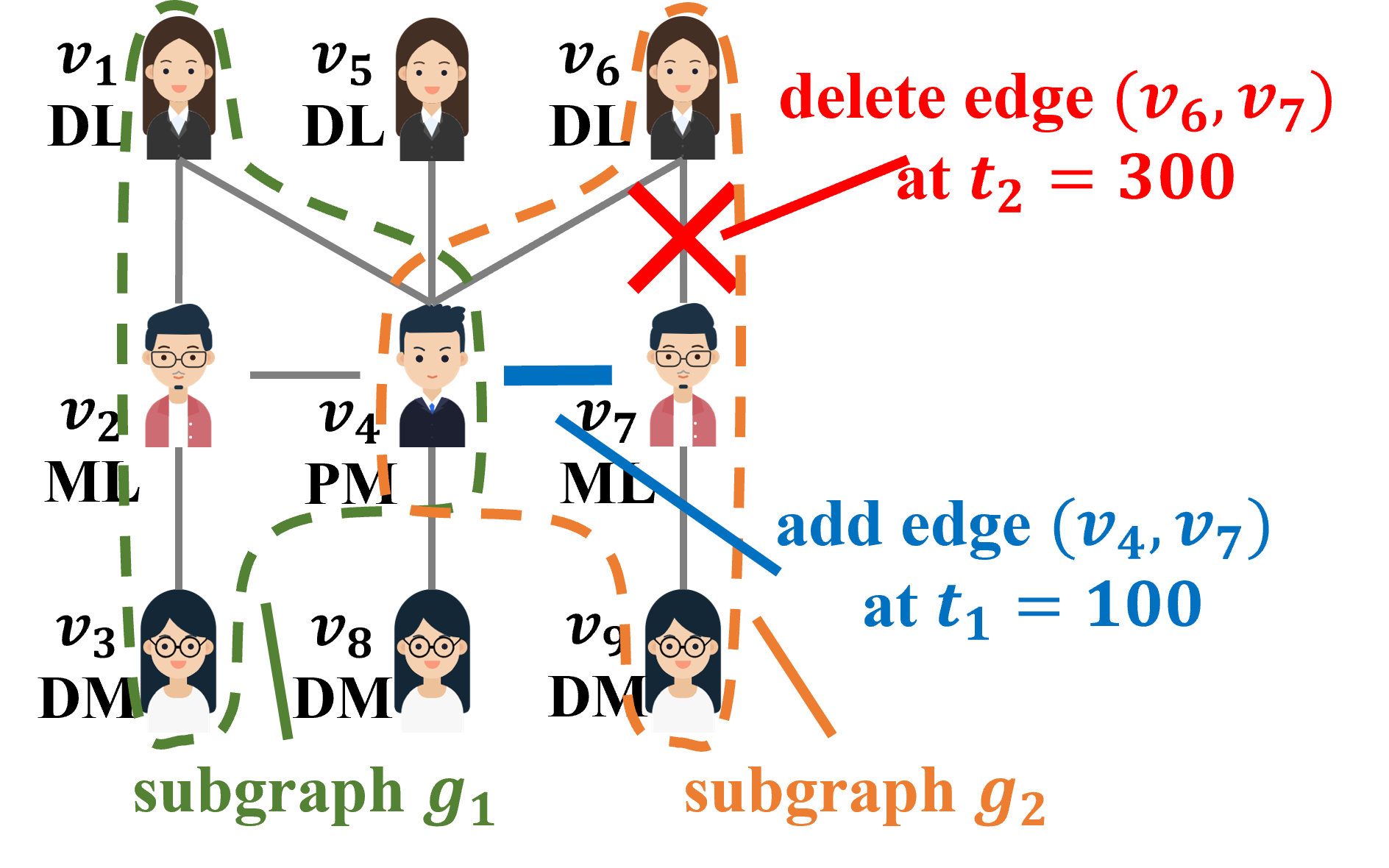}
        \label{subfig:initial_graph}
    }
    \caption{An example of the subgraph matching in dynamic collaboration networks $G_D$.}
    \label{fig:matching}
\end{figure}

Below, we give an example of the subgraph matching over dynamic collaboration networks in the expert team search application.

\begin{example} \textbf{(Monitoring Project Teams in Dynamic Collaboration Networks \cite{fei2013expfinder})} 
Due to numerous requests from departments/companies for recruiting (similar) project teams, a job search advisor may want to monitor the talent market (especially teams of experts) in dynamically changing collaboration networks. 

Consider a toy example of a collaboration network $G_D$ in Figure \ref{fig:matching}. Figure~\ref{subfig:query_graph} shows a desirable expert team pattern $q$, which is represented by a graph of 4 experts $q_1$$\sim$$q_4$, where each expert vertex $q_i$ is associated with one's skill/role keyword (e.g., $q_1$ with keyword ``Deep\_Learning'' and $q_3$ with ``Project\_Manager''), and each edge indicates that two experts had the collaboration before (e.g., edge $(q_1, q_3)$). 

Figure \ref{subfig:initial_graph} provides a dynamic collaboration network $G_D$, where collaboration edges among experts are inserted or deleted over time. For example, a new edge $(v_4, v_7)$ is inserted at timestamp $t_1=100$, implying that two experts $v_4$ and $v_7$ start to collaborate with each other in a team. Similarly, at timestamp $t_2=300$, an existing edge $(v_6, v_7)$ is removed from the graph $G_D$, which may indicate that experts $v_6$ and $v_7$ have not collaborated with each other for a certain period of time (e.g., two years). 

In this example, the job search advisor can register a continuous subgraph matching query over $G_D$ to continuously identify subgraphs of $G_D$ that match with the query graph pattern $q$. At timestamp $t_1$, subgraphs $g_1$ and $g_2$ are the subgraph matching answers, whereas at timestamp $t_2$,  subgraph $g_1$ is the only answer (since the deletion of edge $(v_6, v_7)$ invalidates the subgraph $g_2$).\qquad $\blacksquare$ 
\end{example}

Continuous subgraph matching over dynamic graph has many other real applications. For example, in the application of anomaly detection for online shopping \cite{qiu2018real}, the interactions (e.g., purchases, views, likes, etc.) among users/products over time form a dynamic social/transaction graph, and it is important to detect some abnormal events or fraudulent activities (i.e., query graph patterns) over such a dynamically changing graph.

Inspired by the examples above, we formulate a \textit{continuous subgraph matching} (CSM) query over a dynamic graph, which continuously monitors subgraph matching results, upon updates to the data graph.

\noindent\textbf{Prior Works:} Existing exact continuous matching methods obtain incremental matching results in the following three ways: i) recomputing matching results at each timestamp \cite{fan2013incremental}; ii) direct updating matching results over the data graph \cite{kankanamge2017graphflow}, and; iii) using auxiliary index over matching results to find updates \cite{choudhury2015selectivity,idris2017dynamic,idris2020general,kim2018turboflux,min2021symmetric}.

\noindent\textbf{Our Contributions:} 
Different from existing works that directly use structural information (e.g., vertex label and neighbors’ label set) to filter out vertex candidates, 
In this paper, we design a novel and effective \textit{\underline{D}ynam\underline{I}c \underline{V}ertex Dom\underline{IN}ance \underline{E}mbedding} (DIVINE) approach for candidate vertex/subgraph retrieval. 
Specifically, we propose a dynamic vertex dominance embedding technique, which transforms our CSM problem to a \textit{dominating region search} problem in the embedding space with no false negatives, and allows the generated vertex embeddings to be incrementally maintained in dynamic graphs.
To enhance the pruning power of our embeddings for high-degree vertices, we propose a new \textit{degree grouping} approach for vertex embeddings over basic subgraph patterns in different degree groups (i.e., groups of star substructures). We also devise a cost model to further guide the vertex embedding generation to effectively achieve high pruning power. Finally, we develop algorithms to incrementally maintain dynamic graphs and answer CSM queries efficiently.

Specifically, we make the following main contributions:
\begin{enumerate}
    \item We formally define the \textit{continuous subgraph matching} (CSM) problem in Section~\ref{sec:problem_definition}, and propose a general framework for CSM query answering in Section~\ref{sec:framework}.
    \item We carefully design incremental vertex dominance embeddings to facilitate CSM processing in Section~\ref{sec:vertex_dominance_embeddings}.
    \item We devise a degree grouping technique to enhance the pruning power of vertex embeddings, and construct \textit{degree-aware star substructure synopses} (DAS$^3$) to support the continuous subgraph matching in Section~\ref{sec:star_substructure_synopses}.
    \item We design effective synopsis pruning strategies, and develop an efficient algorithm, named \textit{\underline{D}ynam\underline{I}c \underline{V}ertex Dom\underline{IN}ance \underline{E}mbedding} (DIVINE), for CSM processing in Section~\ref{sec:query_processing}.
    \item We present novel cost-model-guided vertex embeddings to further improve the CSM pruning power in Section~\ref{sec:cost_model}.
    \item We demonstrate through extensive experiments the efficiency and effectiveness of our DIVINE approach over real/synthetic graphs in Section~\ref{sec:experiments}.
\end{enumerate}

\noindent Section \ref{sec:related_work} reviews previous works on dynamic graph management and graph embeddings. Finally, Section \ref{sec:conclusion} concludes this paper.

\section{Problem Definition}
\label{sec:problem_definition}
Table~\ref{tab:notations} depicts the commonly used symbols and their descriptions in this paper.

\subsection{Static Graph Model}
We give the model of a static undirected, vertex-labeled graph $G$. 

\begin{definition}
\textbf{(Static Graph, $G$)} A \textit{static graph}, $G$, is denoted as a triple $(V(G), E(G), L(G))$, where $V(G)$ is a set of vertices $v_i$, $E(G)$ is a set of edges $e_{i,j}=(v_i,v_j)$ between vertices $v_i$ and $v_j$, and $L(G)$ is a label function from each vertex $v_i\in V(G)$ to a label $l(v_i)$.\qquad $\blacksquare$
\label{def:graph}
\end{definition}

The static graph model (as given in Definition~\ref{def:graph}) has been widely used in various real-life applications to reflect relationships between different entities, such as social networks \cite{al2020efficient}, Semantic Web \cite{hassanzadeh2012data}, transportation networks \cite{rai2023top}, biological networks \cite{kan2023rmixup}, citation networks \cite{yang2023revisiting}, and so on.

\noindent {\bf The Graph Isomorphism:} Next, we define the classic graph isomorphism problem between two graphs.

\begin{definition}
\textbf{(Graph Isomorphism)} Given two graphs $G_1$ and $G_2$, graph $G_1$ is \textit{isomorphic} to graph $G_2$ (denoted as $G_1\equiv G_2$), if there exists a bijection mapping function $M: V(G_1)\Rightarrow V(G_2)$, such that: 
i) $\forall v_i\in V(G_1)$, we have $L(v_i)=L(M(v_i))$, and; 
ii) $\forall e_{i,j}\in E(G_1)$, edge $e_{M(v_i),M(v_j)} = (M(v_i), M(v_j)) \in E(G_2)$ holds.\qquad $\blacksquare$
\label{def:graph_isomorphism}
\end{definition}

In Definition~\ref{def:graph_isomorphism}, the graph isomorphism problem checks whether the graphs $G_1$ and $G_2$ exactly match each other.

\noindent {\bf The Subgraph Matching Problem:} The subgraph isomorphism (or subgraph matching) problem is defined as follows.

\begin{definition}
\textbf{(Subgraph Matching)} Given graphs $G$ and $g$, a \textit{subgraph matching} problem identifies subgraphs $g'$ of $G$ (i.e., $g'\subseteq G$) such that $g$ and $g'$ are isomorphic.\qquad $\blacksquare$
\end{definition}

Note that, the subgraph matching problem has been proven to be NP-complete \cite{lewis1983michael}.

\begin{table}[t]\small
\begin{center}
\caption{Symbols and Descriptions}
\label{tab:notations}
\begin{tabular}{|l||l|}
\hline
\textbf{Symbol}&\textbf{Description} \\
\hline\hline
    $G$ & a static graph\\\hline
    $G_D$ & a dynamic graph\\\hline
    $\Delta G_t$ & a set of graph updates at timestamp $t$\\\hline
    $G_t$ & a snapshot graph of $G_D$ at timestamp $t$\\\hline
    $e+$ (or $e-$) & an insertion (or deletion) of edge $e$\\\hline
    $q$ & a query graph\\\hline
    $Q$ & a set of query graph patterns\\\hline
    $x_i$ & a SPUR vector of vertex $v_i$\\\hline
    $y_i$ & a SPAN vector of vertex $v_i$'s 1-hop neighbors\\\hline
    $o(v_i)$ (or $o(g_{v_i})$, $o(s_{v_i})$)& a vertex dominance embedding vector \\\hline
\end{tabular}
\end{center}
\end{table}

\subsection{Dynamic Graph Model}
Real-world graphs are often continuously evolving and updated. Examples of such dynamic graphs include social networks with new or broken friend relationships among users, bibliographical networks with new/discontinued collaboration relationships, and so on. In this subsection, we provide the data model for dynamic graphs with update operations (i.e., edge insertion/deletion) below.

\begin{definition} \textbf{(Dynamic Graph, $G_D$)} A \textit{dynamic graph}, $G_D$, consists of an initial graph $G_0$ (as given by Definition~\ref{def:graph}) and a sequence of graph update operations $\Delta \mathcal{G}=\{\Delta G_1,\Delta G_2, \dots, \Delta G_t, \dots\}$ on $G$. 
Here, $\Delta G_t$ is a graph update operator, in the form of $(e+, t)$ or $(e-, t)$, which indicates either an insertion ($+$) or a deletion ($-$) of an edge $e$ at timestamp $t$, respectively.
\qquad $\blacksquare$
\label{def:dynamic_graph}
\end{definition}

Note that, our dynamic graph model follows the widely used edge-centric update paradigm \cite{choudhury2015selectivity,kankanamge2017graphflow,idris2017dynamic,idris2020general,kim2018turboflux,min2021symmetric,sun2022depth}, which has many real-world dynamic graph applications, such as adding/deleting social relationships or collaborations over time. In this model, the insertion of an edge may implicitly introduce new vertices (i.e., ending vertices of the edge), and the deletion of an edge will not remove (isolated) vertices. Thus, the graph size may monotonically grow, as our dynamic graph model evolves. We would like to study topics of supporting isolated vertex insertions/deletions in dynamic graphs as our future work.

\noindent {\bf Snapshot Graph, $G_t$:} From Definition \ref{def:dynamic_graph}, after applying to initial graph $G_0$ all the graph update operations up to the current timestamp $t$ (i.e., $\Delta G_1$, $\Delta G_2$, $\dots$, and $\Delta G_t$), we can obtain a snapshot of the dynamic graph $G_D$, denoted as $G_t$. 

\noindent {\bf Discussions on the Edge Insertion:} For the edge insertion $(e+, t)$ in Definition \ref{def:dynamic_graph}, there are three cases: 

\begin{itemize}
    \item both ending vertices of $e$ exist in the graph snapshot $G_{t-1}$;
    \item one ending vertex of $e$ exists in $G_{t-1}$ and the other one is a new vertex, and;
    \item both ending vertices of $e$ are new vertices.
\end{itemize}

We will later discuss how to deal with these three cases of edge insertions above for dynamic graph maintenance and incremental query answering.

\subsection{Continuous Subgraph Matching Queries}
In this subsection, we define a \textit{continuous subgraph matching} (CSM) query in a large dynamic graph $G_D$, which continuously monitors subgraph matching answers for a set, $Q$, of query graph patterns.

\begin{definition}
\textbf{(Continuous Subgraph Matching, CSM)} Given a dynamic graph $G_D$, a set, $Q$, of registered query graph patterns $q$, and a current timestamp $t$, a \textit{continuous subgraph matching} (CSM) query maintains a subgraph matching answer set $A(q,t)$ for each query graph pattern $q\in Q$, such that any subgraph $g\in A(q,t)$ in $G_D$ is isomorphic to the query graph $q$ (denoted as $g\equiv q$).

Alternatively, for each $q\in Q$, the CSM problem incrementally computes an update, $\Delta A(q,t)$, to the answer set $A(q,t-1)$, upon the change $\Delta G_t$ to $G_{t-1}$.\quad $\blacksquare$
\label{def:CSM}
\end{definition}

The CSM problem in Definition~\ref{def:CSM} is useful in real applications such as cyber-attack event detection in computer networks and credit card fraud monitoring over transaction networks. In particular, CSM incrementally maintains the subgraph matching answer set, $A(q, t)$, which contains subgraphs $g \subseteq G_t$ isomorphic to query graph $q$ at current timestamp $t$, where $t$ is the current logical time or update step of dynamic graph $G_D$. Note that, the answer set $A(q, t)$ may change for different timestamps $t$, upon edge insertions/deletions. For example, $A(q, 2)$ and $A(q, 5)$ can be different, due to graph updates $\Delta G_3$, $\Delta G_4$, and $\Delta G_5$.

\subsection{Challenges}
The subgraph matching problem in a data graph has been proven to be NP-complete \cite{lewis1983michael}. Thus, it is even more challenging to design efficient and effective techniques for answering subgraph matching queries in the scenario of dynamic graphs. 
Due to frequent updates to the data graph, it is non-trivial how to effectively maintain a large-scale dynamic graph that can support efficiently obtaining initial subgraph matching results, and continuously monitor CSM answer sets, upon fast graph updates.

\begin{algorithm}[t]
\caption{{\bf The DIVINE-based Framework for Continuous Subgraph Matching}}
\label{alg1}
\KwIn{
    i) a dynamic data graph $G_D$;
    ii) a set, $Q$, of registered query graph patterns, and;
    iii) current timestamp $t$\\
}
\KwOut{
    the subgraph matching answer set, $A(q,t)$, for each registered query graph $q\in Q$
}

\tcp{\bf Graph Maintenance Phase}

generate a vertex embedding $o(v_i)$ for each vertex $v_i\in G_0$\\

build $m$ synopses, $Syn_j$, over vertex dominance embeddings with degree grouping\\

\For{each graph update $\Delta G_t\in \Delta \mathcal{G}$}{
    obtain $G_t$ by applying the graph update $\Delta G_t$ over $G_{t-1}$\\
    
    maintain synopses $Syn_j$ by updating the corresponding vertex embeddings\\
}

\tcp{\bf CSM Query Answering Phase}

\tcp{Initial Subgraph Matching Answer Set Generation}

\For{each new query graph $q\in Q$}{
compute a query embedding vector $o(q_i)$ of each vertex $q_i\in V(q)$\\

find candidate vertex sets $q_i.cand\_set$ for query vertices $q_i$ by accessing synopses $Syn_j$\\

enumerate candidate subgraphs $g$ from candidate vertices in $q_i.cand\_set$\\ 

{obtain an initial subgraph matching answer set, $A(q,t)$, of matching subgraphs $g$ ($\equiv q$)}\\
}

\tcp{Continuous Subgraph Matching Answer Set Monitoring}

\For{each graph update $\Delta G_{t}\in \Delta \mathcal{G}$}{
    \For{each query graph $q\in Q$}{
    
        update candidate vertex sets $q_i.cand\_set$ by checking the updated vertices in $\Delta G_{t}$\\
    
        compute answer changes $\Delta A(q,t)$ from the new candidate sets $q_i.cand\_set$\\
        
        apply  changes $\Delta A(q,t)$ to $A(q,t-1)$ and obtain the latest subgraph answers in $A(q,t)$\\
    }
}
\end{algorithm}

\section{Continuous Subgraph Matching Framework}
\label{sec:framework}
Algorithm~\ref{alg1} illustrates a general framework for continuous subgraph matching based on \textit{\underline{D}ynam\underline{I}c \underline{V}ertex Dom\underline{IN}ance \underline{E}mbedding} (DIVINE), which consists of two phases, \textit{graph maintenance} and \textit{CSM query answering phases}. 
That is, we first pre-process the initial graph $G_0$ by constructing $m$ synopses $Syn_j$ over our proposed dynamic vertex dominance embeddings, and incrementally maintain the update (lines 1-5). Here, the subscript $j$ in $Syn_j$ refers to the synopsis structure corresponding to the $j$-th degree group of vertices. Then, we answer the CSM query over synopses via initial answer generation (lines 6-10) and dynamic answer monitoring (lines 11-15).

Specifically, in the graph maintenance phase, we first pre-process the initial graph $G_0$, by constructing $m$ graph synopses, $Syn_j$, over vertex dominance embeddings $o(v_i)$ with $m$ degree groups (discussed later in Section \ref{sec:star_substructure_synopses}), respectively (lines 1-2). Here, the subscript $j$ in $Syn_j$ refers to the synopsis structure corresponding to the $j$-th degree group of vertices. Then, for each graph update operation $\Delta G_t \in \Delta \mathcal{G}$, we incrementally maintain the dynamic graph $G_t$, vertex embeddings, and synopses $Syn_j$ (lines 3-5).

In the CSM query answering phase, for each newly registered query graph $q$, we obtain the initial query answer set, $A(q, t)$, over the current snapshot $G_t$ of dynamic graph $G_D$ (lines 6-10). Then, we monitor the updates in query answer sets (w.r.t. $Q$) over continuously changing dynamic graph $G_t$ (lines 11-15).

In particular, for initial answer set generation, we first compute query vertex embeddings $o(q_i)$ from vertices $q_i \in V(q)$ in each query graph $q$ (lines 6-7). Then, for each query vertex $q_i$, we find and store candidate vertices in a set $q_i.cand\_set$ by accessing $Syn_j$ (line 8). 
After that, we enumerate candidate subgraphs $g$ by assembling candidate vertices in $q_i.cand\_set$ (line 9), and refine/return matching subgraphs $g$ in an initial subgraph matching answer set $A(q,t)$ (line 10).

Next, to monitor CSM query answers, for each graph update operation $\Delta G_{t}\in \Delta \mathcal{G}$ and each query graph $q\in Q$, we check the updated vertices in $\Delta G_t$, update candidate vertex sets $q_i.cand\_set$, and compute the changes $\Delta A(q,t)$ in $A(q,t-1)$ (lines 11-14). Finally, we obtain the CSM query answer set $A(q,t)$ by applying changes $\Delta A(q,t)$ to $A(q,t-1)$ (line 15).

\section{Dynamic Vertex Dominance Embeddings}
\label{sec:vertex_dominance_embeddings}
In this section, we will present our dynamic vertex dominance embeddings, which can be incrementally maintained, preserve dominance relationships between basic subgraph patterns (i.e., unit star subgraphs), and support efficient CSM over dynamic graphs.

\subsection{Preliminaries and Terminologies}
We first introduce two terms of basic subgraph patterns (used for our proposed vertex dominance embedding), that is, \textit{unit star subgraphs} and \textit{star substructures}:

\begin{itemize}
    \item {\bf Unit Star Subgraph $g_{v_i}$:} A unit star subgraph $g_{v_i}$ is defined as a (star) subgraph in $G_t$ containing a center vertex $v_i$ and its one-hop neighbors.
    \item {\bf Star Substructure $s_{v_i}$:} A star substructure $s_{v_i}$ is defined as a (star) subgraph of the unit star subgraph $g_{v_i}$ in $G_t$ (i.e., $s_{v_i} \subseteq g_{v_i}$), which shares the same center vertex $v_i$.
\end{itemize}

As an example in Figure~\ref{subfig:subgraph}, the unit star subgraph $g_{v_1}$ in data graph $G_t$ contains a center vertex $v_1$ and its 1-hop neighbors $v_2$, $v_3$, and $v_4$. Figure \ref{subfig:substructure} shows 8 ($=2^3$) possible star substructures $s_{v_1}$ ($\subseteq g_{v_1}$), which are centered at vertex $v_1$ and with different combinations of $v_1$'s 1-hop neighbors. 

Note that, intuitively, star substructures $s_{v_i}$ are potential query star patterns (i.e., vertex $q_i \in V(q)$ and its 1-hop neighbors) in the query graph $q$ of the CSM problem.

\subsection{Vertex Dominance Embedding}
\label{subsec:embedding}

To tackle the CSM problem, our goal is to avoid costly graph comparisons by transforming them into a search problem over vertex embeddings in a vector space. Thus, we design an embedding method to generate a vector, $o(v_i)$, for each vertex $v_i\in V(G_t)$ based on the labels of vertex $v_i$ and its 1-hop neighbors (i.e., unit star subgraph or star substructure). Specifically, the embedding vector $o(v_i)$ consists of two portions, a \textit{\underline{S}eeded \underline{P}se\underline{U}do \underline{R}andom} (SPUR) vector $x_i$ and a 
\textit{\underline{SP}UR \underline{A}ggregated \underline{N}eighbor} (SPAN) vector $y_i$.

\begin{figure}[t]
    \centering
    \subfigure[{unit star subgraph $g_{v_1}$}]{
        \includegraphics[height=2.2cm]{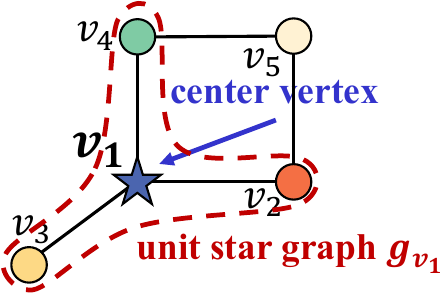}
        \label{subfig:subgraph}
    }
    \qquad
    \subfigure[{8 possible star substructures $s_{v_1}$}]{
        \includegraphics[height=2.2cm]{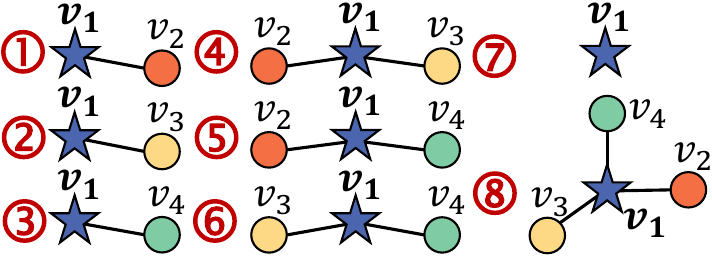}
        \label{subfig:substructure}        
    }
\caption{Illustration of unit star subgraph and substructures.}
\label{fig:unit_star_example}
\end{figure}

\begin{figure}[t]
    \centering
    \includegraphics[width=0.65\textwidth]{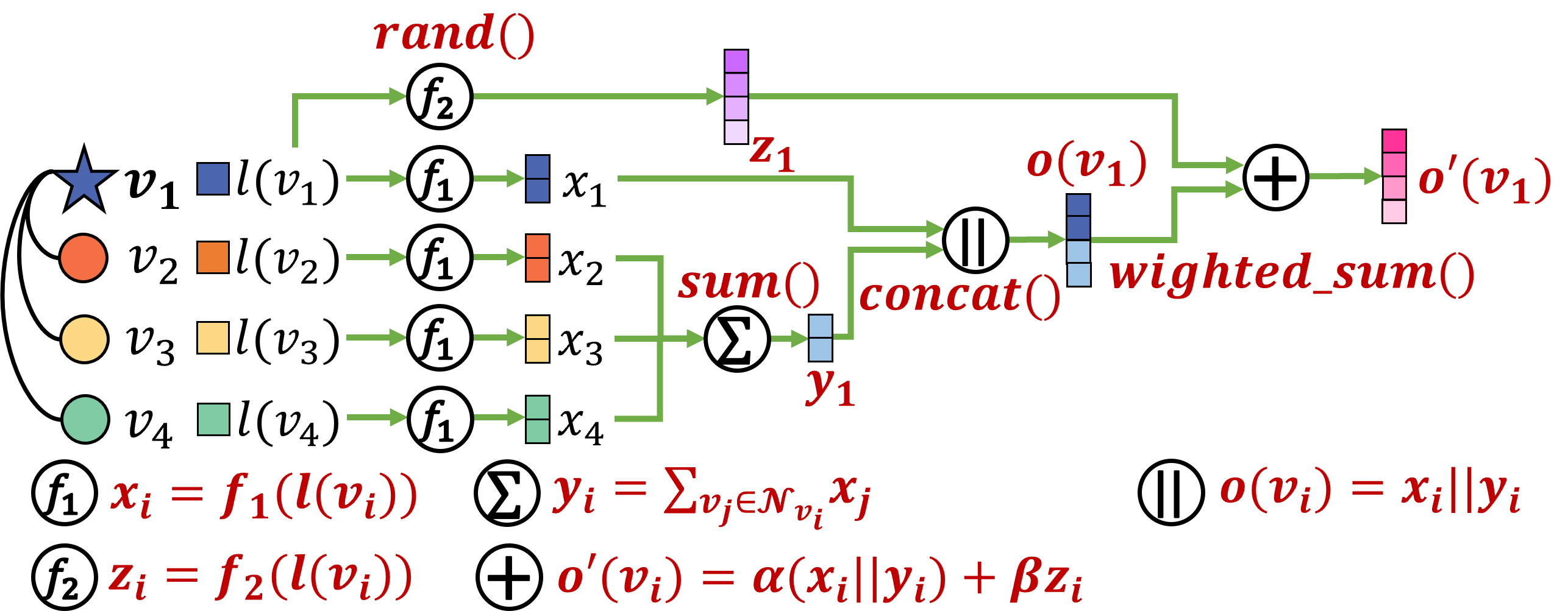}
    \caption{Illustration of vertex dominance embedding $o(v_1)$ and the optimized vertex dominance embedding $o'(v_1)$.}
    \label{fig:randomized_method}
\end{figure}

\noindent{\bf \underline{S}eeded \underline{P}se\underline{U}do \underline{R}andom (SPUR) Vector, \bm{$x_i$}:}
In the dynamic graph $G_D$, each vertex $v_i$ is associated with a label $l(v_i)$, which can be encoded by a nonnegative integer. 
By treating the vertex label $l(v_i)$ as the seed, we can generate a \textit{Seeded PseUdo Random} (SPUR) vector, $x_i$, of a given tunable arity $d$ via a pseudo-random number generator $f_1(\cdot)$, that is,

\begin{equation}
    x_i=f_1(l(v_i)).
    \label{eq:SPUR}
\end{equation}

We denote $x_i[j]$ as the $j$-th element in SPUR vector $x_i$, where $1\leq j\leq d$ and $x_i[j]>0$. 

Note that, since our randomized SPUR vector $x_i$ is generated based on the seed (i.e., label $l(v_i)$ of each vertex $v_i$), vertices with the same label will have the same SPUR vector. 

In the example of Figure \ref{fig:randomized_method}, vertices $v_1 \sim v_4$ are associated with labels $l(v_1) \sim l(v_4)$, respectively. We generate their corresponding SPUR vectors, $x_1 \sim x_4$, via a randomized function $f_1(\cdot)$ in Eq.~(\ref{eq:SPUR}).

\noindent{\bf \underline{SP}UR \underline{A}ggregated \underline{N}eighbor (SPAN) Vector, \bm{$y_i$}:}
For each vertex $v_i$, we also aggregate structural information from its 1-hop neighbors, $v_j \in \mathcal{N}_{v_i}$, in a star subgraph pattern (i.e., unit star subgraph $g_{v_i}$ or star substructure $s_{v_i}$), and sum up their corresponding SPUR vectors $x_j$ to obtain a \textit{SPur Aggregated Neighbor} (SPAN) vector $y_i$:

\begin{equation}
    y_i=\sum_{\forall v_j\in \mathcal{N}_{v_i}}x_j,
    \label{eq:SPAN}
\end{equation}

where $\mathcal{N}_{v_i}$ is a set of 1-hop neighbors of $v_i$, and $x_j$ are the SPUR vectors of 1-hop neighbors of $v_i$. 

Taking Figure~\ref{fig:randomized_method} as an example, the SPAN vector of the vertex $v_1$ is calculated by $y_1=x_2+x_3+x_4$.

\noindent{\bf Vertex Dominance Embedding, \bm{$o(v_i)$}:} 
We define the \textit{vertex dominance embedding} vector, $o(v_i)$, of each vertex $v_i$, by concatenating its SPUR vector $x_i$ and SPAN vector $y_i$:

\begin{equation}
    o(v_i)=x_i \text{ }|| \text{ }y_i,
    \label{eq:embedding}
\end{equation}

where $||$ is the concatenation operation of two vectors, and $x_i$ and $y_i$ are given in Eqs.~(\ref{eq:SPUR}) and~(\ref{eq:SPAN}), respectively. 

Intuitively, vertex dominance embedding $o(v_i)$ encodes label features of a star subgraph pattern (i.e., $g_{v_i}$ or $s_{v_i}$) containing center vertex $v_i$ and its 1-hop neighbors. To distinguish vertex dominance embeddings $o(v_i)$ from different star subgraph patterns, we also use notations $o(g_{v_i})$ or $o(s_{v_i})$.

As shown in Figure~\ref{fig:randomized_method}, the vertex dominance embedding of vertex $v_1$ can be computed by $o(v_1)=x_1\text{ }||\text{ }y_1$.

\begin{figure}[t]
    \centering
    \includegraphics[width=0.65\textwidth]{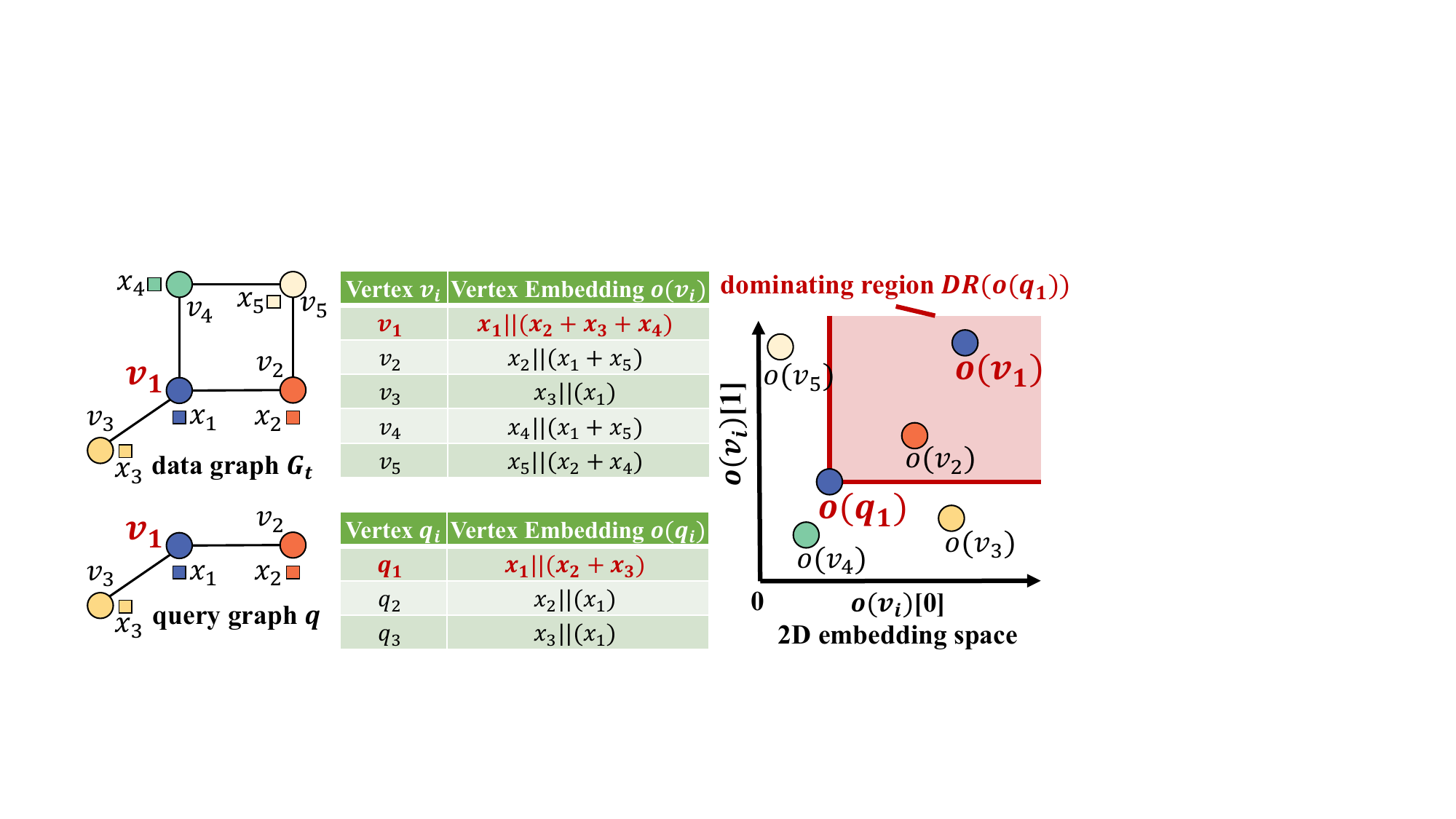}
    \caption{An example of vertex dominance embeddings in subgraph matching  ($||$ is the concatenation of two vectors).}
    \label{fig:static_dominance_example}
\end{figure}

\subsection{Properties of  Vertex Dominance Embedding}
\label{subsec:vertex_properties}

\noindent {\bf The Dominance Property of Vertex Embeddings:} 
From Eq.~(\ref{eq:embedding}), the vertex embedding $o(v_i)$ is a combination of SPUR and SPAN vectors. {\bf Given any unit star subgraph $\bm{g_{v_i}}$ and its star substructure $\bm{s_{v_i}}$ (i.e., $\bm{s_{v_i}\subseteq g_{v_i}}$), their vertex dominance embeddings always satisfy the \textit{dominance} \cite{borzsony2001skyline} (or equality) condition: $\bm{o(s_{v_i})[j]\leq o(g_{v_i})[j]}$, for all dimensions $\bm{1\leq j\leq 2d}$}, denoted as $o(s_{v_i}) \preceq o(g_{v_i})$ (including the case that $o(s_{v_i})=o(g_{v_i})$). 
Here, our (relaxed) dominance semantics include both classic dominance (in skyline definition \cite{borzsony2001skyline}) and equality cases. Formally, for two embedding vectors $o(x), o(y) \in \mathbb{R}^{2d}$, we say that: $o(x)$ dominates $o(y)$, if $o(x)[i] \leq o(y)[i]$ holds for all dimensions $1\leq i  \leq 2d$.

The reason for the property above is as follows. Since star patterns $s_{v_i}$ and $g_{v_i}$ have the same center vertex $v_i$, the vertex embeddings, $o(s_{v_i})$ and $o(g_{v_i})$, for $s_{v_i}$ and $g_{v_i}$, respectively, must share the same SPUR vector $x_i$. Moreover, since $s_{v_i}\subseteq g_{v_i}$ holds, 1-hop neighbors of vertex $v_i$ in $s_{v_i}$ and $g_{v_i}$ must satisfy the condition that $\mathcal{N}_{v_i} (s_{v_i}) \subseteq \mathcal{N}_{v_i} (g_{v_i})$. As the SPAN vectors are defined as the summed SPUR aggregates of 1-hop neighbors in $\mathcal{N}_{v_i} (s_{v_i})$ and $\mathcal{N}_{v_i} (g_{v_i})$, respectively, we have the dominance relationship between SPAN vectors $y_i(s_{v_i})$ and $y_i(g_{v_i})$ (i.e., $y_i(s_{v_i}) \preceq y_i(g_{v_i})$).

\begin{lemma}
\label{lemma:embedding_property}
(The Dominance Property of Vertex Embeddings) Given a unit star subgraph $g_{v_i}$ centered at vertex $v_i$ and any of its star substructures $s_{v_i}$ (i.e., $s_{v_i}\subseteq g_{v_i}$), their vertex embeddings satisfy the dominance condition that: $o(s_{v_i}) \preceq o(g_{v_i})$ (including $o(s_{v_i})=o(g_{v_i})$) in the embedding space.
\end{lemma}

\begin{proof}
For a given unit star subgraph $g_{v_i}$ centered at vertex $v_i$ and any of its star substructures $s_{v_i}$ (i.e., $s_{v_i}\subseteq g_{v_i}$), since they share the same center vertex $v_i$ (i.e., with the same label $l(v_i)$), they must have the same SPUR vector $x_i$. 

Moreover, since $s_{v_i}\subseteq g_{v_i}$ holds, we also have $\mathcal{N}_{v_i} (s_{v_i}) \subseteq \mathcal{N}_{v_i} (g_{v_i})$, where $\mathcal{N}_{v_i} (s_{v_i})$ (or $\mathcal{N}_{v_i} (g_{v_i})$) is a set of $v_i$'s 1-hop neighbors in subgraph $s_{v_i}$ (or $g_{v_i}$). Thus, 
for SPAN vectors $y_i(s_{v_i})$ and $y_i(g_{v_i})$ of $s_{v_i}$ and $g_{v_i}$, respectively, it must hold that: $y_i(s_{v_i})[k] \leq y_i(g_{v_i})[k]$ for all dimensions $k$ (since we have $y_i(s_{v_i})[k] =\sum_{\forall v_j\in \mathcal{N}_{v_i} (s_{v_i})}x_j[k] \leq \sum_{\forall v_j\in \mathcal{N}_{v_i} (g_{v_i})}x_j[k] = y_i(g_{v_i})[k]$). In other words, their SPAN vectors satisfy the condition that $y_i(s_{v_i}) \preceq y_i(g_{v_i})$. Therefore, for $s_{v_i}\subseteq g_{v_i}$, their vertex dominance embeddings satisfy the condition that: $o(s_{v_i}) = (x_i || y_i(s_{v_i})) \preceq (x_i || y_i(g_{v_i})) = o(g_{v_i})$ (including $o(s_{v_i})=o(g_{v_i})$) in the embedding space, which completes the proof.
\end{proof}

\underline{\it The Usage of the Vertex Embedding Property:} 
Note that, a star substructure $s_{v_i}$ ($\subseteq g_{v_i}$) can be a potential query unit star subgraph $s_{q_i}$ from the query graph $q$ (containing a center query vertex $q_i$ and its 1-hop neighbors) for continuous subgraph matching. 

During the subgraph matching, given a query embedding vector $o(q_i)$ of query vertex $q_i$, we can identify candidate data vertices $v_i$ in dynamic graph $G_D$ with embeddings $o(v_i)$ dominated by $o(q_i)$. This way, we can convert the continuous subgraph matching problem into a dominance search problem in the embedding space.

\begin{example}
{\it
Figure~\ref{fig:static_dominance_example} shows an example of our vertex dominance embedding for subgraph matching between a snapshot graph $G_t$ and a query graph $q$. 
Based on Eq.~(\ref{eq:embedding}), we can obtain the embedding for each vertex in $G_t$ or $q$. For example, vertices $v_1$ and $q_1$ have the embeddings $o(v_1)=x_1||(x_2+x_3+x_4)$ and $o(q_1)=x_1||(x_2+x_3)$, respectively.

Since vertex $q_1$ in $q$ matches with vertex $v_1$ in (a subgraph of) $G_t$, we can see that $o(q_1)$ is dominating $o(v_1)$ in a 2D embedding space (based on the property of vertex dominance embedding), which implies that $g_{q_1}$ is potentially a subgraph of (i.e., matching with) $g_{v_1}$. Moreover, although $o(v_2)$ is dominated by $o(q_1)$, $g_{v_2}$ does not match $g_{q_1}$. Thus, in this case, vertex $v_2$ is a false positive during the subgraph candidate retrieval.
On the other hand, since $o(q_1)$ is not dominating $o(v_3)$ in the 2D embedding space, query vertex $q_1$ cannot match with vertex $v_3$ in graph $G_t$.  \qquad $\blacksquare$}
\end{example}

\begin{figure}[t]
    \centering
    \includegraphics[width=0.65\textwidth]{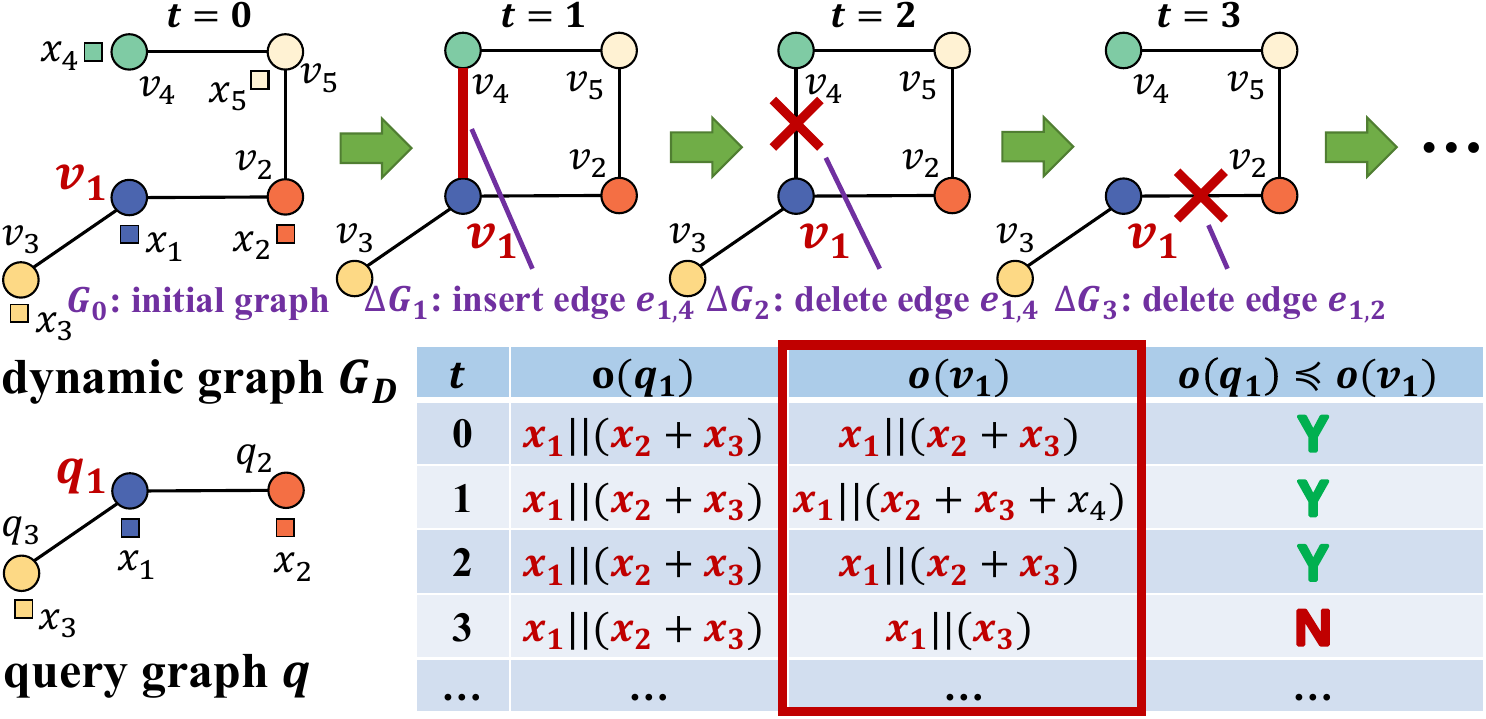}
    \caption{An example of vertex embeddings for CSM.}
    \label{fig:dynamic_dominance_example}
\end{figure}

\noindent {\bf Ease of Incremental Updates for Vertex Dominance Embeddings:} In the dynamic graph $G_D$, our proposed vertex dominance embedding $o(v_i)$ can be incrementally maintained, upon graph update operations in $\Delta G_t$. 
Specifically, any edge insertion $(e+, t)$ or deletion $(e-, t)$ operation (as given in Definition~\ref{def:dynamic_graph}) will affect the embeddings of two ending vertices, $v_i$ and $v_j$, of edge $e = (v_i, v_j)$. We incrementally update the embedding $o(v_i)$ for vertex $v_i$ below. 

\begin{itemize}
    \item When $v_i$ is a newly inserted vertex, we obtain its embedding vector $o(v_i)$ from scratch (i.e., $o(v_i) = x_i || x_j$), and;

    \item When $v_i$ is an existing vertex in dynamic graph $G_D$, for an insertion (or deletion) operator of edge $e = (v_i, v_j)$, we update the SPAN vector, $y_i$, in $o(v_i)$ with $(y_i+x_j)$ (or $(y_i-x_j)$), where $x_j$ is the SPUR vector of vertex $v_j$. 
\end{itemize}    
The case of updating vertex embedding  $o(v_j)$ is similar and omitted.

Note that, for each graph update (i.e., edge insertion $(e+, t)$ or deletion $(e-, t)$), the time complexity of incrementally updating vertex dominance embedding vectors is given by $O(d)$, where $d$ is the dimension of SPUR/SPAN vectors.

\begin{example}
{\it
Figure~\ref{fig:dynamic_dominance_example} illustrates an example of using our vertex dominance embeddings to conduct CSM over the dynamic graph $G_D$ over time $t$ from 0 to 3, where $q$ is a given query graph. Consider two vertices $v_1$ and $q_1$ in graphs $G_D$ and $q$, respectively. At timestamp $t=0$, their embeddings satisfy the dominance relationship, that is, $o(q_1) \preceq o(v_1)$, thus, vertex $v_1$ is a candidate matching with $q_1$. 

At timestamp $t=1$, a new edge $e_{1,4}$ is inserted (i.e., $\Delta G_1$), and the vertex embedding $o(v_1)$ changes from $x_1||(x_2+x_3)$ to $x_1||(x_2+x_3+x_4)$ (by including new neighbor $v_4$'s SPUR vector $x_4$). Since $o(q_1) \preceq o(v_1)$ still holds, $v_1$ remains a candidate matching with $q_1$.

At timestamp $t=2$, since edge $e_{1,4}$ is deleted  (i.e., $\Delta G_2$), the vertex embedding $o(v_1)$ changes back to $x_1||(x_2+x_3)$. At timestamp $t=3$, edge $e_{1,2}$ is deleted (i.e., $\Delta G_3$), and $o(v_1)$ is updated with $x_1||(x_3)$ (by removing the expired neighbor $v_2$'s SPUR vector $x_2$). In this case, $o(q_1) \preceq o(v_1)$ does not hold, and vertex $v_1$ fails to match with $q_1$. \text{ } $\blacksquare$
}
\end{example}

\subsection{Embedding Optimization with  Base Vector $z_i$}
Figure~\ref{fig:node_embedding_1} shows the distributions of 2D SPUR and SPAN vectors, $x_i$ and $y_i$, respectively, in vertex dominance embedding $o(v_i)$ over Yeast graph \cite{sun2020memory} (with 3,112 vertices and 12,519 edges). We can see that, in Figure \ref{subfig:spur_vectors}, the SPUR vectors $x_i$ generated by a seeded randomized function are distributed uniformly in the embedding space, whereas the SPAN vectors (i.e., the summed SPUR vectors) $y_i$ are more clustered along the reverse diagonal line in Figure \ref{subfig:span_vectors}.

As mentioned in Section \ref{subsec:vertex_properties}, given a query vertex $q_i$, we can use its embedding vector $o(q_i)$ as a query point to find the dominated embedding vectors $o(v_i)$ ($= x_i \text{ }||\text{ } y_i$) in the 4D embedding space. Due to the scattered SPUR/SPAN vectors in the embedding space, it is very likely that some false alarms of embedding vectors $o(v_i)$ are included as candidate matching vertices (i.e., dominated by $o(q_i)$).

\begin{figure}[t]
    \centering
    \subfigure[{SPUR vectors $x_i$}]{
        \includegraphics[height=3cm]{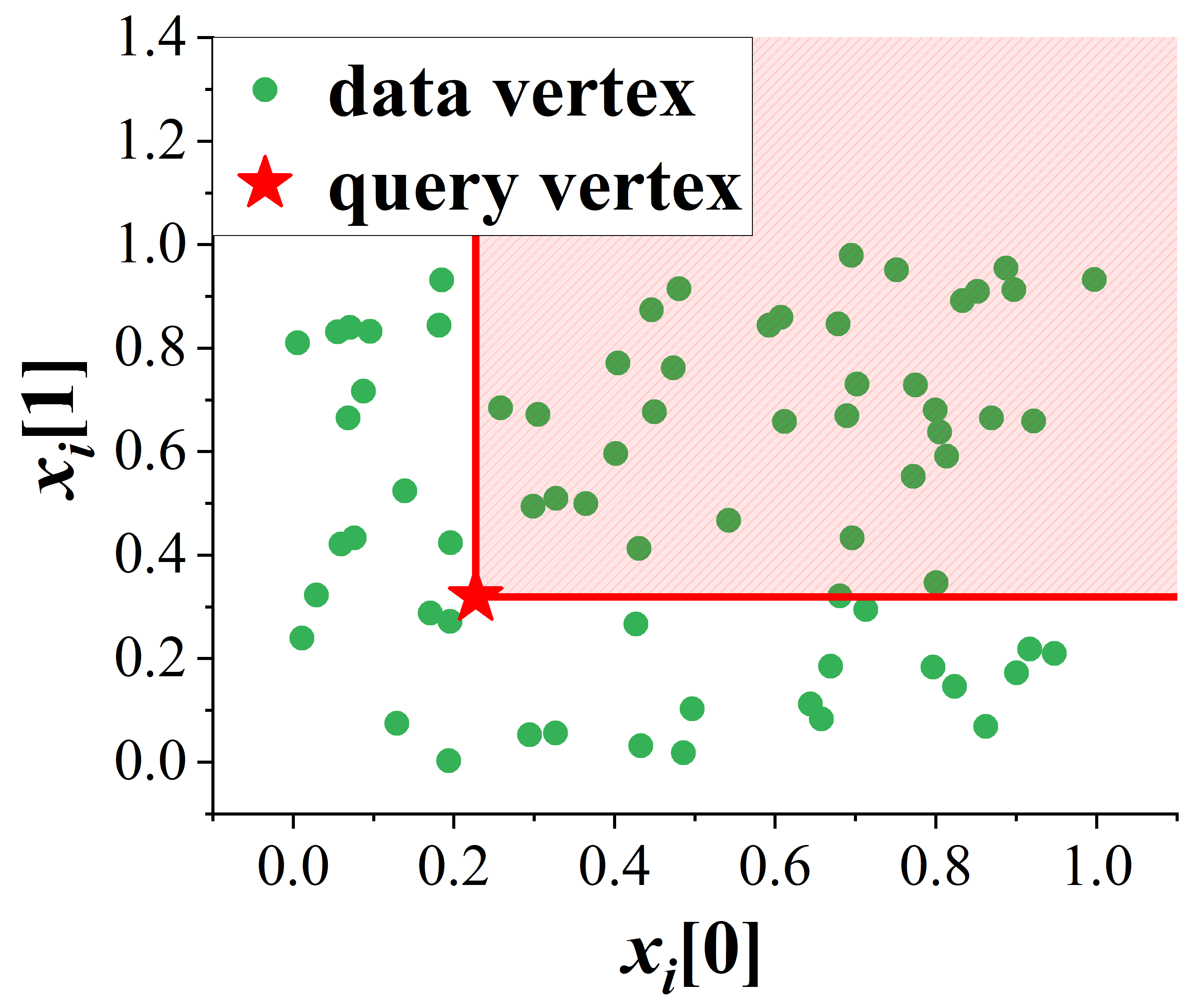}
        \label{subfig:spur_vectors}
    }
    \qquad
    \subfigure[{SPAN vectors $y_i$}]{
        \includegraphics[height=3cm]{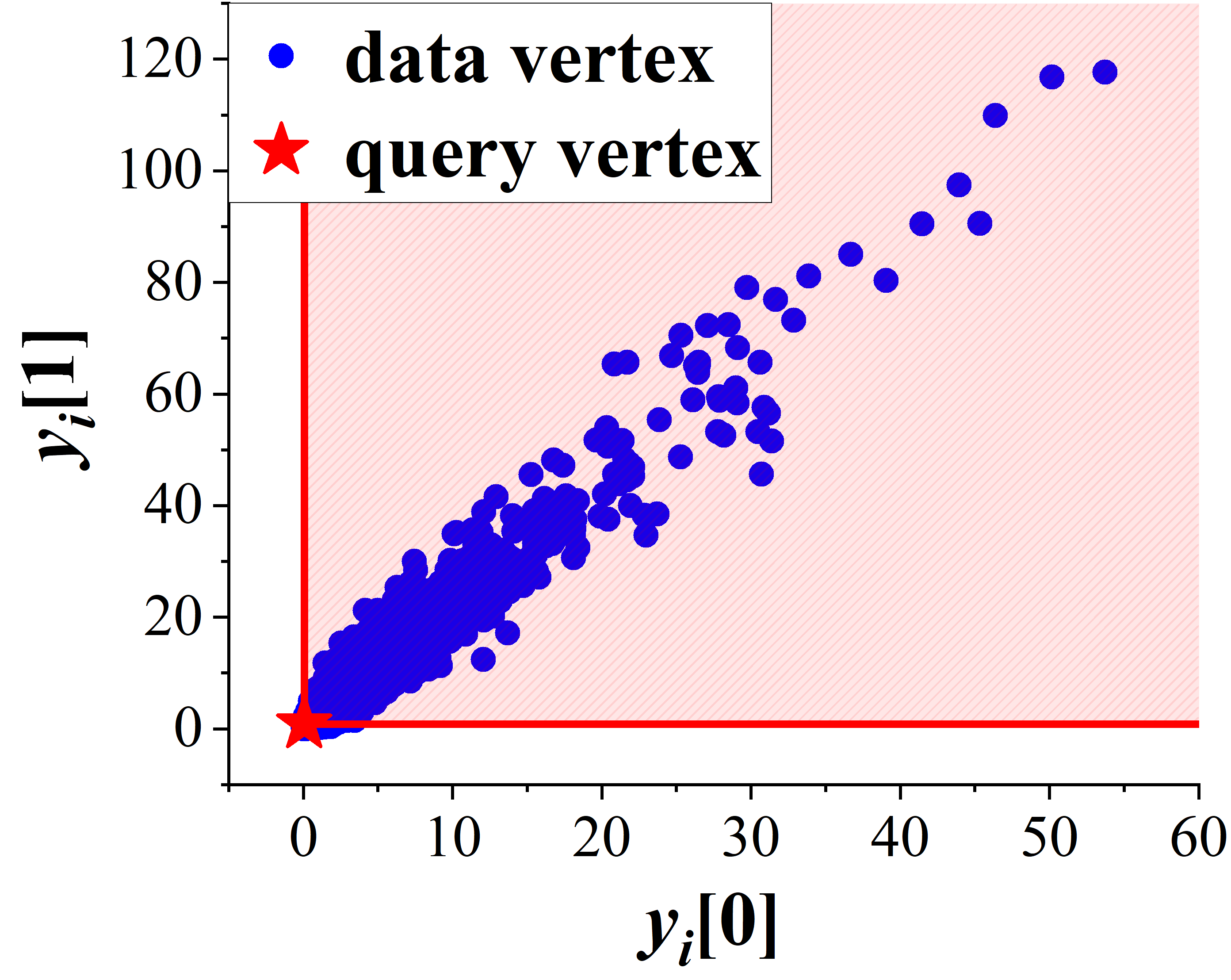}
        \label{subfig:span_vectors}        
    }
\caption{An example of SPUR/SPAN vector distributions in vertex dominance embedding $o(v_i)$.}
\label{fig:node_embedding_1}
\end{figure}

\begin{figure}[t]
    \centering
    \subfigure[{$\alpha x_i+\beta (z_i[0], z_i[1])^T$}]{
        \includegraphics[height=3cm]{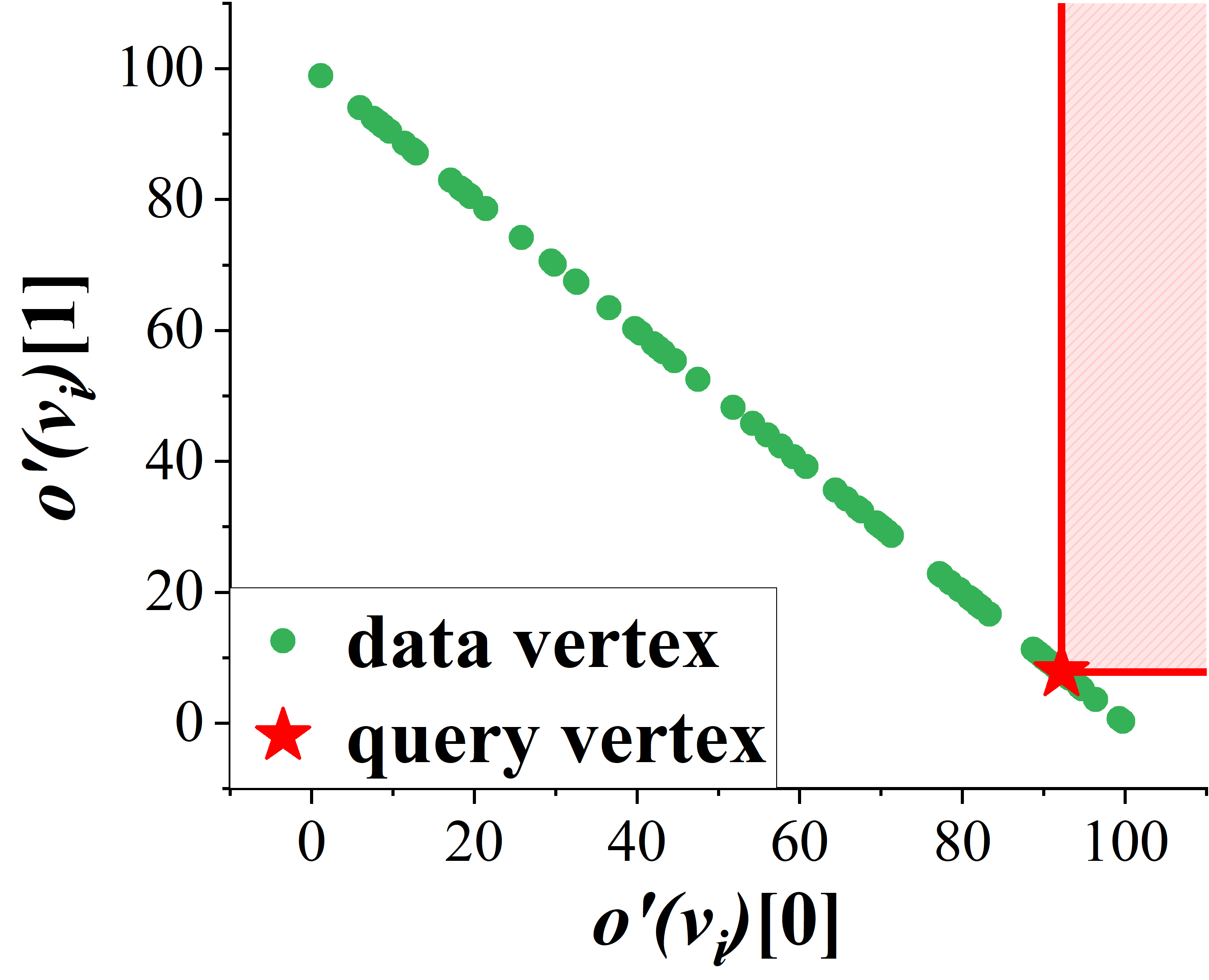}
        \label{subfig:o_prime_x}
    }
    \qquad
    \subfigure[{$\alpha y_i+\beta(z_i[2], z_i[3])^T$}]{
        \includegraphics[height=3cm]{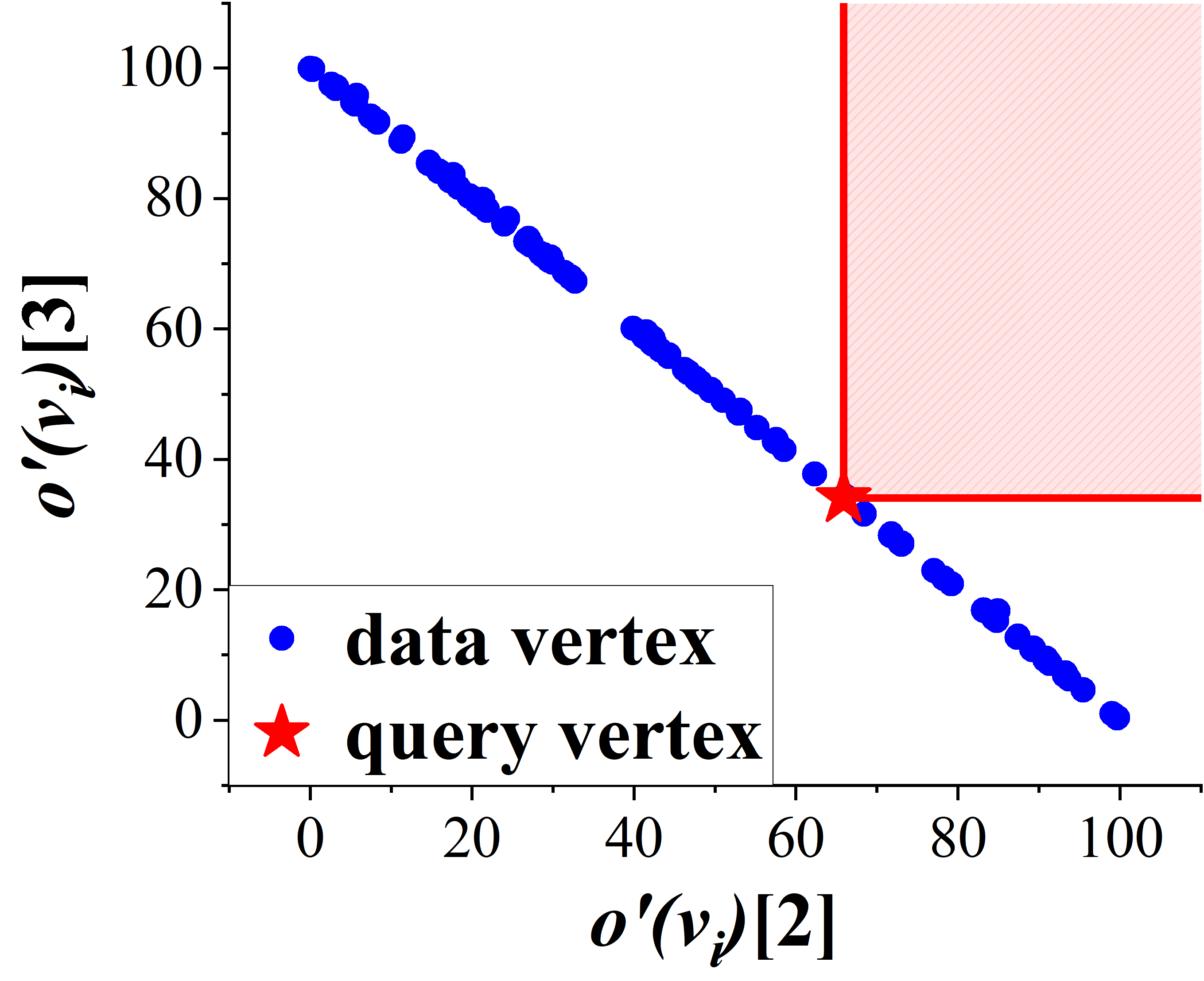}
        \label{subfig:o_prime_y}        
    }
\caption{An example of $o'(v_i)$ distributions with $L_1$-norm diagonal-line base vector $z_i$ ($\alpha=0.01$ and $\beta=100$).}
\label{fig:node_embedding_2}
\end{figure}

In order to further enhance the pruning power, our goal is to update our vertex dominance embedding vector from $o(v_i)$ to $o'(v_i)$, such that the number of candidate vertices dominated by $o(q_i)$ is reduced (i.e., as small as possible). Specifically, we propose to use an optimized embedding $o'(v_i)$, given by:

\begin{equation}
    o'(v_i)=\alpha (x_i||y_i)+\beta z_i,
    \label{eq:new_vertex_embedding}
\end{equation}

where $\alpha$ and $\beta$ are positive constants ($\alpha \ll \beta$), and $z_i=f_2(l(v_i))$ is a random \textit{base vector} of size $(2d)$, generated by a pseudo-random function $f_2(\cdot)$ with $l(v_i)$ as the seed. Please refer to Figure \ref{fig:randomized_method} on how to compute this optimized vertex dominance embedding.

Note that, since the base vector $z_i$ is generated only based on the (seed) label $l(v_i)$ of center vertex $v_i$, different star subgraphs with the same center vertex label $l(v_i)$ will result in the same base vector $z_i$. Therefore, the new (optimized) vertex embedding vector $o'(v_i)$ in Eq.~(\ref{eq:new_vertex_embedding}) still follows the dominance relationship, that is, $o'(q_i) \preceq o'(v_i)$ holds, if a query star pattern $s_{q_i}$ is a subgraph of unit star subgraph $g_{v_i}$ (i.e., $s_{q_i} \subseteq g_{v_i}$ and $l(q_i) \equiv l(v_i)$).

\begin{lemma}
(The Dominance Property of the Optimized Vertex Dominance Embeddings) 
Given a unit star subgraph $g_{v_i}$ centered at vertex $v_i$ and any of its star substructures $s_{v_i}$ (i.e., $s_{v_i}\subseteq g_{v_i}$), their optimized vertex 
dominance embeddings (via base vector $z_i$) satisfy the condition that: $o'(s_{v_i}) \preceq o'(g_{v_i})$ (including $o'(s_{v_i})=o'(g_{v_i})$) in the embedding space. 
\end{lemma}

\begin{proof}
Given a unit star subgraph $g_{v_i}$ centered at vertex $v_i$ and any of its star substructures $s_{v_i}$ (i.e., $s_{v_i}\subseteq g_{v_i}$), since they have the same center vertex $v_i$ (with the same vertex label), their base vectors have the same values, i.e., $z_i(s_{v_i})=z_i(g_{v_i})$, or equivalently $\beta z_i(s_{v_i})=\beta z_i(g_{v_i})$ (for $\beta >0$).

Due to the property of vertex dominance embeddings (as given in Lemma \ref{lemma:embedding_property}), we have $o(s_{v_i}) \preceq o(g_{v_i})$, or equivalently $\alpha o(s_{v_i}) \preceq \alpha o(g_{v_i})$ (for $\alpha>0$). 

Therefore, we can derive that $\alpha o(s_{v_i}) + \beta z_i(s_{v_i}) \preceq \alpha o(g_{v_i}) + \beta z_i(g_{v_i})$, or equivalently $o'(s_{v_i}) \preceq o'(g_{v_i})$ (including $o'(s_{v_i})=o'(g_{v_i})$), which completes the proof.
\end{proof}

\begin{figure}[t]
    \centering
    \includegraphics[width=0.85\textwidth]{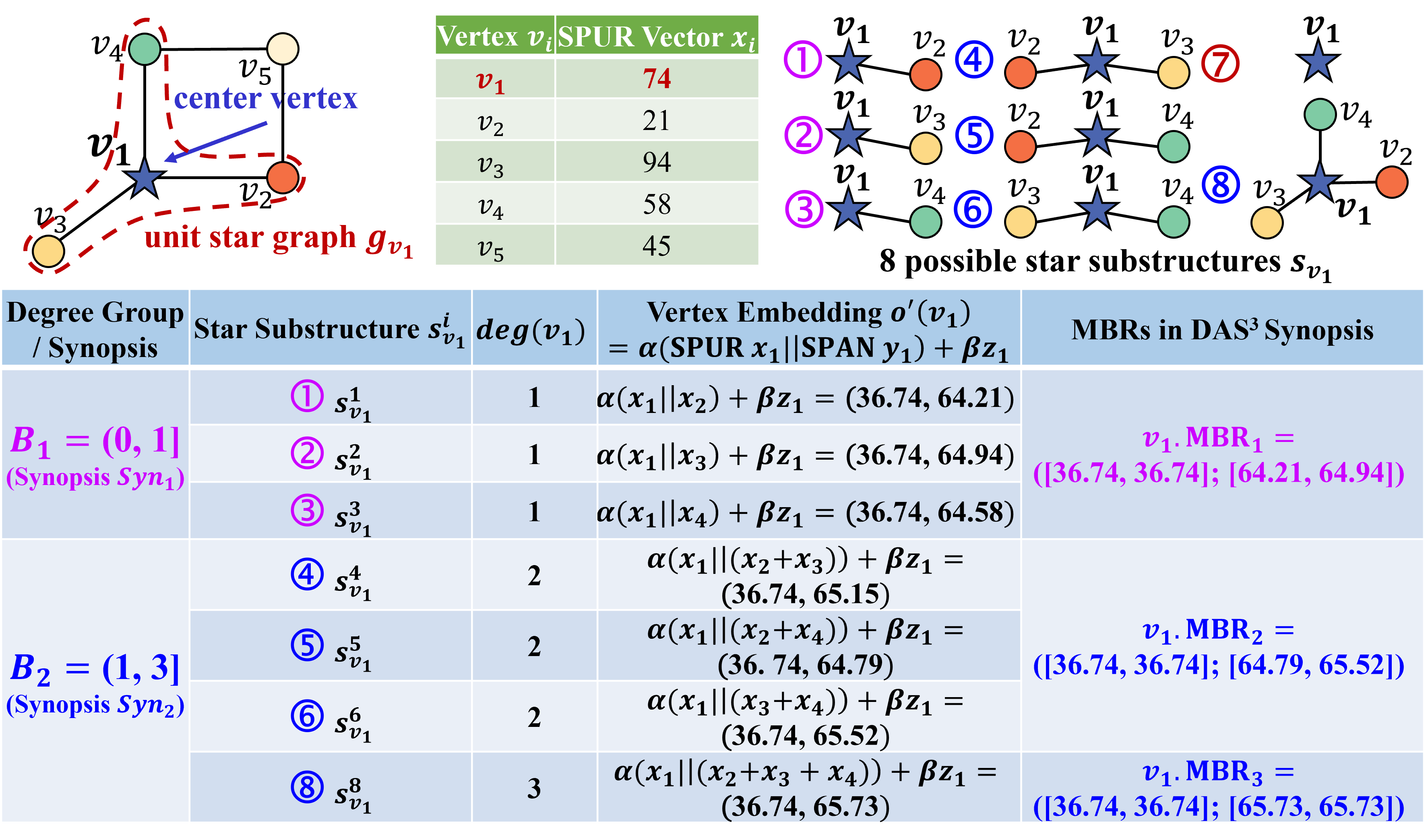}
    \caption{An example of degree grouping and degree-aware star substructure synopses ($\alpha=0.01$, $\beta=100$, and $z_1=(0.36,0.64)$).}
    \label{fig:degree_group_example}
\end{figure}

\noindent {\bf Discussions on How to Design a Base Vector, \bm{$z_i$}: } To improve the pruning power, we would like to make the distribution of the updated vertex embeddings in Eq.~(\ref{eq:new_vertex_embedding}) more dispersed along diagonal line (or plane), so that fewer false alarms (dominated by $o'(q_i)$) can be obtained during continuous subgraph matching. 

In this paper, we propose to design the base vector $z_i$ as a random vector distributed on the unit diagonal line/hyperplane ($L_1$-norm) in the first quadrant. That is, we can normalize $z_i$ to $\frac{z_i}{||z_i||_1}$, where $||\cdot||_1$ is $L_1$-norm. 
Moreover, since constant $\alpha$ is far smaller than $\beta$, $\alpha (x_i || y_i)$ can be considered as noise added to the scaled base vector $\beta z_i$. Thus, vertex embedding vectors $o'(v_i)$ ($=\alpha (x_i||y_i)+\beta z_i$) in Eq.~(\ref{eq:new_vertex_embedding}) are still distributed close to a diagonal line (or hyperplane) in the first quadrant.

We would like to leave interesting topics of using other base vectors $z_i$ (e.g., $L_2$-norm) as our future work.

Figure \ref{fig:node_embedding_2} shows the distributions of the updated 2D SPUR and SPAN vectors (from $o'(v_i)$), by adding a base vector $z_i$, where $z_i$ is given by $L_1$-norm. Since we have $\alpha \ll \beta$ in Eq.~(\ref{eq:new_vertex_embedding}), after adding $\alpha (x_i || y_i)$ to the scaled base vector $\beta z_i$, the previous embedding vector $o(v_i)$ (as shown in Figure \ref{fig:node_embedding_1}) is re-located close to diagonal lines in Figure \ref{fig:node_embedding_2}. This way, we can reduce the number of vertex false alarms dominated by $o(q_i)$ in both SPUR and SPAN spaces, and in turn enhance the CSM pruning power.

\section{Degree-Aware Star Substructure Synopses}
\label{sec:star_substructure_synopses}

We observed that high-degree center vertices of star substructures tend to have vertex dominance embeddings with low pruning power during CSM processing. Therefore, in this section, we will classify star substructures into different degree groups, and build synopses over their embeddings under degree groups. This way, CSM can be conducted over a synopsis corresponding to a group of lower degrees, which can further improve the pruning power.

\subsection{Vertex Embedding via Degree Grouping}
\noindent{\bf Rationale behind the Degree Grouping:} During the subgraph matching over a dynamic graph $G_t$, those vertices $v_i \in V(G_t)$ with high degrees $deg(v_i)$ are less likely to be pruned in the embedding space. This is because, for high degrees, the SPAN vectors in $o'(v_i)$ tend to have large values on all dimensions. Thus, $o'(v_i)$ also tends to be dominated by query embedding vector $o'(q_i)$, and treated as candidate vertices for the refinement (i.e., $v_i$ cannot be pruned). 

Therefore, to enhance the pruning power for high-degree vertices, in this subsection, we propose a novel and effective \textit{degree grouping} technique, which divides all possible star substructures $s_{v_i}$ of each vertex $v_i$ in $G_D$ (that may match with a query star pattern $s_{q_i}$) into different groups (corresponding to different degree intervals of center vertices $v_i$ in star substructures $s_{v_i}$). Intuitively, during the CSM query processing,  the embedding vector $o'(q_i)$  of query star pattern $s_{q_i}$ only needs to be compared with those of star substructures $s_{v_i}$ with similar degrees in a group, rather than all star substructures (including those with high degrees), which can greatly improve the pruning power.

\noindent{\bf Equi-Frequency Degree Grouping:} One straightforward way to perform the degree grouping is to let each single degree value be a degree group. However, since vertex degrees in real graphs usually follow the power-law distribution \cite{barabasi1999emergence,newman2005power,barabasi2003scale,albert2002statistical,dorogovtsev2003evolution}, only a small fraction of vertices have high degrees. That is, only a few unit star subgraphs $g_{v_i}$ have star substructures $s_{v_i}$ of very high degrees, which is not space-/time-efficient to maintain in the dynamic graph $G_t$. 

Inspired by this, in this paper, we propose to maintain vertex embeddings for star substructures with $m$ degree groups (of different degree intervals). We assume that the degree statistics of the initial graph $G_0$ are similar to those in $G_t$. 

This way, we obtain the \textit{probability mass function} of vertex degrees in $G_0$
is $pmf(\delta)=\frac{freq(deg(v_i)\geq \delta)}{|V(G_0)|}$, where $freq(deg(v_i)\geq \delta)$ is the number of vertices $v_i \in V(G_0)$ whose degrees are greater than or equal to a degree value $\delta$. Our goal is to divide the degree interval, $(0, deg_{max}]$, into $m$ degree groups $\mathcal{B}=\{B_1, B_2, \cdots, B_m\}$, such that each \textit{degree bucket} $B_j$ ($=(\delta_{j-1}, \delta_j]$) contains the same (or similar) mass in the distribution of $pmf(\delta)$ (for $\delta \in B_j$).

\begin{example}
Figure~\ref{fig:degree_group_example} illustrates the degree grouping process for star substructures with center vertex $v_1$, where 7 possible star substructures $s_{v_1}^i$ (excluding Case \ding{178}, the isolated vertex $v_1$) are partitioned into two groups with 2 degree intervals, $B_1 = (0, 1]$ and $B_2 = (1, 3]$, respectively. For example, in Case \ding{172}, center vertex $v_1$ in star substructure $s_{v_1}^1$ has degree 1, which thus falls into the first degree bucket $B_1=(0, 1]$. 

A query star pattern $s_{q_1}$ with center vertex $q_1$ of degree 1 will be compared with star substructures in degree bucket $B_1=(0, 1]$ only. In contrast, for a query star pattern $s_{q_1}$ with degree equal to 2 or 3, we only need to search the second degree bucket $B_2$ (with interval $(1, 3]$).
\label{example4}
\end{example}

\subsection{The Construction of \underline{D}egree-\underline{A}ware \underline{S}tar \underline{S}ubstructure \underline{S}ynopses (DAS$^3$)}
\label{subsec:das3}
Based on the degree grouping $\mathcal{B}=\{B_1, B_2, \cdots, B_m\}$, we will construct $m$ \textit{degree-aware star substructure synopses} (DAS$^3$), $Syn_j$ (for $1\leq j\leq m$), which are essentially $m$ grid files \cite{guting1994introduction}, respectively. 
In each grid synopsis $Syn_j$, we store data vertices in different cells according to their (aggregated) vertex embeddings of star substructures in the corresponding degree group $B_j$.

\noindent {\bf Data Structure of DAS$^3$ Synopses:} Specifically, each DAS$^3$ synopsis $Syn_j$ partitions the embedding data space into cells, $C$, of equal size. Each grid cell $C$ is associated with a list, $C.list$, of vertices $v_i$. Each vertex $v_i \in C.list$ contains the following aggregates:

\begin{itemize}
    \item an \textit{embedding upper bound} vector, $v_i.UB_{\delta}$, for all star substructures $s_{v_i}$ with center vertex degrees $\delta\in (\delta_{j-1}, \delta_j]$, and;
    \item a list of MBRs, $v_i.MBR_\delta$ (for $\delta \in (\delta_{j-1}, \delta_j]$), that minimally bound all embedding vectors $o'(v_i)$ of star substructures $s_{v_i}$ with degrees equal to $\delta$.
\end{itemize}

To facilitate the access of cells in DAS$^3$ synopsis $Syn_j$ efficiently, we sort all cells in descending order of their keys $C.key$. Here, for each cell $C$, we use $C.UB$ to denote its maximal corner point (i.e., taking the maximum value of the cell interval on each dimension). The key, $C.key$, is defined as $||C.UB||_2$, where $||Z||_2=\sum_{i=1}^{d}Z[i]^2$ for a $d$-dimensional vector $Z$. 
Note that, if a query vertex embedding $o(q_i)$ dominates $C.UB$, then the key $||o(q_i)||_2$ of query vertex embedding $o(q_i)$ must be smaller than or equal to the cell key $C.key$. Thus, with such a sorted list of cells, we can efficiently access all cells that are fully or partially dominated by a given query vertex embedding $o(q_i)$.

\noindent {\bf DAS$^3$ Synopsis Construction:} Each DAS$^3$ synopsis $Syn_j$ corresponds to a degree group $B_j$ ($1\leq j\leq m$) with degree interval $(\delta_{j-1}, \delta_j]$. 
We build the DAS$^3$ synopsis $Syn_j$ as follows. For each vertex $v_i$ with degree $deg(v_i) >\delta_{j-1}$ in $G_t$ and for each degree value $\delta$ ($\in (\delta_{j-1}, \delta_j]$), we obtain all its star substructures, $s_{v_i}$, with degrees of center vertex $v_i$ equal to $\delta$, compute their embedding vectors $o'(s_{v_i})$, and use a \textit{minimum bounding rectangle} (MBR), $v_i.MBR_\delta$, to minimally bound these embedding vectors.

Let $v_i.UB_{\delta}$ be the maximum corner point of the MBR $v_i.MBR_{\delta}$ (by taking the upper bound of the MBR on each dimension). Then, for vertex $v_i$, we insert vertex $v_i$ into the vertex list, $C.list$, of a cell $C \in \mathcal{I}_j$, into which corner point $v_i.UB_{ub\_\delta}$ falls, where degree upper bound $ub\_\delta = \min\{deg(v_i), \delta_j\}$. Intuitively, this corner point $v_i.UB_{ub\_\delta}$ is the one most likely to be dominated by a query vertex embedding $o'(q_i)$ in this degree group $B_j$. 

Moreover, we also associate the corner point $v_i.UB_{ub\_\delta}$ with the list of MBRs, $v_i.MBR_{\delta}$ (for each $\delta \in (\delta_{j-1}, ub\_\delta]$), in synopsis $Syn_j$, which can be used for further pruning for a specific degree in interval $(\delta_{j-1}, \delta_j]$.

\begin{figure}[t]
    \centering
    \includegraphics[width=0.55\textwidth]{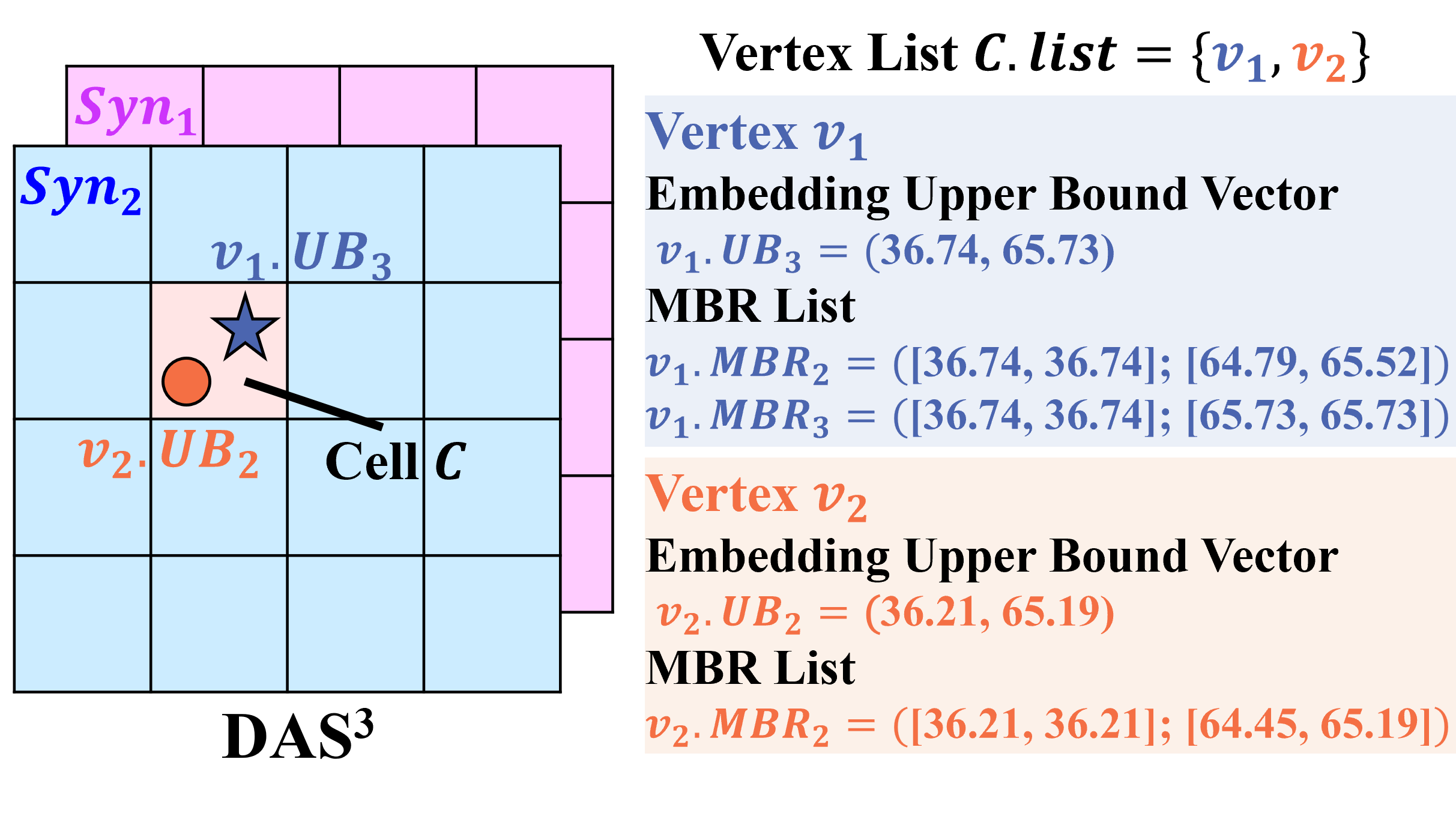}
    \caption{An example of DAS$^3$ data structure ($\alpha=0.01$, $\beta=100$, and $z_1=(0.36,0.64)$).}
    \label{fig:index_example}
\end{figure}

\noindent {\bf Discussions on MBR Computations and Incremental Maintenance:} 
Note that, the MBRs, $v_i.MBR_{\delta}$, mentioned above minimally bound embedding vectors $o'(v_i)$ for all star substructures $s_{v_i}$. 
Since all possible star structures of $v_i$ share the same SPUR vector $x_i$ in $o'(v_i)$, we can directly assign the lower and upper bounds of the $k$-th dimension in the MBR (i.e., $v_i.MBR_{\delta} [2k]$ and $v_i.MBR_{\delta} [2k+1]$, for $0\leq k<d$) as follows:

\begin{equation}
    \hspace{-1ex}
    v_i.MBR_{\delta} [2k]=v_i.MBR_{\delta}[2k+1]=o'(v_i)[k] = \alpha x_i[k]+\beta z_i[k].\label{eq:MBR_SPUR}
\end{equation}

For the SPAN vector part in the MBR, since there are an exponential number of possible star structures (i.e., $2^{deg(v_i)}$), it is not efficient to enumerate all star substructures, compute their embedding vectors, and obtain the MBRs. Therefore, in this paper, we propose a sorting-based method to obtain $v_i.MBR_\delta$ without enumerating all star substructures with degree $\delta$. 

Specifically, for each vertex $v_i$, we maintain $d$ sorted lists, each of which, $v_i.spur\_list_k$ (for $0\leq k < d$), contains the $k$-th elements, $x_j[k]$, of SPUR vectors $x_j$ (obtained from $v_i$'s 1-hop neighbors $v_j\in \mathcal{N}(v_i)$) in ascending order. 

For a degree $\delta$ in the degree group $B_j = (\delta_{j-1},\delta_j]$, we obtain lower and upper bounds of MBR $v_i.MBR_{\delta}$, that is, $v_i.MBR_{\delta} [2k]$ and $v_i.MBR_{\delta}[2k+1]$, respectively, as follows:

\begin{eqnarray}
    v_i.MBR_{\delta} [d+2k]&=&\sum_{deg \in [1,\delta]}v_i.spur\_list_k[deg],\label{eq:MBR1}\\
    v_i.MBR_{\delta}[d+2k+1]&=&\hspace{-5ex}\sum_{deg\in [deg(v_i)-\delta+1,deg(v_i)]}\hspace{-8ex}v_i.spur\_list_k[deg],
    \label{eq:MBR2}
\end{eqnarray}
where $deg(v_i)$ is the degree of vertex $v_i$ in $G_t$ (or the size of the sorted list $v_i.spur\_list_k$).

Intuitively, Eqs.~(\ref{eq:MBR1}) and (\ref{eq:MBR2}) give the lower/upper bounds of the MBR $v_i.MBR_{\delta}$ on the $(d+k)$-th dimension, by summing up the first $\delta$ (smallest) and the last $\delta$ (largest) values, respectively, in the sorted list $v_i.spur\_list_k$. 

\underline{\it The Time Complexity of the MBR Maintenance:} The time complexity of sorting the $k$-th dimension of SPUR vectors (for $1\leq k\leq d$) is given by $O(d\cdot deg(v_i)\cdot log (deg(v_i)))$, which is much less than that of enumerating all star substructures (i.e., $O(d\cdot 2^{deg(v_i)})$).

For the maintenance of a sorted list $v_i.spur\_list_k$ in dynamic graph $G_t$, upon edge insertion or deletion (or insertion/deletion of a 1-hop neighbor $v_j$), we can use a binary search to locate where we need to insert into (or remove from) $v_i.spur\_list_k$  elements $y_j[k]$ in the SPUR vector $y_j$. Therefore, the time complexity of incrementally maintaining $d$ sorted lists $v_i.spur\_list_k$ is given by $O(d\cdot log_2 (deg(v_i)))$.

\underline{\it The Space Complexity of DAS$^3$ Synopses:}
Given a data graph $G_D$ and the dimension, $d$, of SPUR/SPAN vector, the space complexity of vertex embeddings in Section \ref{sec:vertex_dominance_embeddings} is $O(|V(G_D)|\cdot 2d)$. In DAS$^3$, each vertex $v_i$ is associated with an embedding upper bound vector $v_i.UB_{\delta}$ and a list of MBRs, $v_i.MBR_\delta$. Moreover, to support incremental updates to vertex embeddings, each vertex maintains $d$ sorted lists $v_i.spur\_list_k$ (for $0 \leq k < d$). 
Thus, the space complexity of DAS$^3$ synopses is given by $O(|V(G_D)|\cdot 2d)+\sum_{v_i \in V(G_D)}(deg(v_i)\cdot4d+deg(v_i)\cdot d)$.
Since $\sum_{v_i \in V(G_D)} deg(v_i) = 2|E(G_D)|$, the overall space complexity of our DIVINE approach is given by $O(d \cdot (|V(G_D)| + |E(G_D)|))$, which is linear to the graph size.

\begin{example} 
(Continued with Example \ref{example4}) In the example of Figure \ref{fig:degree_group_example}, we consider two degree groups $B_1=(0,1]$ and $B_2=(1,3]$. As shown in Figure~\ref{fig:index_example}, we construct 2 DAS$^3$ synopses (grids of cells with equal size), $Syn_1$ and $Syn_2$, for embeddings $o'(v_1)$ of star substructures $s_{v_1}^i$ falling into degree buckets $B_1$ and $B_2$, respectively. Specifically, in $Syn_2$, we store data vertices $v_1$ and $v_2$ in the cell $C$, which are associated with embedding upper bound vectors, $v_1.UB_3$ and $v_2.UB_2$, and MBR lists, $\{v_1.MBR_2$, $v_1.MBR_3\}$ and $\{v_2.MBR_2\}$, respectively. Here, each MBR $v_i.MBR_\delta$ (e.g., $v_1.MBR_2$ in Figure \ref{fig:degree_group_example}) minimally bounds embedding vectors $o'(v_i)$ of star substructures with degree $\delta$ (i.e., $s_{v_1}^4$, $s_{v_1}^5$, and $s_{v_1}^6$ with degree 2), and each embedding upper bound vector $v_i.UB_{ub\_\delta}$ (e.g., $v_1.UB_3$) is the maximum corner point of MBRs in the MBR list (i.e., maximum corner point of MBRs, $v_1.MBR_2$ and $v_1.MBR_3$).
\label{example5}
\end{example}

\section{CSM Query Processing}
\label{sec:query_processing}
In this section, we will illustrate query processing algorithms based on vertex dominance embeddings in DAS$^3$ synopses to efficiently answer CSM queries.

\subsection{Synopsis Pruning Strategies}
\label{subsec:pruning}
Given a query vertex $q_i$ in the query graph $q$ and SAS$^3$ synopses $Syn_j$, we would like to obtain candidate vertices $v_i$ from synopsis $Syn_j$ which may match with the query vertex $q_i$, where $deg(q_i)$ falls into the degree group $B_j = (\delta_{j-1}, \delta_j]$. 

In this subsection, we present two effective pruning methods over index $Syn_j$, named \textit{embedding dominance pruning} and \textit{MBR range pruning}, which are used to rule out false alarms of cells/vertices.

\noindent{\bf Embedding Dominance Pruning:} We first provide the embedding dominance pruning method, which filters out those cells/vertices in synopsis $Syn_j$ that are not dominated by a given query embedding vector $o'(q_i)$. 

\begin{lemma}
{\bf (Embedding Dominance Pruning)} Given a query embedding vector $o'(q_i)$ of the query vertex $q_i$, any cell $C$ or vertex $v_i$ can be safely pruned, if $o'(q_i)$ does not dominate any portion of cell $C$ or embedding upper bound vector $v_i.UB_\delta$.
\label{lemma:dominance_pruning}
\end{lemma}

\begin{proof}
    If a query vertex $q_i$ in query graph $q$ matches with a data vertex $v_i$ in a subgraph of the data graph, then their vertex dominance embeddings must hold that $o'(q_i)\preceq o'(v_i)$. Therefore, if $o'(q_i)$ does not dominate the cell $C$, then  $o'(q_i)$ cannot dominate any vertex $o'(v_i)$ inside cell $C$, and all vertices in cell $C$ (or cell $C$) can be safely pruned.
    
    Moreover, if $o'(q_i)$ does not dominate embedding upper bound vector $v_i.UB_\delta$, i.e., $v_i.UB_\delta\notin DR(o'(q_i))$, then 
    $o'(q_i)$ does not dominate any star substructure $s_{v_i}$ with center vertex $v_i$ and the corresponding degree group. In other words, $q_i$ does not match $v_i$. Thus, vertex $v_i$ can be safely pruned.
\end{proof}

\noindent{\bf MBR Range Pruning:} Next, we give an effective \textit{MBR range pruning} method, which further utilizes the MBR ranges, $v_i.MBR_{\delta}$, of star substructure embeddings and prunes vertices $v_i$ that do not fall into the MBRs.

\begin{lemma}
{\bf (MBR Range Pruning)}
Given a query embedding $o'(q_i)$ of the query vertex $q_i$ and a vertex $v_i$ in a cell of DAS$^3$ synopsis $Syn_j$ (for $deg(q_i) \in (\delta_{j-1}, \delta_j]$), vertex $v_i$ can be safely pruned, if it holds that $o'(q_i)\notin v_i.MBR_{deg(q_i)}$. 
\label{lemma:mbr_range_pruning}
\end{lemma}

\begin{proof}
    The MBR $v_i.MBR_{deg(q_i)}$ minimally bounds vertex embeddings for all possible star substructures $s_{v_i}$ with center vertex $v_i$ and degree $deg(q_i)$. Thus, if query vertex $q_i \in V(q)$ and its 1-hop neighbors match with some star substructures with the same degree $deg(q_i)$, then its query embedding vector $o'(q_i)$ must fall into this MBR $v_i.MBR_{deg(q_i)}$. Therefore, if this condition does not hold, i.e., $o'(q_i)\notin v_i.MBR_{deg(q_i)}$, then $q_i$ does not match with $v_i$, and $v_i$ can be safely pruned, which completes the proof.   
\end{proof}

\begin{algorithm}[t]
\caption{\bf Initial CSM Answer Set Generation}
\label{alg2}
\KwIn{
    i) a snapshot graph $G_t$ at current timestamp $t$; 
    ii) a DAS$^3$ synopsis $Syn_j$ on $G_t$, and;
    iii) a query graph $q$\\
}
\KwOut{
    a set, $A(q,t)$, of initial subgraph matching results in $G_t$
}

compute vertex dominance embeddings $o'(q_i)$ for all $q_i\in V(q)$\\
\tcp{traverse synopses $Syn_j$ to obtain candidate vertices}
\For{each query vertex $q_i\in V(q)$}{
    obtain a synopsis $Syn_j$ such that $deg(q_i) \in (\delta_{j-1}, \delta_j]$\\
    \For{each cell $C\in Syn_j$ with $C.key \geq ||o'(q_i)||_2$}{
        \For{each candidate vertex $v_i\in C.list$}{
            \If{$o'(q_i)\preceq v_i.UB_\delta$ }{
                \If{$o'(q_i)\in v_i.MBR_{deg(q_i)}$ }{
                    $q_i.cand\_set\leftarrow q_i.cand\_set\cup \{v_i\}$
                    \quad \tcp{Lemmas~\ref{lemma:dominance_pruning} and ~\ref{lemma:mbr_range_pruning}}
                }
            }
        }
    }
}
\tcp{generate a query plan $Q$}

select a query vertex $q_i$ with the smallest candidate set size $|q_i.cand\_set|$ as the first query vertex in an ordered list $Q$\\

iteratively select a neighbor $q_k$ of $q_i\in Q$ with minimum $|q_k.cand\_set|$ and append $q_k$ to $Q$, until all query vertices in $V(q)$ have been added to $Q$\\ 

\tcp{assemble and refine candidate subgraphs}
$A(q,t)\leftarrow \emptyset$\\
invoke {\sf Refinement}$(q,Q,G_t,A(q,t),\emptyset,0)$ to obtain actual CSM query answers in $A(q, t)$\\
\Return $A(q,t)$;\\
\end{algorithm}

\subsection{Initial CSM Answer Set Generation}
\label{subsec:sdsm}
Algorithm~\ref{alg2} illustrates the algorithm of generating the initial subgraph matching answer set, $A(q, t)$, over a snapshot $G_t$ of a dynamic subgraph $G_D$, by traversing synopses over vertex dominance embeddings. 
Specifically, given a query graph $q$, we first obtain vertex dominance embeddings $o'(q_i)$ for each query vertex $q_i \in V(q)$ (line 1), and then traverse DAS$^3$ synopses, $Syn_j$, to retrieve vertex candidate sets $q_i.cand\_set$ (lines 2-8). Next, we generate a query plan $Q$, which is an ordered list of connected query vertices that can be used to join their corresponding candidate vertices (lines 9-10). Finally, we assemble candidate vertices into candidate subgraphs, and refine candidate subgraphs by the \textit{left-deep join based method} \cite{kankanamge2017graphflow,shang2008taming} by invoking function {\sf Refinement($\cdot$)}, and return the initial subgraph matching answers in $A(q,t)$ (lines 11-13).

\noindent{\bf The Synopsis Search:}
For each query vertex $q_i\in V(q)$, we first need to find a synopsis $Syn_j$ that matches the degree group of $q_i$ (i.e., $deg(q_i) \in (\delta_{j-1}, \delta_j]$ must hold; line 3).  
As mentioned in Section~\ref{subsec:das3}, cells, $C$, in the DAS$^3$ synopsis $Syn_j$ are sorted in descending order of their keys $C.key$. We thus only need to search for those cells $C \in Syn_j$ satisfying the condition that $C.key\geq ||o'(q_i)||_2$ (line 4). Then, for each candidate vertex $v_i \in C.list$, we apply our proposed embedding dominance and MBR range pruning strategies (as discussed in Lemmas \ref{lemma:dominance_pruning} and \ref{lemma:mbr_range_pruning}, respectively), and obtain the candidate vertex set, $q_i.cand\_set$, for the query vertex $q_i$ (lines 5-8).

\begin{example}
(Continued with Example \ref{example5}) In DAS$^3$ of Figure \ref{fig:index_example}, the query vertex $q_1$ with degree 2 will only be compared with data vertices in synopsis $Syn_2$ (degree bucket $B_2 = (1,3]$). As mentioned above, by checking the dominance condition between embedding upper bound vector $q_1.UB_2$ and the cells' keys $C.key$, we can obtain candidate cells. Then, for each vertex $v_i$ in these cells, we apply our proposed embedding dominance and MBR range pruning strategies (as discussed in Lemmas \ref{lemma:dominance_pruning} and \ref{lemma:mbr_range_pruning}, respectively). Specifically, to implement embedding dominance and MBR range pruning strategies, we compare embedding upper bound vectors $q_1.UB_2$ and $v_i.UB_{\delta}$ (e.g., $v_1.UB_3$ and $v_2.UB_2$), and MBRs with same degree $q_1.MBR_2$ and $v_i.MBR_2$ (e.g., $v_1.MBR_2$ and $v_2.MBR_2$), respectively. After that, we can obtain the $q_1$'s candidate vertex set, $q_1.cand\_set$.
\label{example6}
\end{example}

\noindent{\bf Query Plan Generation:}
After obtaining candidate vertices, we generate a query plan $Q$ for refinement. Intuitively, we would like to reduce the size of intermediate join results. Therefore, we first initialize $Q$ with a query vertex $q_i\in V(q)$ having the minimum candidate set size $|q_i.cand\_set|$ (line 9), and then iteratively add to $Q$ the one, $q_k$, that is connected with those selected query vertices in $Q$ and having the smallest $|q_k.cand\_set|$ value (line 10).

\begin{algorithm}[t]
\caption{{\bf Refinement}}
\label{alg3}
\KwIn{
    i) a query graph $q$;
    ii) a sorted list (query plan) $Q$;
    iii) a snapshot graph $G_t$;
    iv) matching results set $A$;
    v) a sorted list, $M$, of vertices matching with query vertices in $Q$, and;
    vi) the recursion depth $n$\\
}
\KwOut{
    a set, $A$, of subgraph matching results
}

\eIf{$n=|Q|$}{
     $A\leftarrow A\cup\{M\}$\\
     \Return;\\
}
{   
    $S_{cand} =\emptyset$\\
    \eIf{$n=0$}{
        $S_{cand}\leftarrow S_{cand}\cup Q[n].cand\_set$\\
    }
    {
        \For{each candidate vertex $u\in Q[n].cand\_set$ and $u\notin M$}{
            \If{edges $(M[i], u)$ exist in $E(G_t)$ for all edges $(Q[i], Q[n]) \in E(q)$ (for $0\leq i <n$)}{
                $S_{cand}\leftarrow S_{cand}\cup \{u\}$\\
            }
            
        }
    }
}
\For{each candidate vertex $u\in S_{cand}$}{
        $M[n]=u$\\
        {\sf Refinement ($q, Q,G_t,A,M,n+1$)};\\
}
\end{algorithm}

\noindent{\bf Refinement:} 
Algorithm~\ref{alg3} uses a \textit{left-deep join based method} \cite{kankanamge2017graphflow,shang2008taming} to assemble candidate vertices in $q_i.cand\_set$ into candidate subgraphs, and obtain the actual matching subgraph answer set $A(q,t)$ (i.e., lines 11-13 of Algorithm \ref{alg2}). 

Specifically, we maintain a vertex vector $M$ that will store vertices of the subgraph $g$ matching with ordered query vertices in $Q$. To enumerate all matching subgraphs, each time we recursively expand partial matching results by including a new candidate vertex $u$ as $M[n]$ that maps with the $n$-th query vertex $Q[n]$. If we find all vertices in $M$ mapping with $Q$ (i.e., recursion depth $n=|Q|$), then we can add $M$ to the answer set $A$ (lines 1-3). Otherwise, we will find a vertex candidate set $S_{cand}$ from $Q[n].cand\_set$, such that each vertex $u$ in $S_{cand}$ has the same edge connections to that in $M$ as that in $Q$ (lines 5-11). Next, for each vertex $u \in S_{cand}$, we treat it as vertex $M[n]$ matching with $Q[n]$, and recursively call function {\sf Refinement($\cdot$)} with more depth $n+1$ (lines 13-17). 

After the recursive function {\sf Refinement($\cdot$)} has been executed, the answer set $A$ (or $A(q, t)$ in Algorithm \ref{alg3}) will contain a set of actual subgraph matching results.

\noindent {\bf Complexity Analysis:}
In Algorithm~\ref{alg2}, since we need to access each query vertex $q_i$ and its 1-hop neighbors $\mathcal{N}_{v_i}$, the time complexity of computing vertex dominance embeddings $o'(q_i)$ (line 1) is given by $O(|V(q)|+|E(q)|)$.

For the synopsis traversal (lines 2-8), assume that $PP_C$ is the pruning power of the cells' key value and $PP_v$ is the pruning power in each cell's point list. Thus, the synopsis traversal cost is $O(|V(q)|\cdot (1-PP_C)K^{d}\cdot (1-PP_v)|C.list|)$, where $K^d$ is the number of cells and $|C.list|$ is the number of vertices in the cell $C$.

Next, for the greedy-based query plan generation (lines 9-10), we need to iteratively select a neighbor of vertices in the query plan $Q$, which requires $O(|V(q)|^2)$ cost.

Finally, we invoke the recursive function {\sf Refinement($\cdot$)} to find actual subgraph matching results (lines 11-13), with the worst-case time complexity 
$O(\prod_{i=0}^{|V(q)|-1}|q_i.cand\_set|)$.

Therefore, the overall time complexity of Algorithm~\ref{alg2} is given by: $O(|V(q)|+|E(q)|+|V(q)|\cdot (1-PP_C)K^{d}\cdot (1-PP_v)|C.list|+\prod_{i=0}^{|V(q)|-1}|q_i.cand\_set|+|V(q)|^2)$.

\begin{algorithm}[t]
\caption{{\bf The Query Answering Algorithm for CSM Query}}
\label{alg4}
\KwIn{
    i) a dynamic graph $G_{t-1}$ at previous timestamp $(t-1)$;
    ii) a graph update operation $\Delta G_t$;
    iii) a registered query graph $q$, and;
    iv) matching result set $A(q,t-1)$ at timestamp $(t-1)$
}
\KwOut{
    an update matching result set, $A(q,t)$, at timestamp $t$
}

    \tcp{edge insertion operation $\Delta G_t = (e+,t)$}
    \eIf{$\Delta G_t = (e+,t)$ for edge $e = (v_i,v_j)$}{
        $\Delta A(q,t)\leftarrow \emptyset$\\
        \For{each query edge $(q_i,q_j)\in E(q)$ with labels matching with that of $(v_i, v_j)$}{
            \tcp{check whether $(v_i,v_j)$ matches $(q_i,q_j)$}
            \If{$o'(q_i)\leq o'(v_i)$ and $o'(q_j)\leq o'(v_j)$}{
                \If{$o'(q_i)\in v_i.MBR_{deg(q_i)}$ and $o'(q_j)\in v_j.MBR_{deg(q_j)}$}{
                    \quad \tcp{Lemmas~\ref{lemma:dominance_pruning} and ~\ref{lemma:mbr_range_pruning}}
                    generate a query plan $Q'$ by using $q_i$ and $q_j$ as the first two query vertices and the remaining ones from $Q$\\
                    $M'[0]=v_i, M'[1]=v_j$\\
                    
                    invoke {\sf Refinement}$(q, Q', G_t,\Delta A(q,t),M',2)$ to incrementally obtain new matching results in $\Delta A(q,t)$\\
                }
            }
            \nop{
            \tcp{Check whether $(v_i,v_j)$ matches $(q_j,q_i)$}
            \If{$o'(q_j)\leq o'(v_i)$ and $o'(q_i)\leq o'(v_j)$}{
                \If{$o'(q_j)\in v_i.MBR_{deg(q_j)}$ and $o'(q_i)\in v_j.MBR_{deg(q_i)}$}{
                    $\Delta Q\leftarrow Q$\\
                    $\Delta Q[i].cand\_set\leftarrow \{v_j\}$\\
                    $\Delta Q[j].cand\_set\leftarrow \{v_i\}$\\   
                    \quad \tcp{Lemmas~\ref{lemma:dominance_pruning} and ~\ref{lemma:mbr_range_pruning}}
                    invoke {\sf Refinement}$(q, \Delta Q, G_t,\Delta A(q,t),\{\emptyset\},0)$ to obtain incremental matching results in $\Delta A(q,t)$\\
                }
            }
            }
        }
        $A(q,t)\leftarrow A(q,t-1) \cup \Delta A(q,t)$\\
    }    
    {
        \tcp{edge deletion operation $\Delta G_t =(e-, t)$}
        $A(q,t) = A(q,t-1)$\\
        \For{each subgraph answer $g\in A(q,t-1)$}{
            \If{edge $(v_i,v_j)$ exists in $E(g)$}{
                $A(q,t)\leftarrow A(q,t)-\{g\}$\\
            }
        }   
    }
\Return $A(q,t)$\\
\end{algorithm}

\subsection{CSM Query Answering}
\label{subsec:csm_query_answering}
Given a dynamic graph $G_{t-1}$ at timestamp $(t-1)$, a registered query graph $q$, and a graph update operation $\Delta G_t$ at timestamp $t$, 
Algorithm~\ref{alg4} provides a CSM query answering algorithm, which maintains/obtains the latest subgraph matching answer set, $A(q, t)$.
Specifically, if the graph update operation $\Delta G_t$ is an insertion of edge $e= (v_i, v_j)$, we will calculate an incremental update, $\Delta A(q, t)$, of new subgraph matching results (lines 1-11).

\noindent {\bf Edge Insertion:} 
For the insertion of an edge $e= (v_i, v_j)$, we will first find matching edge(s) $(q_i, q_j)$ in the query graph with the same labels (line 3). Then, we will check whether or not edges $(v_i, v_j)$ and $(q_i, q_j)$ match each other, by applying the {\it embedding dominance} and {\it MBR range pruning} strategies (described in Section~\ref{subsec:pruning}; lines 4-5). If the answer is yes, then we will generate a new query plan $Q'$ starting from query vertices $q_i$ and $q_j$ (note: the remaining ones are obtained from $Q$, similar to line 10 of Algorithm \ref{alg2}; line 7), initialize the first two matching vertices $v_i$ and $v_j$ in a sorted list $M'$ (lines 8-9), and invoke the function {\sf Refinement($\cdot$)} with parameters of new query plan $Q'$, incremental answer set $\Delta A(q,t)$, initialized sorted list $M'$, and recursion depth $n=2$ (as two matching vertices have been found; line 10). After that, we can update the answer set $A(q,t)$ with $A(q,t-1) \cup \Delta A(q,t)$ (line 11).

\noindent {\bf Edge Deletion:} When $\Delta G_t$ is an edge deletion operation $(e-, t)$, we can simply remove those existing subgraph matching answers $g \in A(q, t-1)$ which contain the deleted edge $e$, and obtain the updated answer set $A(q, t)$ (lines 12-16). Finally, we return the CSM query answer set $A(q, t)$ at current timestamp $t$ (line 17).

\noindent{\bf Complexity Analysis:} 
In Algorithm~\ref{alg4}, for the edge insertion operation, we check all query edges, generate a new query plan $Q'$, and refine new candidate subgraphs with the new edge (lines 1-11). Therefore, the worst-case time complexity is given by $O((|V(q)|^2+\prod_{i=2}^{|V(q)|-1}|Q'[i].cand\_set|)\cdot |E(q)|)$.

For the edge deletion operation (lines 12-16), since we delete the matching results from $A(q,t-1)$ that contain the deleted edge, by using the hash file to check the edge existence with $O(1)$ cost, the time complexity is given by $O(|A(q,t-1)|)$.

\noindent{\bf Discussions on Batch Updates:}
Similar to existing works on CSM
\cite{choudhury2015selectivity,kankanamge2017graphflow,idris2017dynamic,idris2020general,kim2018turboflux,min2021symmetric,sun2022depth}, we adopt dynamic graph model (Definition \ref{def:dynamic_graph}) with one single edge update per timestamp. Thus, as a straightforward method, batch updates can be supported via sequential execution of atomic edge updates.

To enhance the efficiency of batch processing multiple graph edge updates at a time, we can group similar embeddings of query edges (from registered query graphs), as well as those of data edges affected by edge updates from the batch, into MBRs. Then, we can obtain candidate pairs of query/data MBRs, in which we filter/refine candidate query/data vertices. Specifically, to enable batch processing, in Algorithm \ref{alg4}, we can add an intersection checking between edge label MBRs for query/data edges (line 3). Then, we check the dominance between the minimum and maximum corner points of the query and data edge MBRs, respectively, and if the dominance holds, we consider candidate pairs of query/data vertices from MBR pairs (line 4). This way, we can save the cost of some unnecessary pairwise dominance checks for batch updates. 
We would like to leave the interesting topics of further optimizations for batch update as our future work.

\noindent{\bf Discussions on Dynamic Vertex Label Change:}
In this paper, our proposed DIVINE framework can naturally support dynamic vertex label changes by updating the embeddings of the vertex and its 1-hop neighbors. Specifically, for an affected vertex $v_i$, we only need to (1) update the SPUR vector of this vertex $v_i$, using seeded pseudo-random values derived from vertex labels, and; (2) update the SPAN vectors of 1-hop neighbors of vertex $v_i$, using $v_i$'s updated SPUR vector. Then, new embedding vectors of the affected vertices (i.e., $v_i$ and its 1-hop neighbors) will be updated in DAS$^3$.

\noindent{\bf Discussions on Graphs with Directed, Labeled and/or Loop Edges:}
In this paper, our DIVINE approach focuses on a dynamic, undirected, vertex-label graph without loops. To consider labeled edges, we may construct the SPAN vector $y_i$ for each vertex $v_i$, by summing up SPUR vectors $x_j$ of $v_i$'s 1-hop neighbors $v_j$, as well as randomized vectors of edges $e_{i,j}$ (with their edge labels as seeds). To support CSM over graphs with directed and/or loop edges, we can first use DIVINE to compute candidate subgraphs by ignoring edge directions, and then apply a post-processing refinement phase that validates edge direction and/or loop constraints. We would like to leave the topics of generalizing our DIVINE method to graphs with directed, labeled, and/or loop edges as our future work.

\section{Cost-Model-Guided Embedding Optimization}
\label{sec:cost_model}

\subsection{Cost Model for CSM via Vertex Embeddings}

In this subsection, we will propose a cost model to evaluate the performance of our continuous subgraph matching (or the pruning power of using vertex dominance embeddings).

Specifically, we estimate the number, $Cost_{CSM}$, of candidate vertices $v_i$ (to retrieve and refine) whose embedding vectors $o'(v_i)$ are dominated by query embedding vector $o'(q_i)$ (given by Lemma \ref{lemma:dominance_pruning}): 

\begin{equation}
    Cost_{CSM}=\sum_{v_i\in V(G_t)}\prod_{j=1}^{2d}Pr\{o'(q_i)[j]\leq o'(v_i)[j]\},
\label{eq:cost_model}
\end{equation}
where $2d$ is the dimension of the embedding vector $o'(\cdot)$, and $Pr\{\cdot\}$ is a probability function. 

\noindent{\bf Analysis of the Cost Model:} We consider $o'(v_i)[j]$ as a random number generated from a random variable, with mean $\mu_{o'(v_i)[j]}$ and variance $\sigma_{o'(v_i)[j]}^2$. Moreover, $o'(q_i)[j]$ can be considered as a constant. We have the following equation:

\begin{align}
    & \Pr\{o'(q_i)[j]\leq o'(v_i)[j]\} \notag \\
    =\ & \Pr\left\{\frac{(o'(q_i)[j]-o'(v_i)[j])-(o'(q_i)[j]-\mu_{o'(v_i)[j]})}{\sigma_{o'(v_i)[j]}^{2}} \right.\notag\\
    &\left. \quad \leq \frac{-(o'(q_i)[j]-\mu_{o'(v_i)[j]})}{\sigma_{o'(v_i)[j]}^{2}}\right\}.
\label{eq:cdf_function}
\end{align}

By applying \textit{Central Limit Theorem} (CLT) \cite{CLT23} to Eq.~(\ref{eq:cdf_function}), we have:

\begin{eqnarray}
    Pr\{o'(q_i)[j]\leq o'(v_i)[j]\}\approx \Phi\left(\frac{\mu_{o'(v_i)[j]}-o'(q_i)[j]}{\sigma_{o'(v_i)[j]}^{2}}\right),
\label{eq:cdf_function2}
\end{eqnarray}

where $\Phi(\cdot)$ is the \textit{cumulative density function} (cdf) of \textit{standard normal distribution}.

We substitute Eq.~(\ref{eq:cdf_function2}) into Eq.~(\ref{eq:cost_model}) and obtain:

\begin{equation}
    Cost_{CSM}=|V(G_t)| \cdot \prod_{j=1}^{2d}\Phi\left(\frac{\mu_{o'(v_i)[j]}-o'(q_i)[j]}{\sigma_{o'(v_i)[j]}^{2}}\right).
\label{eq:simplified_cost}
\end{equation}

\subsection{Cost-Model-Based Vertex Embedding Design}
Eq.~(\ref{eq:simplified_cost}) estimates the (worst-case) query cost for one query vertex $q_i$ during the continuous subgraph matching. Intuitively, our goal is to design vertex embeddings $o'(\cdot)$ that minimizes the cost model $Cost_{CSM}$ (given in Eq.~(\ref{eq:simplified_cost})). 

Note that, $\Phi(\cdot)$ in Eq.~(\ref{eq:simplified_cost}) is a monotonically increasing function. Therefore, if we can minimize the term $\frac{\mu_{o'(v_i)[j]}-o'(q_i)[j]}{\sigma_{o'(v_i)[j]}^{2}}$ in Eq.~(\ref{eq:simplified_cost}), then low query cost $Cost_{CSM}$ can be achieved. In other words, guided by our proposed cost model (to minimize the query cost $Cost_{CSM}$), {\bf our target is to design/select a ``good'' distribution of vertex dominance embeddings $o'(v_i)$ with:

\begin{enumerate}
    \item low mean $\mu_{o'(v_i)[j]}$ , and;
    \item high variance $\sigma_{o'(v_i)[j]}^2$.
\end{enumerate}
}

In this paper, unlike standard vertex dominance embeddings generated by \textit{Uniform} random function $f_1(\cdot)$ (as discussed in Section \ref{subsec:embedding}), we choose to use a (seeded) \textit{Zipf} random function to produce SPUR vectors $x_i$ in vertex embeddings $o'(v_i)$ (or in turn generate SPAN vectors $y_i = \sum x_j$). The $Zipf$ distribution exactly follows our target of finding a random variable with low mean and high variance, which can achieve low query cost, as guided by our proposed cost model $Cost_{CSM}$ (given in Eq.~(\ref{eq:cost_model})). Please refer to experimental comparison of $Zipf$ with $Uniform$ in Figure~\ref{subfig:emb_syn} of Section \ref{sec:experiments}. We would like to leave the interesting topic of studying other low-mean/high-variance distributions as our future work.

\noindent {\bf Discussions on the Seeded Zipf Generator:} We consider two distributions, $Uniform$ and $Zipf$, each of which is divided to $b$ buckets with the same area. This way, we create a 1-to-1 mapping between buckets in a Uniform distribution and that in a $Zipf$ distribution. Given a seeded pseudo-random number, $r$, from the Uniform distribution, we can first find the bucket this Uniform random number falls into, obtain its corresponding bucket in the $Zipf$ distribution, and compute its proportional location $r'$ in the $Zipf$ bucket. As a result, $r'$ is the random number that follows the $Zipf$ distribution.

\noindent {\bf Integration of the Cost-Model-Based Vertex Embeddings into our Dynamic Subgraph Matching Framework:} 
In light of our cost model above, we design a novel cost-model-based vertex dominance embedding, denoted as $o'_{C}(v_i)$, for each vertex $v_i$, that is:

\begin{equation}
    o'_{C}(v_i)=\alpha(x'_i||y'_i)+\beta z_i,
    \label{eq:cost_model_embedding}
\end{equation}

where $x'_i=f_Z(l(v_i))$ is the newly designed SPUR vector generated by a seeded $Zipf$ random generator $f_Z(\cdot)$, and $y_i' = \sum_{\forall v_j \in \mathcal{N}_{v_i}} x_j'$.

Note that, since the SPAN vector $y'_i$ in Eq.~(\ref{eq:cost_model_embedding}) is given by $\sum_{\forall v_j \in \mathcal{N}_{v_i}} x_j'$, its distribution still follows the property of low mean and high variance to achieve low query cost.

\begin{table}[t]\small
\begin{center}
\caption{Statistics of real-world graph data sets.}
\label{tab:datasets}
\begin{tabular}{|l||c|c|c|c|}
\hline
\textbf{\text{ }\text{ }Data Sets}&\textbf{$|V(G)|$}&\textbf{$|E(G)|$}&\textbf{$|\sum|$}&\textbf{$avg\_deg(G)$} \\
\hline\hline
    Yeast (ye) & 3,112 & 12,519 & 71 & 8.0\\\hline
    HPRD (hp) & 9,460 & 34,998 & 307 & 7.4\\\hline
    Reddit (re) & 232,965 & 57,307,946 & 41 & 491.98 \\\hline
    DBLP (db) & 317,080 & 1,049,866 & 15 & 6.6\\\hline
    Youtube (yt) & 1,134,890 & 2,987,624 & 25 & 5.3\\\hline
    US Patents (up) & 3,774,768 & 16,518,947 & 20 & 8.8\\\hline
\end{tabular}
\end{center}
\end{table}

\section{Experimental Evaluation}
\label{sec:experiments}
\subsection{Experimental Settings}
To evaluate the performance of our DIVINE approach, we conduct experiments on a Ubuntu server equipped with an Intel Core i9-12900K CPU and 128GB memory. Our source code in C++ and real/synthetic graph data sets are available at URL: {\it \url{https://github.com/JamesWhiteSnow/DSM}}.

\noindent{\bf Real/Synthetic Graph Data Sets:} We evaluate our DIVINE approach over both real and synthetic graphs.

\underline{\textit{Real-world graphs.}} We test six real-world graph data used by previous works \cite{he2008graphs, shang2008taming, zhao2010graph, lee2012depth, sun2012efficient, han2013turboiso, ren2015exploiting, bi2016efficient, katsarou2017subgraph, hamilton2017inductive, bhattarai2019ceci, han2019efficient, sun2020memory}, which can be classified into three categories: i) biology networks (Yeast and HPRD); ii) bibliographical/social networks (Reddit, DBLP and Youtube), and; iii) citation networks (US Patents). 
Statistics of these real graphs are summarized in Table~\ref{tab:datasets}.

\underline{\textit{Synthetic graphs.}} We generate synthetic graphs via NetworkX \cite{hagberg2020networkx}, and produce small-world graphs following the Newman-Watts-Strogatz model \cite{watts1998collective}. 
Parameter settings of synthetic graphs are depicted in Table~\ref{tab:parameters}. For each vertex $v_i$, we generate its label $l(v_i)$ by randomly picking up an integer within $[1, |\sum|]$, following the $Uniform$, $Gaussian$, or $Zipf$ distribution. Accordingly, we obtain 3 types of graphs, denoted as $Syn\text{-}Uni$, $Syn\text{-}Gau$, and $Syn\text{-}Zipf$, respectively.

The insertion (or deletion) rate is defined as the ratio of the number of edge insertions (or deletions) to the total number of edges in the raw graph data.
Following the literature \cite{choudhury2015selectivity,kim2018turboflux,min2021symmetric,sun2022depth}, we set the insertion rate to 10\% by default (i.e., initial graph $G_0$ contains 90\% edges, and the remaining 10\% edges arrive as the insertion stream).
Similarly, for the case of edge deletions, we set the default deletion rate to 10\%  (i.e., initial graph $G_0$ contains all edges, and deletion operations of 10\% edges arrive in the stream). 
Below, we will report the CSM performance with edge insertion only, unless specified otherwise (see the edge-deletion case in Figure~\ref{subfig:dsm_del_real_mix}).

\noindent{\bf Query Graphs:} Similar to previous works \cite{sun2012efficient, han2013turboiso, ren2015exploiting, bi2016efficient, katsarou2017subgraph, archibald2019sequential, bhattarai2019ceci, han2019efficient,sun2022depth}, for each graph $G_D$, we randomly extract/sample 100 connected subgraphs as query graphs, where parameters of query graphs $q$ (e.g., $|V(q)|$ and $avg\_deg(q)$) are depicted in Table~\ref{tab:parameters}. 

\noindent{\bf Baseline Methods:} 
We compare our DIVINE approaches (using the cost-model-based vertex dominance embeddings $o'_C(v_i)$, as discussed in Section \ref{sec:cost_model}) with five representative baseline methods of dynamic subgraph matching and one GNN-based path embedding baseline of static subgraph matching as follows: 
\begin{enumerate}
    \item {\bf Graphflow (GF)} \cite{kankanamge2017graphflow} is a direct-incremental algorithm that enumerates updates of matching results without any auxiliary data structure. 

    \item  {\bf SJ-Tree (SJ)} \cite{choudhury2015selectivity} is an index-based incremental method that evaluates the join query with binary joins using the index.

    \item {\bf TurboFlux (TF)} \cite{kim2018turboflux} is an index-based incremental method that stores matches of paths in $q$ without materialization and evaluates the query with the vertex-at-a-time method.

    \item {\bf SymBi (Sym)} \cite{min2021symmetric} is an index-based incremental method that prunes candidate vertex sets using all query edges.

    \item {\bf IEDyn (IED)} \cite{idris2017dynamic,idris2020general} is an index-based incremental method that supports acyclic queries and can achieve constant delay enumeration under the setting of graph homomorphism.

    \item {\bf GNN-PE} \cite{ye2024efficient} retrains the GNN model for path embedding if the dominance relationship does not hold (due to dynamic graph updates).
\end{enumerate}

We used the code of baseline methods from \cite{sun2022depth,ye2024efficient}, where GNN-PE is implemented by Python and C++, and the other methods are implemented in C++ for a fair comparison.

\begin{table}[t]\small
\setlength{\tabcolsep}{3pt} 
\begin{center}
\caption{Parameter settings.}
\label{tab:parameters}
\begin{tabular}{|p{8cm}||p{6cm}|}
\hline
\textbf{Parameters}&\textbf{Values}\\
\hline\hline
    the dimension, $d$, of the SPUR/SPAN vector  & 1 {\bf 2}, 3, 4\\\hline
    the ratio, $\beta/\alpha$ & 10, 100, {\bf 1,000}, 10,000, 100,000\\\hline
    the number, $m$, of degree groups & 1, 2, {\bf 3}, 4, 5\\\hline 
    the number, $K$, of cell intervals in each dimension of $Syn_j$ & 1, 2, {\bf 5}, 8, 10\\\hline
    the number, $|\Sigma|$, of distinct labels & 5, 10, {\bf 15}, 20, 25\\\hline
    the average degree, $avg\_deg(q)$, of the query graph $q$ & 2, {\bf 3}, 4\\\hline  
    the size, $|V(q)|$, of the query graph $q$ & 5, 6, {\bf 8}, 10, 12\\\hline
    the average degree, $avg\_deg(G_D)$, of dynamic graph $G_D$ & 3, 4, {\bf 5}, 6, 7\\\hline     
    the size, $|V(G_D)|$, of the dynamic data graph $G_D$ & 10K, 30K, {\bf 50K}, 80K, 100K, 500K, 1M, 5M, 10M\\\hline 
\end{tabular}
\end{center}
\end{table}

\noindent{\bf Evaluation Metrics:}
In our experiments, we report the efficiency of our DIVINE approach and baseline methods, in terms of the \textit{wall clock time}, which is the time cost of processing the CSM query over the entire graph update sequence (i.e., continuous subgraph matching answer set monitoring).
In particular, this time cost includes the time of filtering, refinement, embedding updates, and synopsis updates for each CSM query.
We also evaluate the \textit{pruning power} of our \textit{embedding dominance pruning} and \textit{MBR range pruning strategies} (as mentioned in Section \ref{subsec:pruning}), which is defined as the percentage of vertices that can be ruled out by our pruning methods. For all experiments, we take an average of each metric over 100 runs (with 100 query graphs, respectively). 

Table~\ref{tab:parameters} depicts our parameter settings, where default parameter values are in bold. For each set of experiments, we vary the value of one parameter while setting other parameters to default values. 

\subsection{Parameter Tuning}
We first tune the parameters for our DIVINE approach using synthetic graph data sets.

\begin{figure}[t]
\centering
\subfigure[][{DIVINE vs. $d$}]{                    
\scalebox{0.15}[0.15]{\includegraphics{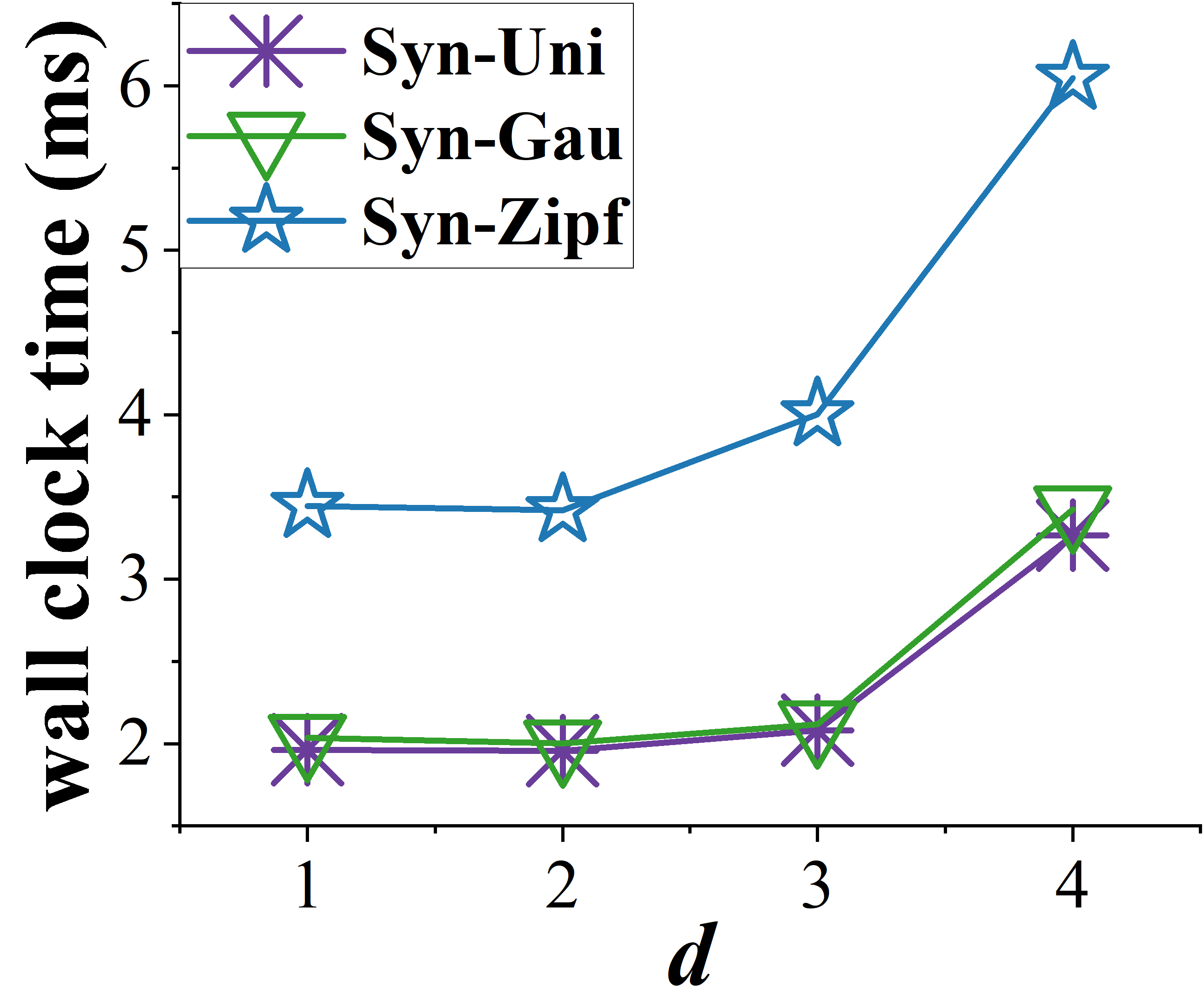}}\label{subfig:dim_syn_mix}}
\qquad
\subfigure[][{DIVINE vs. $\beta/\alpha$}]{              
\scalebox{0.15}[0.15]{\includegraphics{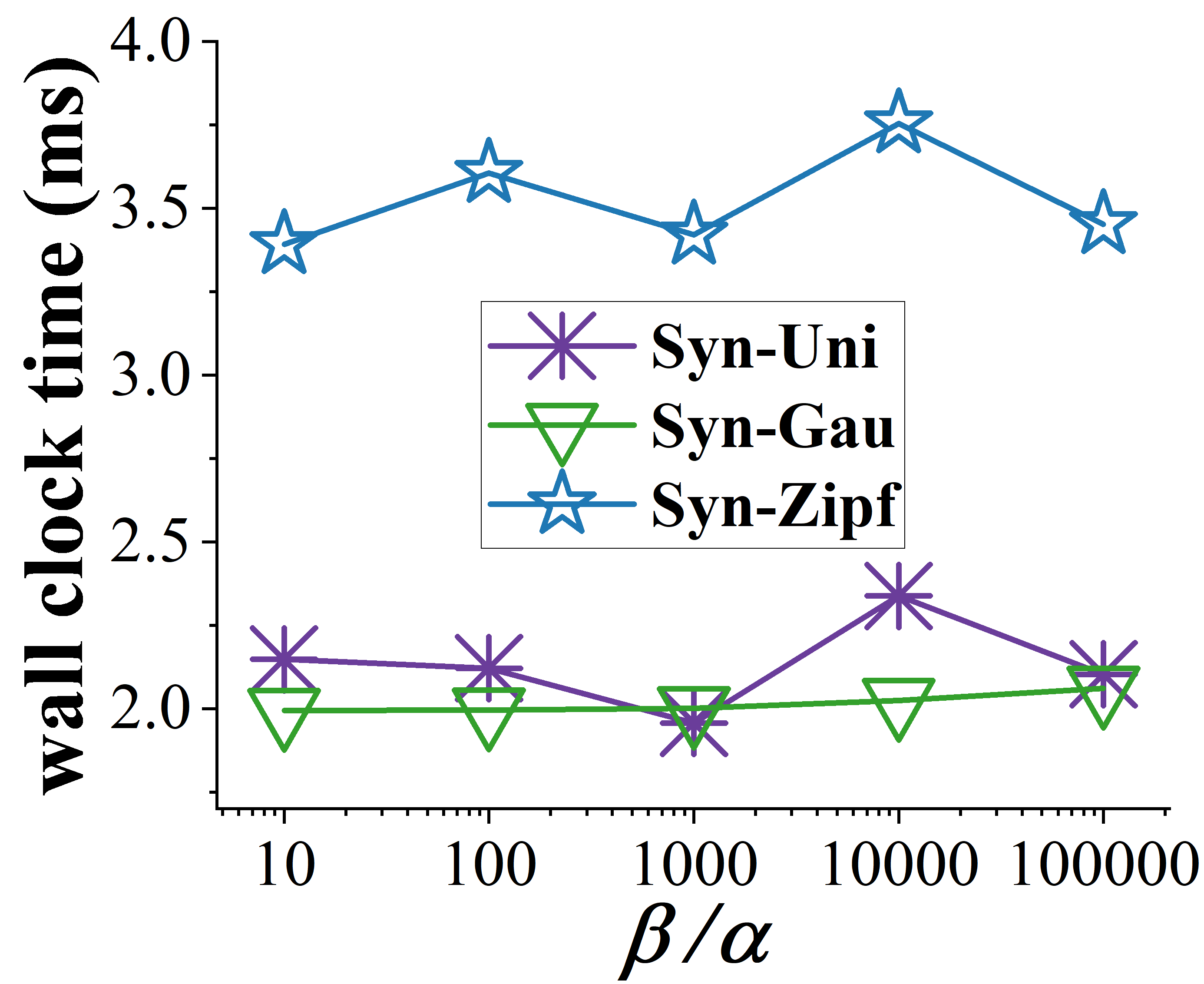}}\label{subfig:ratio_syn}}
\qquad
\subfigure[][{DIVINE vs. $K$}]{                    
\scalebox{0.15}[0.15]{\includegraphics{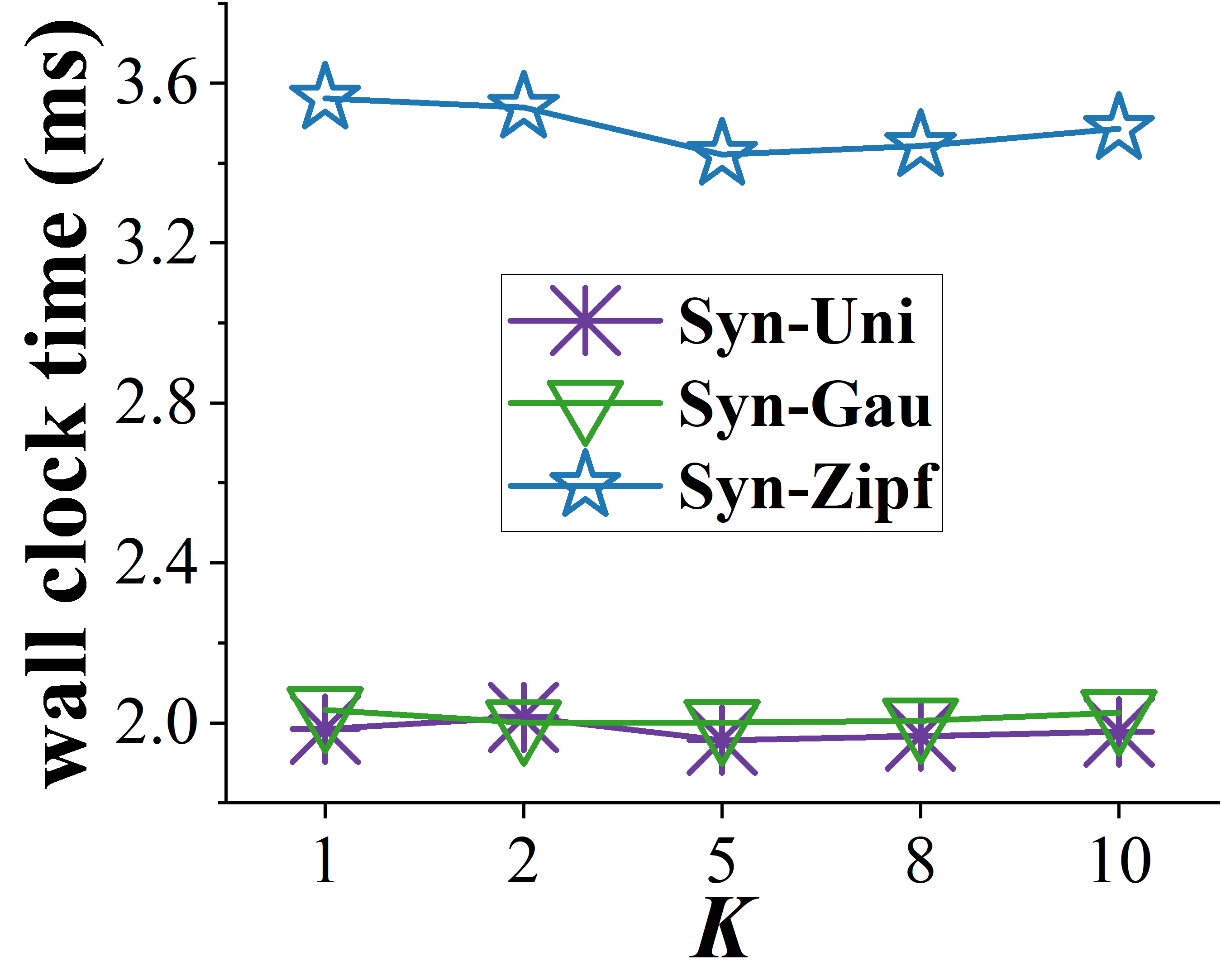}}\label{subfig:cell_syn}}
\qquad
\subfigure[][{DIVINE vs. $m$}]{                    
\scalebox{0.15}[0.15]{\includegraphics{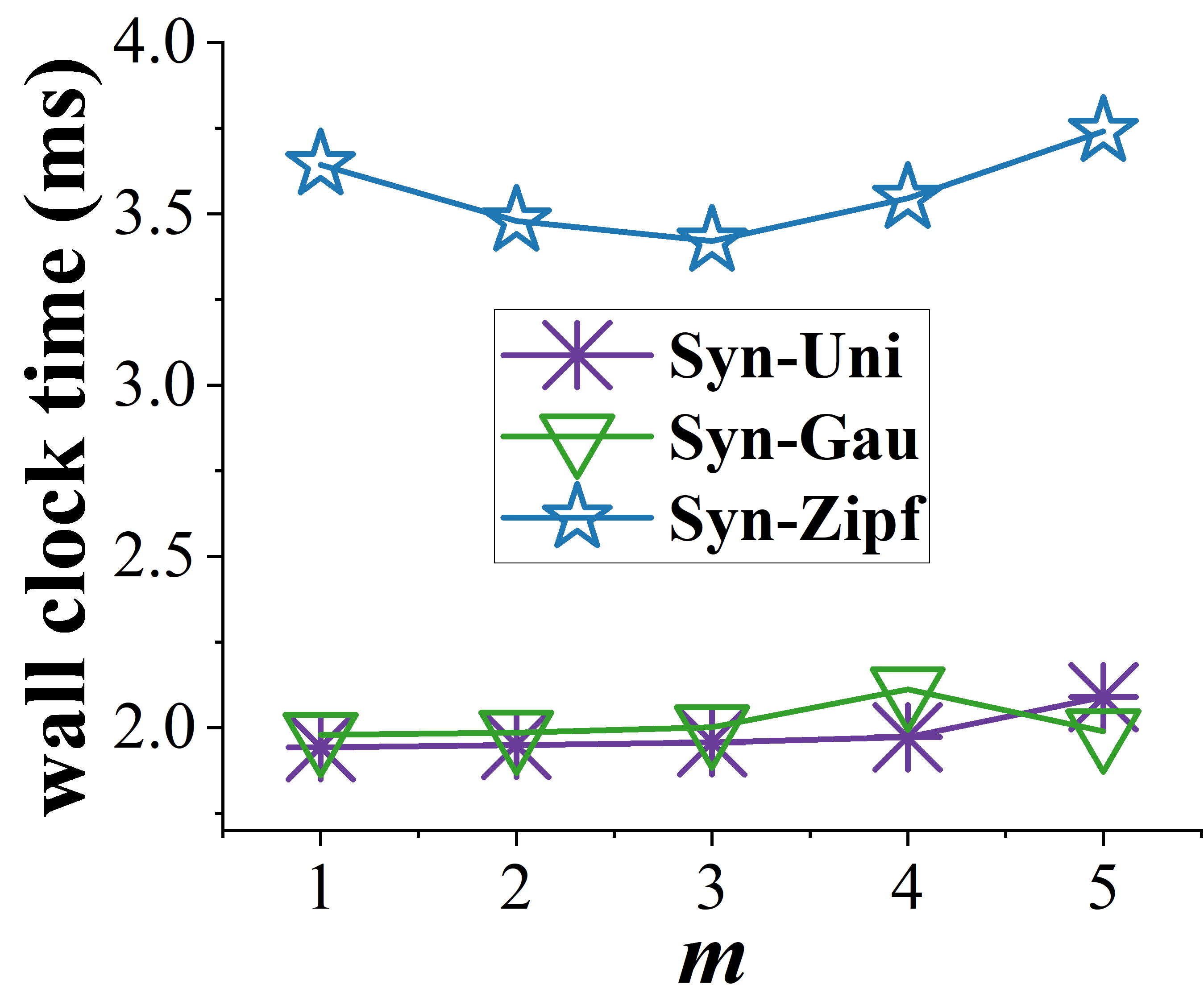}}\label{subfig:group_syn}}
\qquad
\subfigure[][{DIVINE vs. embeddings}]{
\scalebox{0.15}[0.15]{\includegraphics{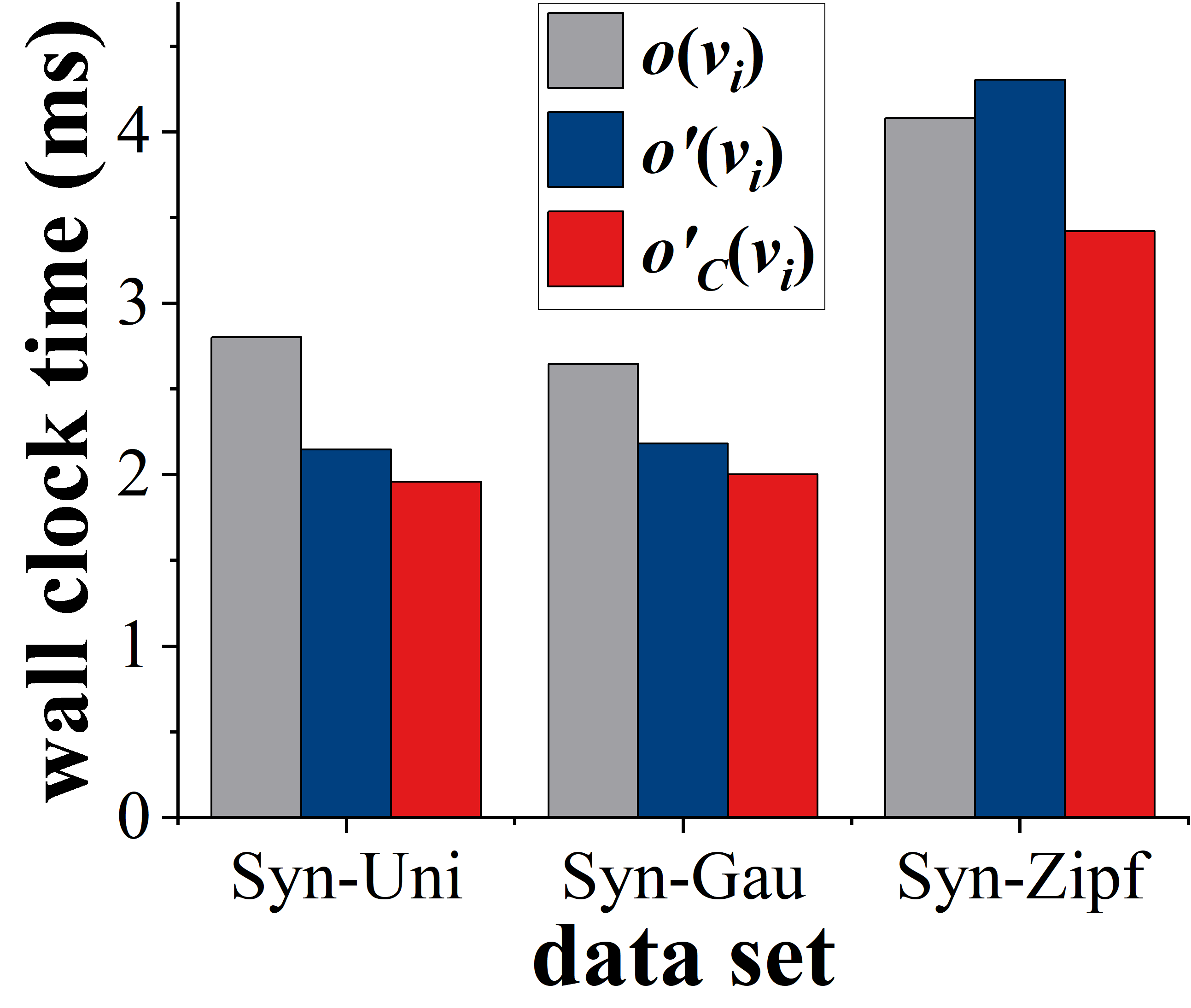}}\label{subfig:emb_syn}}
\caption{The DIVINE efficiency w.r.t parameters $d$, $\beta/\alpha$, $K$, $m$, and embedding design strategies.}
\label{fig:parameter_tuning_app}
\end{figure}

\noindent {\bf The DIVINE Efficiency w.r.t. SPUR/SPAN Vector Dimension $\bm{d}$:}
Figure~\ref{subfig:dim_syn_mix} illustrates the DIVINE performance by varying the dimension, $d$, of the cost-model-based SPUR/SPAN vector (in $o'_C(v_i)$ of Eq.~(\ref{eq:cost_model_embedding})) from 1 to 4, where other parameters are set to their default values. 
With a higher embedding dimension $d$, the pruning power of our proposed pruning strategies in a higher-dimensional space increases. However, the access of synopses with larger $d$ may also incur higher costs due to the "dimensionality curse" \cite{BerchtoldKK96}. Thus, in this figure, for larger $d$, the wall clock time of DIVINE first decreases and then increases over all synthetic graphs. Nonetheless, the time cost remains low (i.e., 1.96 $\sim$ 6.05 $ms$) for different $d$ values.

\noindent {\bf The DIVINE Efficiency w.r.t. $\bm{\beta/\alpha}$ Ratio:}
Figure~\ref{subfig:ratio_syn} varies the ratio, $\beta/\alpha$, from 10 to 100,000 for the optimized (or cost-model-based) vertex dominance embeddings $o'(v_i)$ (or $o'_C(v_i)$), where other parameters are set by default.
In the figure, we can see that our DIVINE approaches are not very sensitive to the ratio $\beta/\alpha$. For different $\beta/\alpha$ ratios, the query cost remains low (i.e., 1.96 $\sim$ 3.75 $ms$).

\noindent {\bf The DIVINE Efficiency w.r.t. $\bm{\#}$ of Cell Intervals on Each Dimension, $\bm{K}$:}
Figure~\ref{subfig:cell_syn} evaluates the effect of the number, $K$, of cell intervals on each dimension on the DIVINE performance, where $K$ varies from 1 to 10, and other parameters are set by default. 
When $K$ becomes larger, more vertices in synopsis cells can be pruned, however, more cells need to be accessed. Therefore, for DIVINE, with the increase of $K$, the time cost first decreases and then increases. Nonetheless, for different $K$ values, the query cost remains low (i.e., 1.96 $\sim$ 3.54 $ms$).

\noindent {\bf The DIVINE Efficiency w.r.t. $\bm{\#}$ of Degree Groups, $\bm{m}$:}
Figure~\ref{subfig:group_syn} reports the performance of our DIVINE approach, by varying the number, $m$, of degree groups from 1 to 5, where other parameters are set by default. 
In this figure, the time costs over $Syn\text{-}Uni$ and $Syn\text{-}Zipf$ first decrease and then increase when $m$ increases, and there are some fluctuations for $Syn\text{-}Gau$ (e.g., for $m=4$ or $5$). 
For all $m$ values, the time cost remains low (i.e., 1.94 $\sim$ 3.74 $ms$).

\noindent {\bf The DIVINE Efficiency w.r.t. Vertex Dominance Embedding Design Strategies:} Figure~\ref{subfig:emb_syn} tests the DIVINE performance with different designs of vertex dominance embeddings, $o(v_i)$ (in Eq.~(\ref{eq:embedding})), $o'(v_i)$ (in Eq.~(\ref{eq:new_vertex_embedding})), and $o'_C(v_i)$  (in Eq.~(\ref{eq:cost_model_embedding})), where default values are used for parameters. 
In the figure, we can see that, the optimized vertex embeddings $o'(v_i)$ (via the base vector $z_i$) incur a smaller time cost than $o(v_i)$ in all cases, whereas the cost-model-based vertex embeddings $o'_C(v_i)$ consistently achieve the lowest time. For different vertex embeddings, the query cost remains low (i.e., 1.96$\sim$4.3$ms$).

The experimental results on real-world graphs are similar and thus omitted here.

In subsequent experiments, we will set parameters $d=2$, $\beta/\alpha=1,000$, $m=3$, and $K=5$, and use the best cost-model-based vertex embeddings $o'_C(v_i)$ (given by Eq.~(\ref{eq:cost_model_embedding})).

\subsection{The DIVINE Effectiveness Evaluation}
In this subsection, we report the pruning power of our pruning strategies (as discussed in Section~\ref{subsec:pruning}), as well as the query selectivity of our proposed embeddings, for continuous subgraph matching over real/synthetic graphs.

\begin{figure}[t]
\centering
\subfigure[][{real-world graphs}]{                    
\scalebox{0.15}[0.15]{\includegraphics{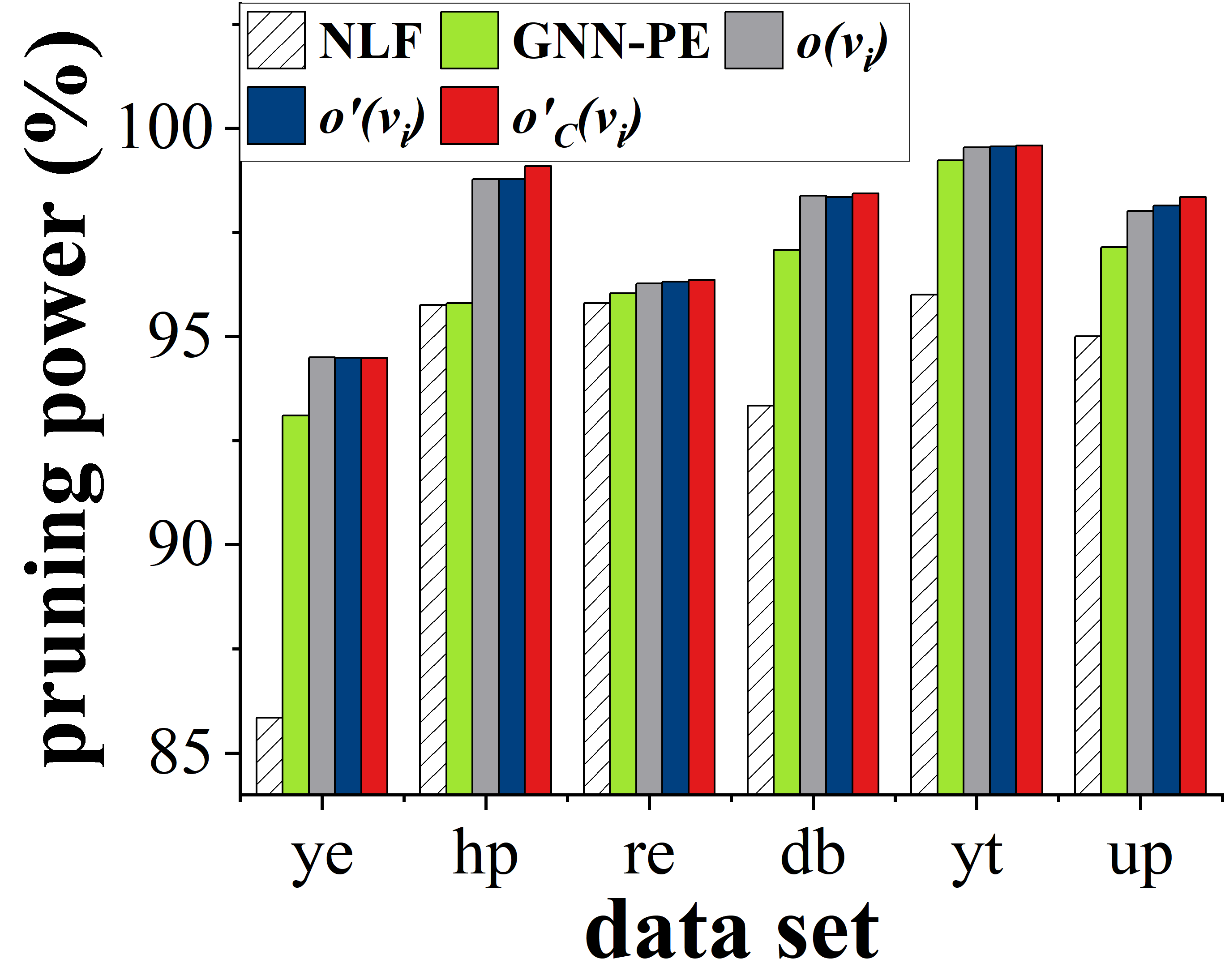}}\label{subfig:prune_real_new}}
\qquad
\subfigure[][{synthetic graphs}]{
\scalebox{0.15}[0.15]{\includegraphics{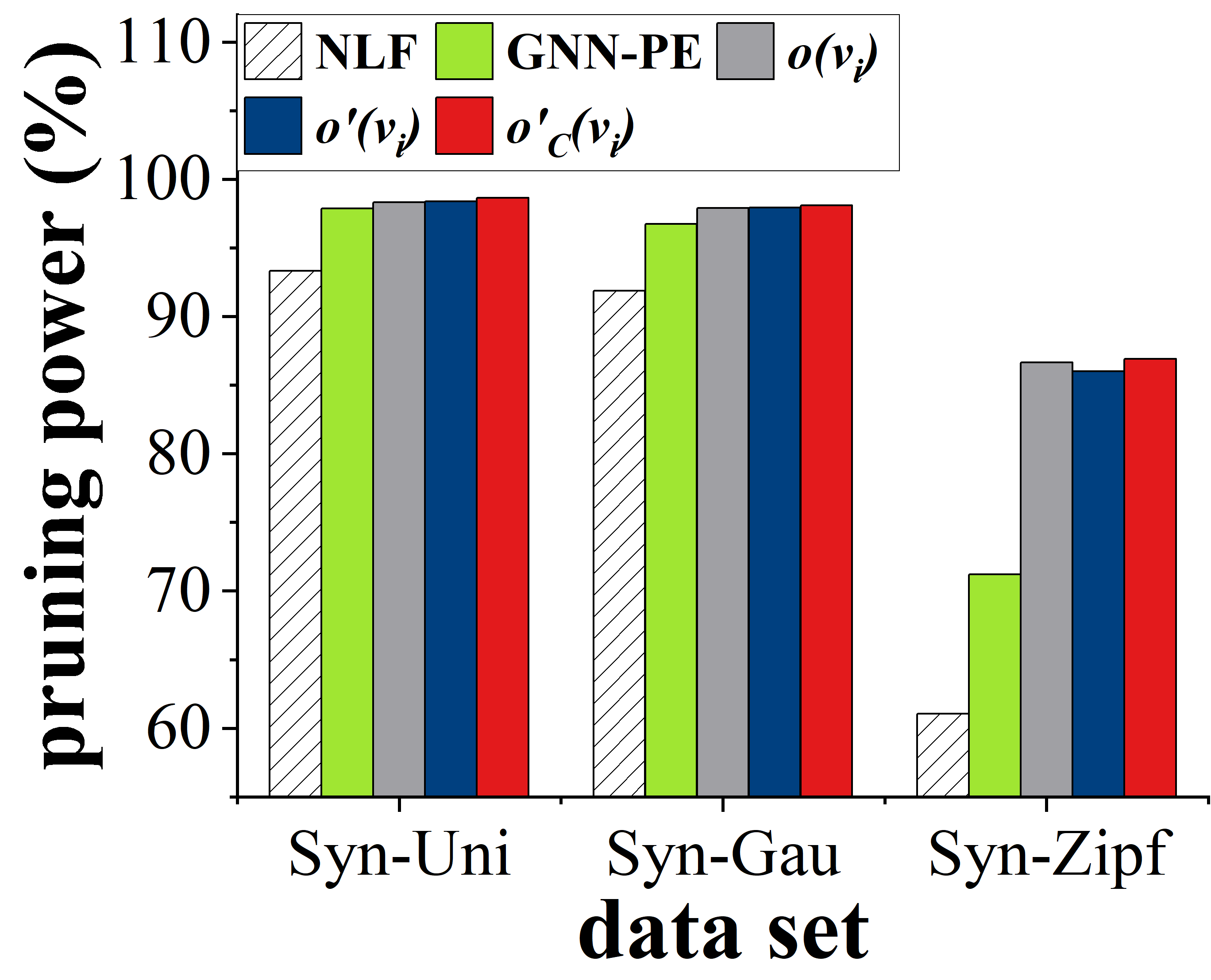}}\label{subfig:prune_syn_new}}
\caption{The DIVINE pruning power w.r.t. different design strategies of vertex dominance embeddings, compared with \textit{naive label filtering} (NLF) and GNN-PE.}
\label{fig:pruning}
\end{figure}

\noindent {\bf The DIVINE Pruning Power:}
Figure~\ref{fig:pruning} shows the pruning power of our DIVINE approach (different embedding designs $o(v_i)$, $o'(v_i)$, and $o'_C(v_i)$) over real/synthetic graphs, compared with that of \textit{naive label filtering} (NLF) and GNN-PE \cite{ye2024efficient}, where all parameters are set by default. Specifically, NLF filters out candidate vertices with labels different from the query vertex, whereas GNN-PE rules out candidate vertices whose learned embedding vectors are not dominated by that of the query vertex. In the subfigures, we can see that our DIVINE approach even with the naive embedding design $o(v_i)$ can significantly outperform both NLF and GNN-PE, over all the real/synthetic graphs. 

Moreover, we can see that the DIVINE pruning power of our proposed vertex embedding designs can reach as high as $94.47\% \sim 99.58\%$ for real-world graphs and $85.99\% \sim 98.65\%$ for synthetic graphs, which confirms the effectiveness of our embedding-based pruning strategies. Our cost-model-based vertex embedding strategy $o'_C(v_i)$ always achieves the highest pruning power (i.e., $86.93\% \sim 99.58\%$). Thus, in the sequel, we will always use the cost-model-based vertex embeddings in our DIVINE approach.

\begin{figure}[t]
\centering
\subfigure[][{real-world graphs}]{                    
\scalebox{0.15}[0.15]{\includegraphics{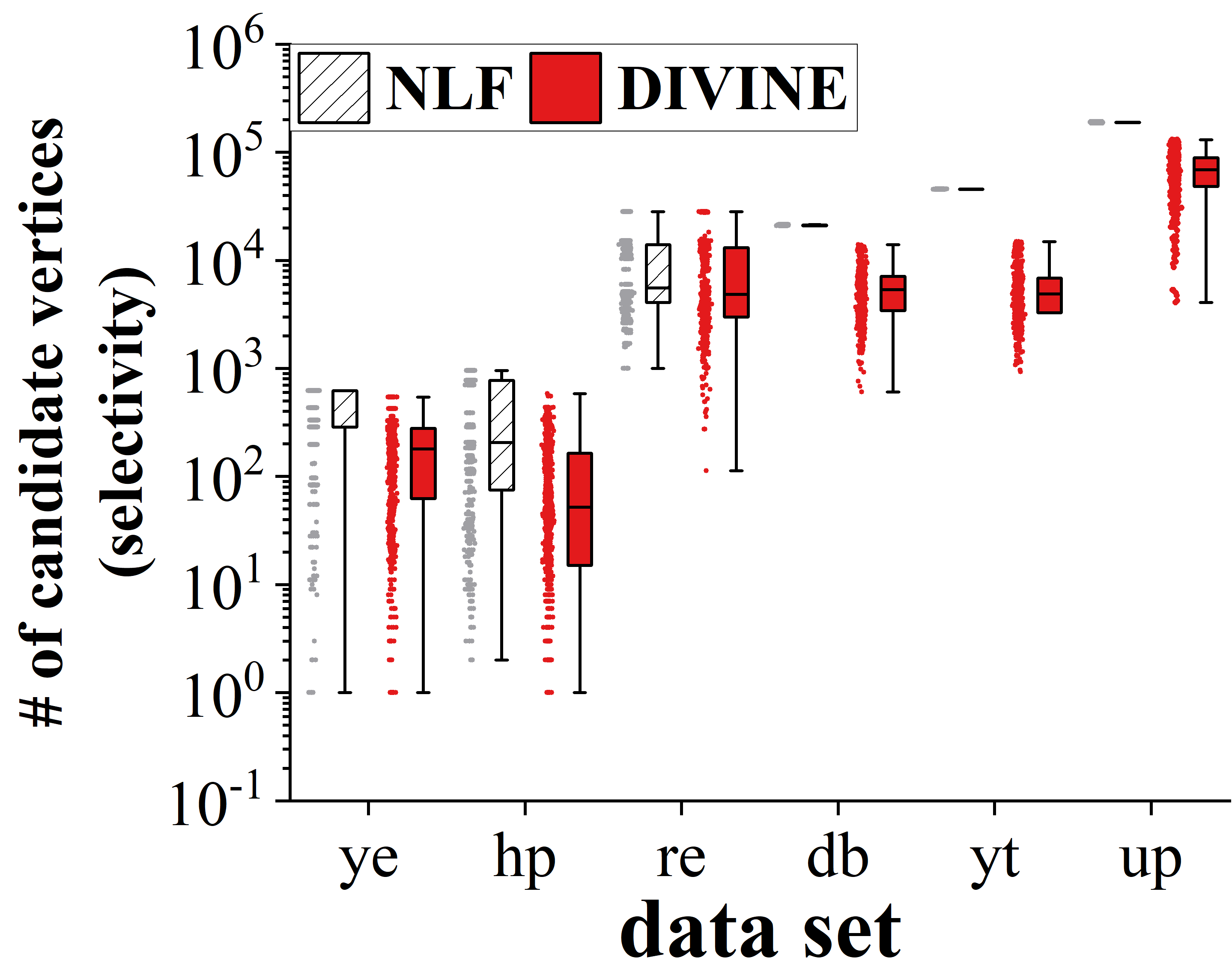}}\label{subfig:selectivity_real}}
\qquad
\subfigure[][{synthetic graphs}]{
\scalebox{0.15}[0.15]{\includegraphics{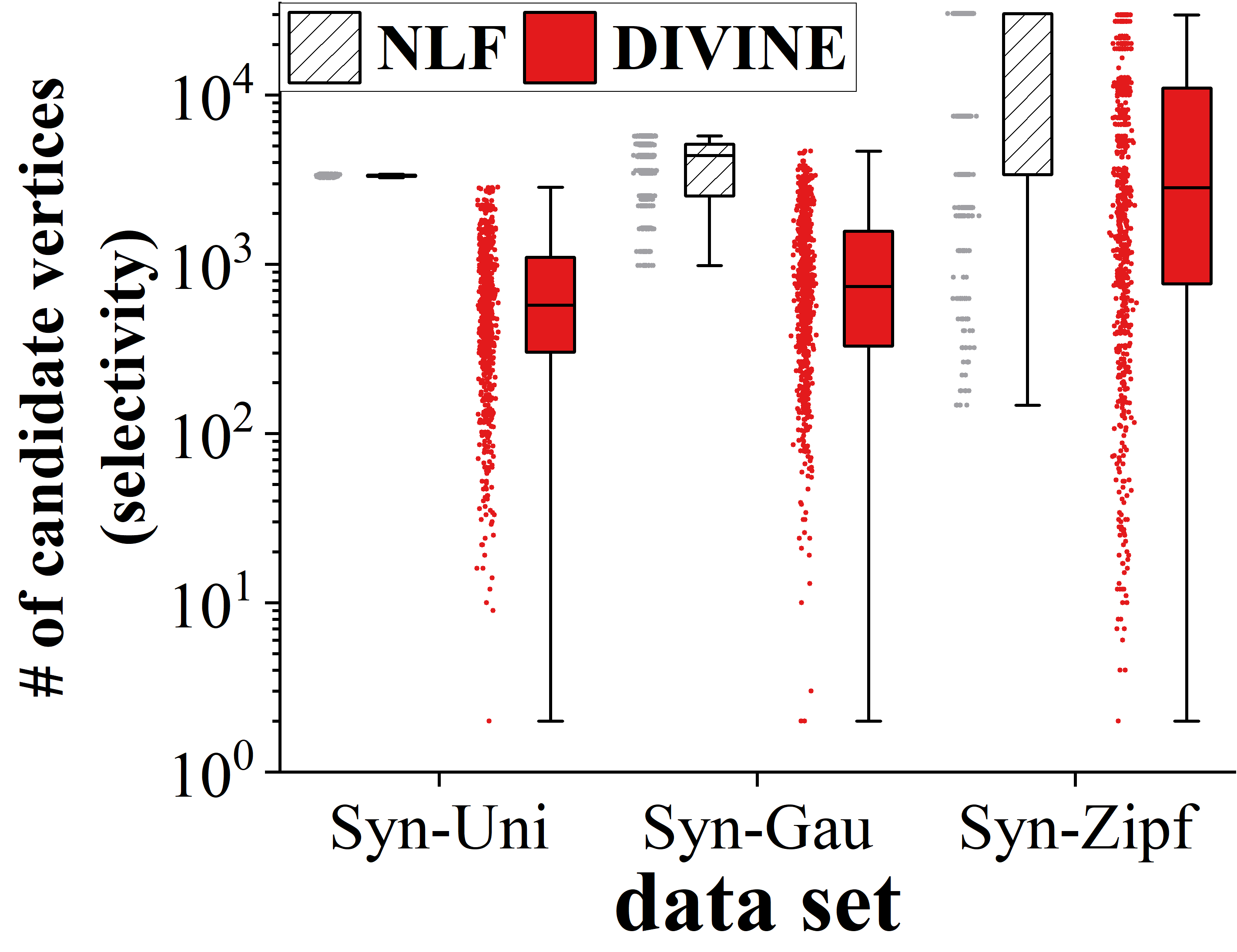}}\label{subfig:selectivity_syn}}
\caption{The comparison of the selectivity distributions between DIVINE and NLF.}
\label{fig:selectivity}
\end{figure}

\noindent {\bf The DIVINE Selectivity Distribution:}
Figure \ref{fig:selectivity} evaluates the query selectivity distribution of our DIVINE approach, compared with that of NFL, where default values are assigned to parameters. Specifically, we first obtain 800 query vertices from 100 query graphs (i.e., each graph with 8 query vertices), and then report the numbers of candidate vertices dominated by their corresponding 800 query embedding vectors (using DIVINE or NLF), respectively, in box plots  (i.e., query selectivity distribution with MIN, MAX, medium, and lower/upper quantiles). From the subfigures, we can see that, compared with NLF, our DIVINE approach has significantly fewer candidate vertices for most query vertices, which indicates better selectivity of our DIVINE embedding design.

\begin{figure}[t]
    \centering
    \includegraphics[width=0.75\textwidth]{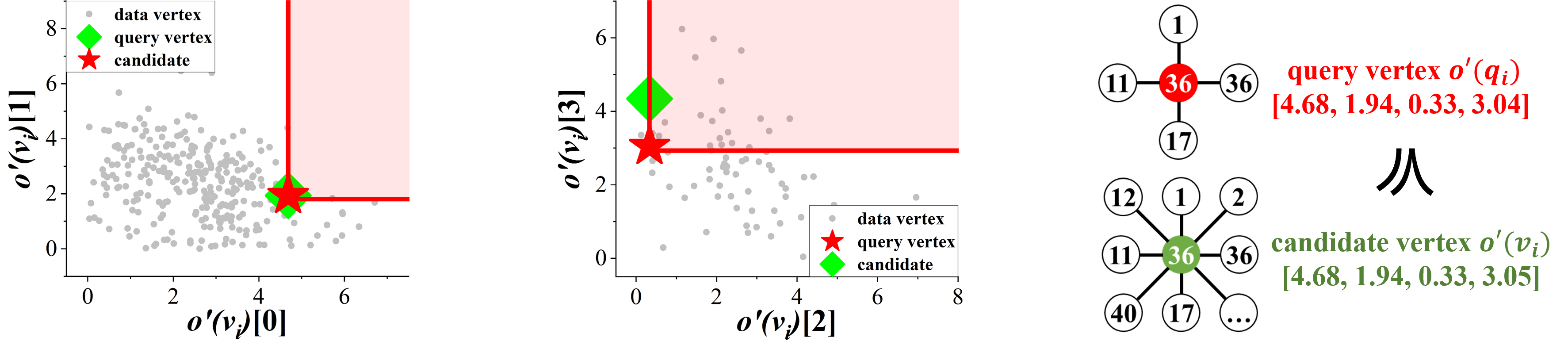}
    \caption{ A visualization of vertex dominance embeddings.}
    \label{fig:emb_visualization}
\end{figure}

\noindent {\bf The Visualization of Our Vertex Dominance Embedding:} 
Figure \ref{fig:emb_visualization} visualizes vertex dominance embeddings over the ye data set (gray points), where we plot 2D SPUR and SPAN vectors in the two leftmost figures. As a case study, given a query vertex $o'(q_i)$ (red star point), its dominating region $DR(o'(q_i))$ contains a candidate vertex $v_i$ (green diamond) that matches with $q_i$. We can see that the subgraph relationship of star structures can be well-preserved if their vertex embeddings follow the dominance relationship.

\subsection{The DIVINE Efficiency Evaluation}
In this subsection, we report the efficiency of our DIVINE approaches over real/synthetic graphs.

\begin{figure}[t]
\centering
\subfigure[][{\footnotesize real-world graphs ($Ins_{10\%}$)}]{                    
\scalebox{0.15}[0.15]{\includegraphics{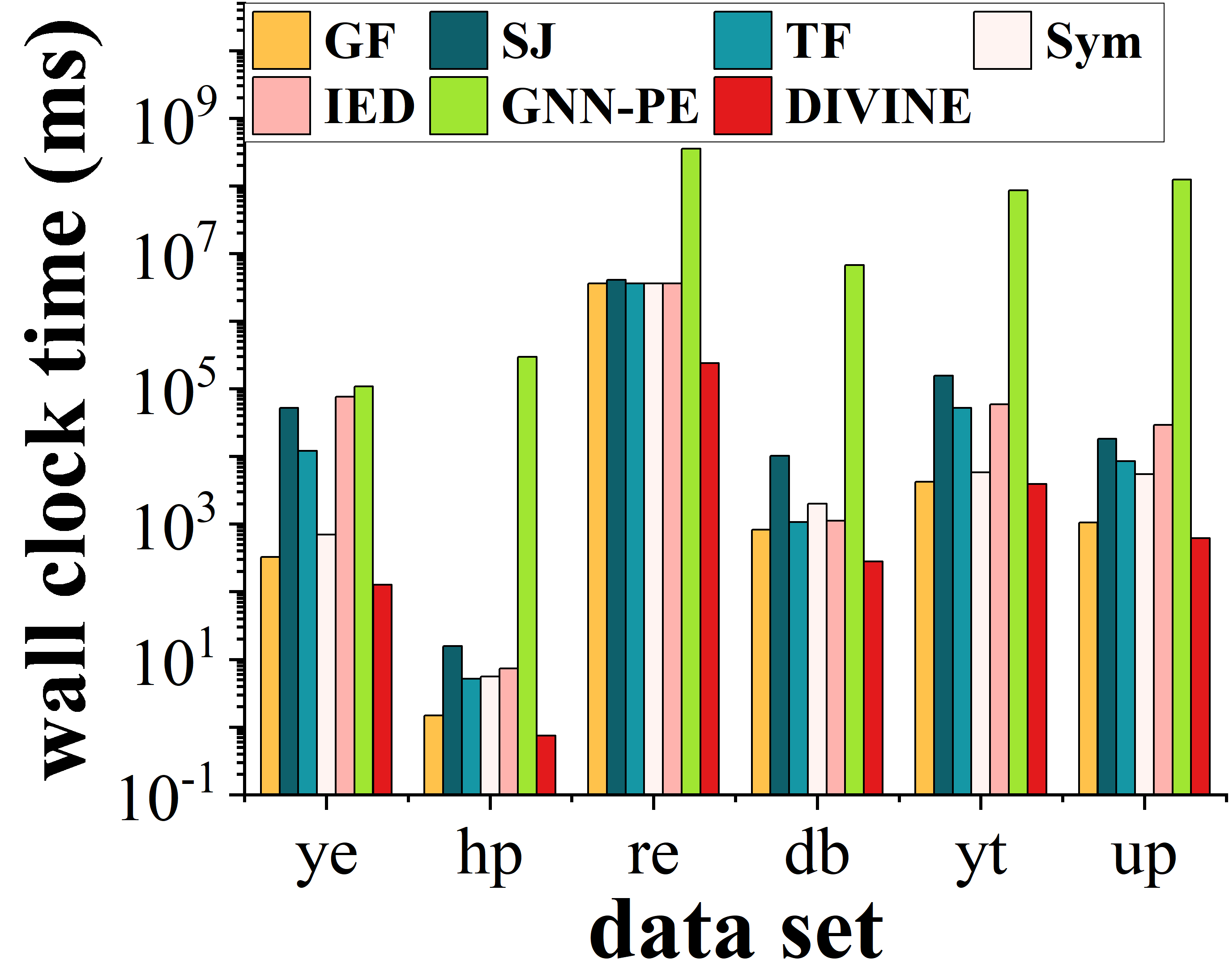}}\label{subfig:dsm_real}}
\qquad
\subfigure[][{\footnotesize real-world graphs ($Ins_{20\%}$)}]{                    
\scalebox{0.15}[0.15]{\includegraphics{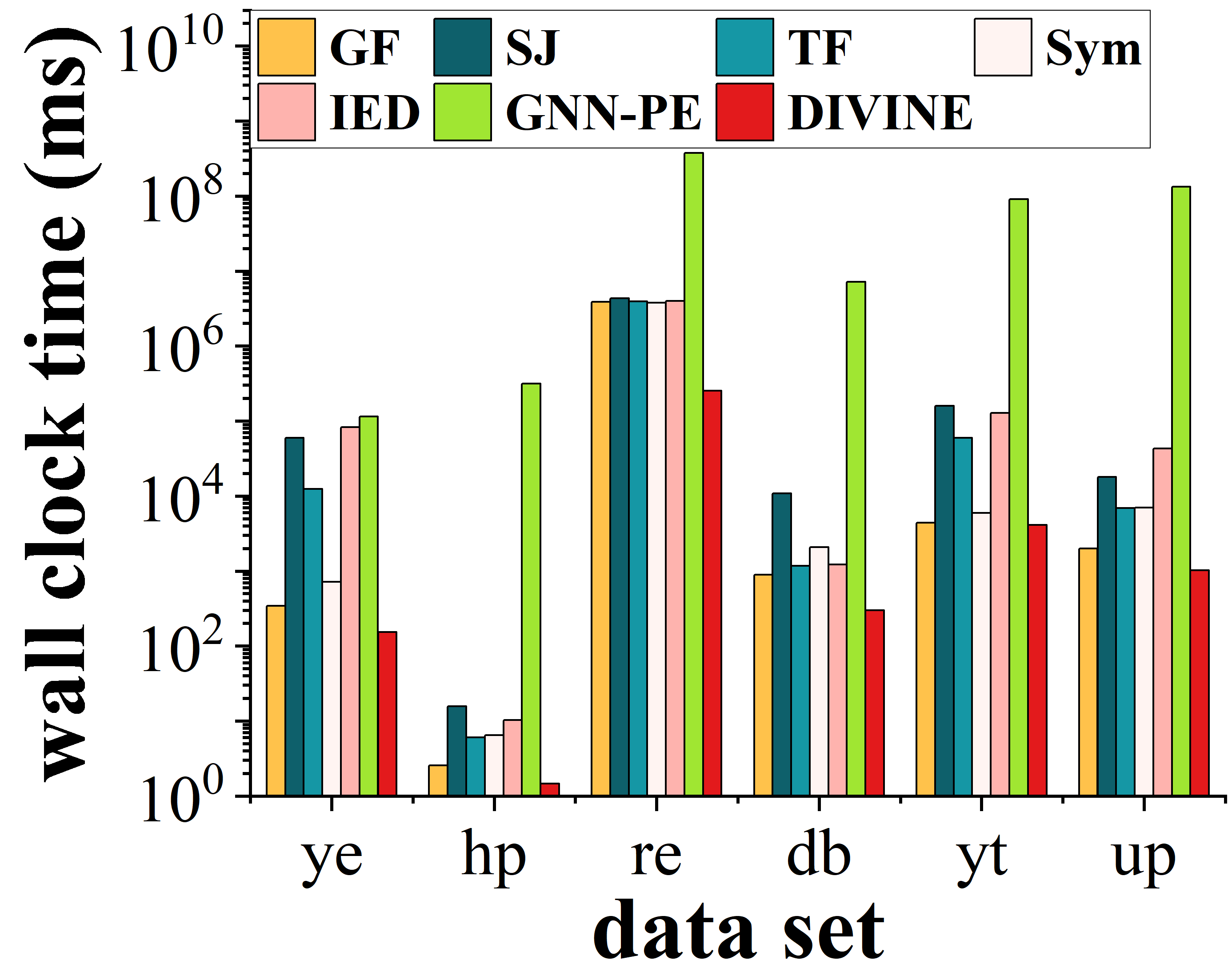}}\label{subfig:ins_real_20}}
\qquad
\subfigure[][{\footnotesize real-world graphs ($Ins_{30\%}$)}]{                    
\scalebox{0.15}[0.15]{\includegraphics{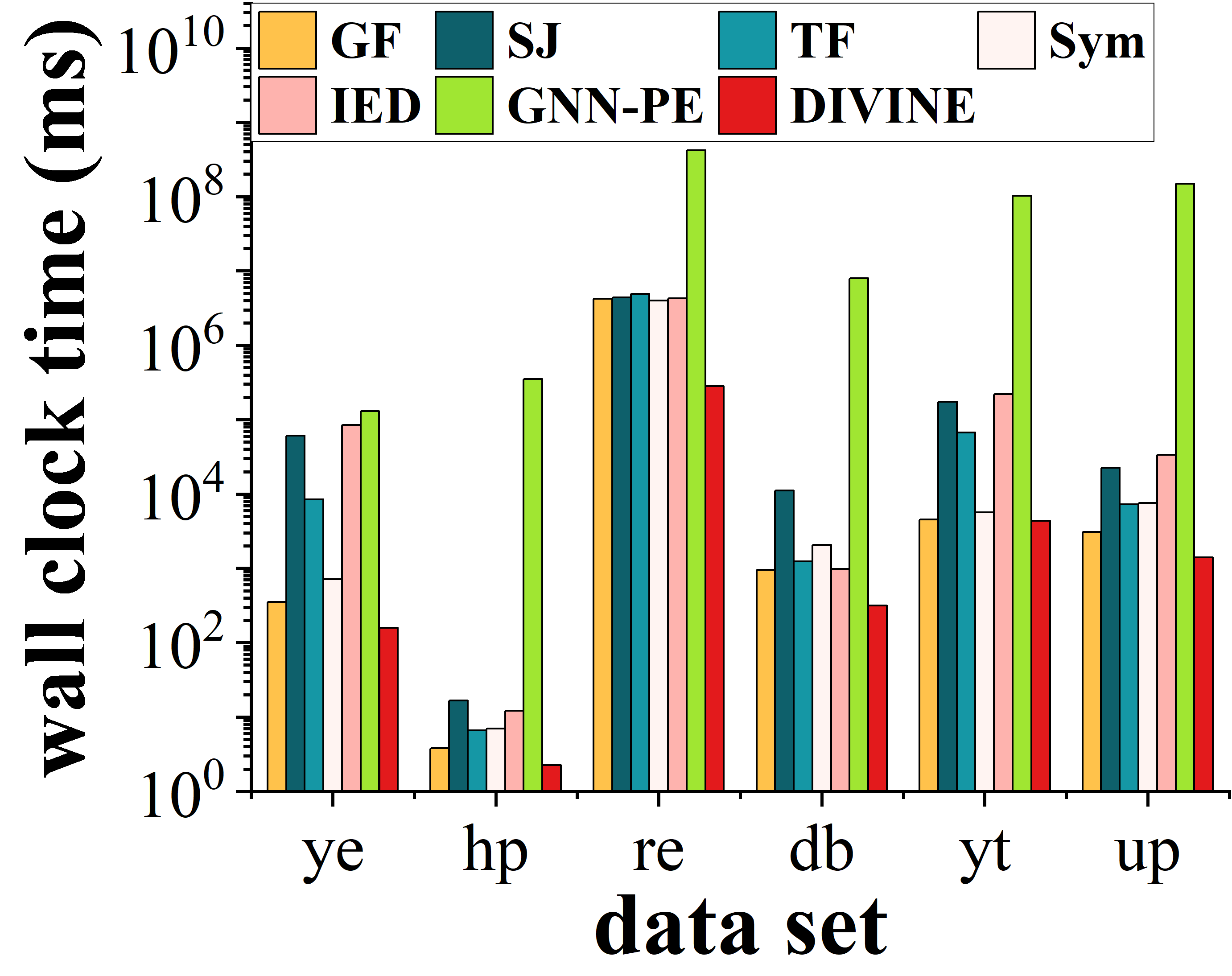}}\label{subfig:ins_real_30}}
\\
\subfigure[][{\footnotesize real-world graphs ($Ins_{40\%}$)}]{                    
\scalebox{0.15}[0.15]{\includegraphics{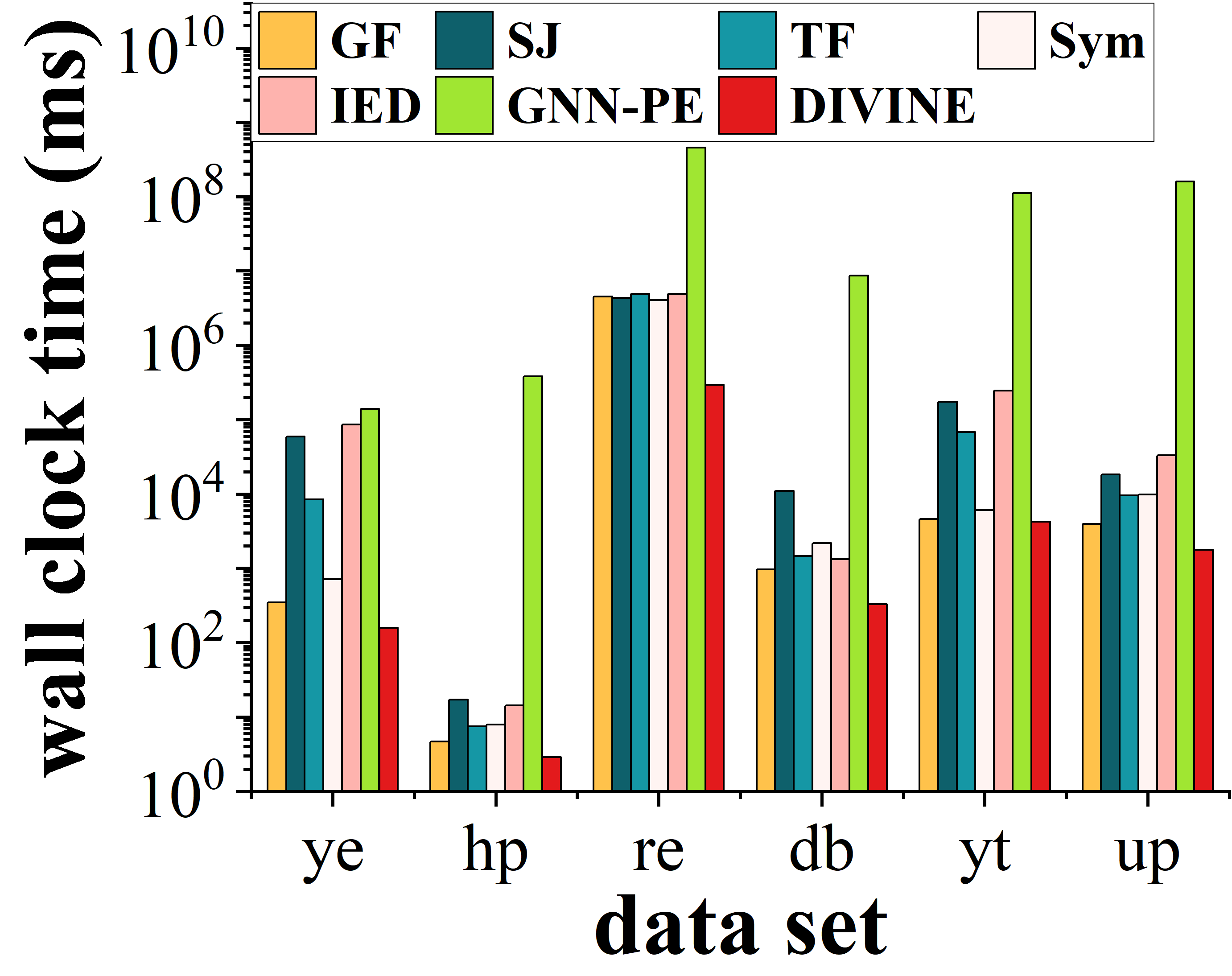}}\label{subfig:ins_real_40}}
\qquad
\subfigure[][{\footnotesize real-world graphs ($Ins_{50\%}$)}]{                    
\scalebox{0.15}[0.15]{\includegraphics{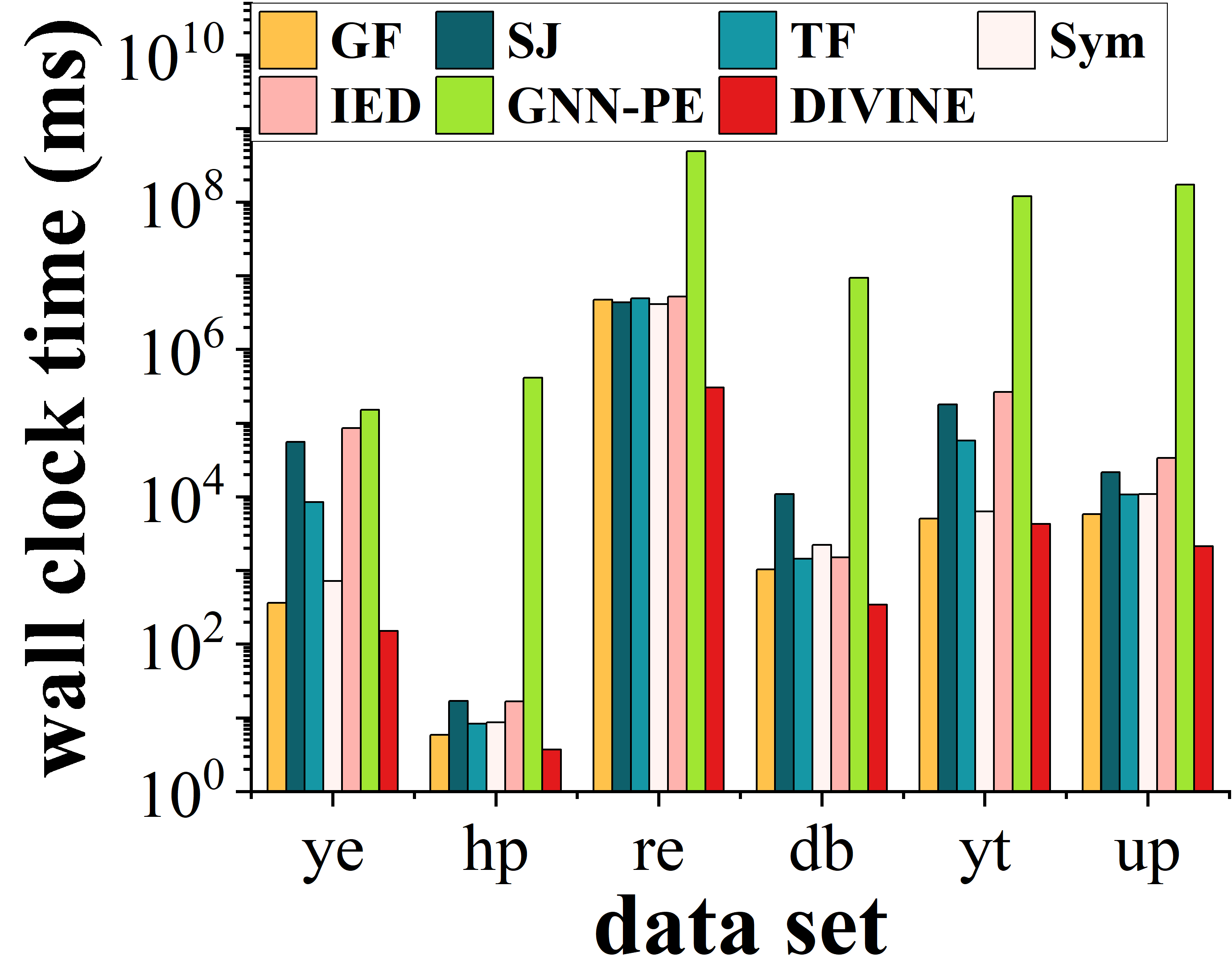}}\label{subfig:ins_real_50}}
\\
\subfigure[][{\footnotesize synthetic graphs ($Ins_{10\%}$)}]{
\scalebox{0.15}[0.15]{\includegraphics{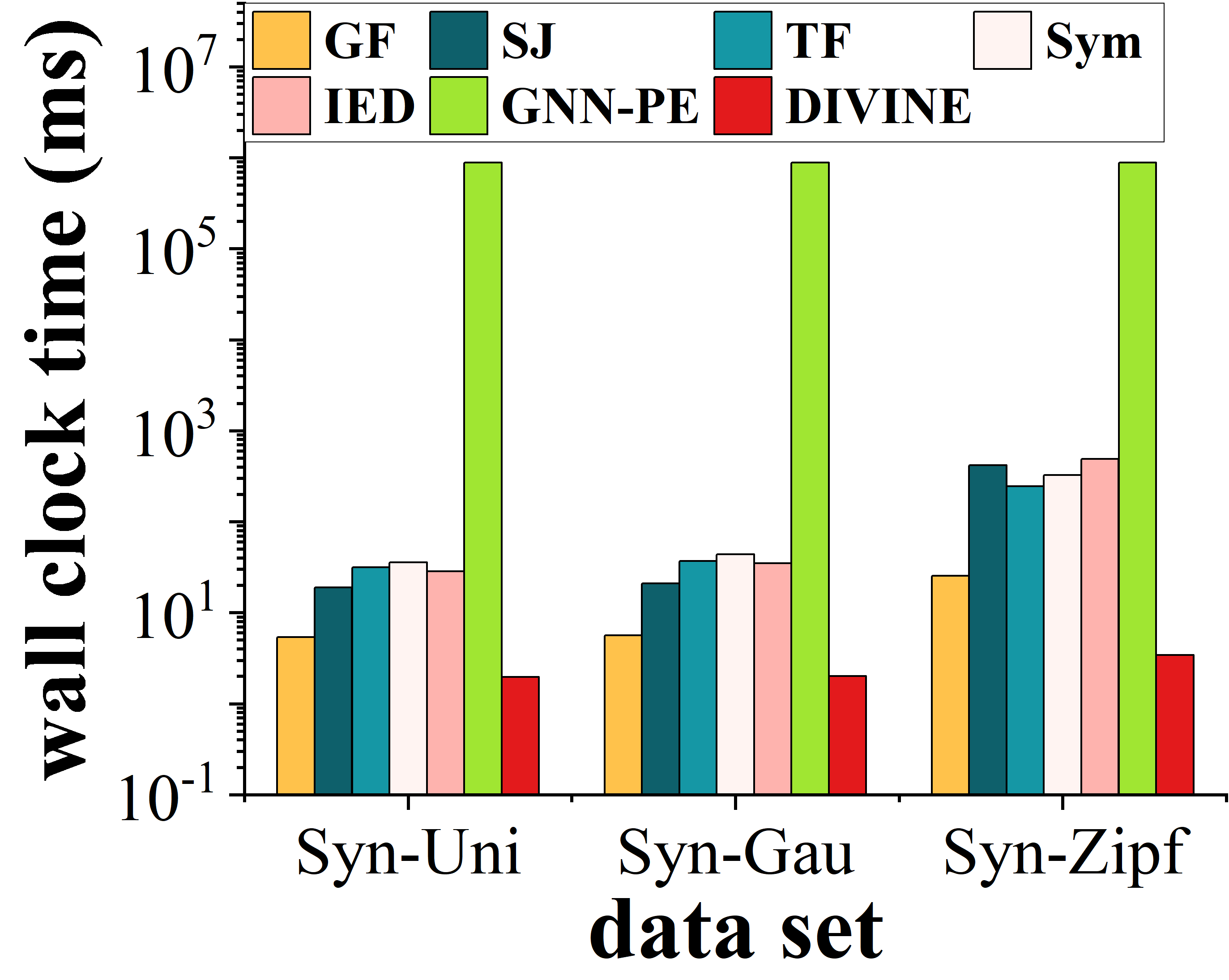}}\label{subfig:dsm_syn}}
\qquad
\subfigure[][{\footnotesize synthetic graphs ($Ins_{20\%}$)}]{                    
\scalebox{0.15}[0.15]{\includegraphics{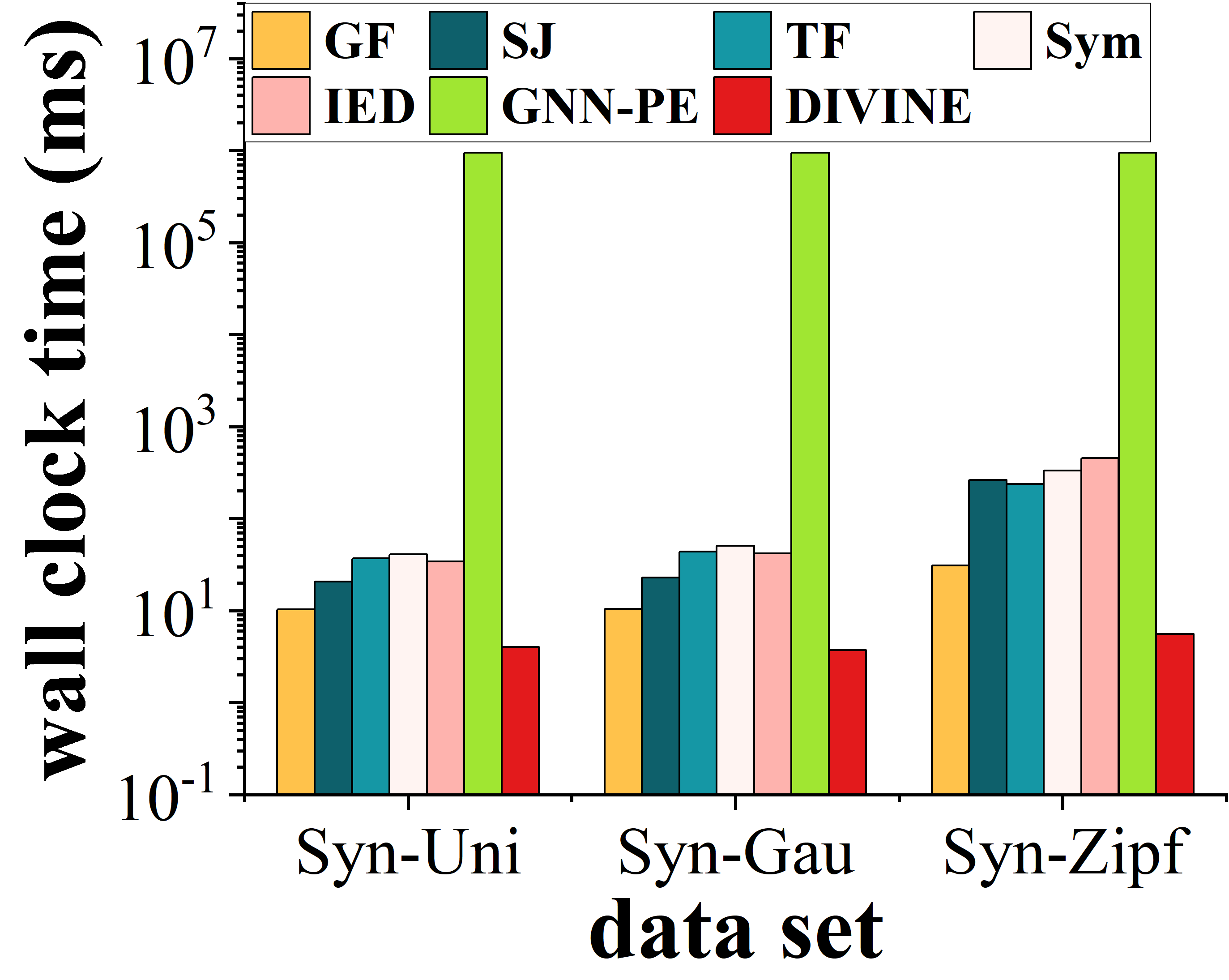}}\label{subfig:ins_syn_20}}
\qquad
\subfigure[][{\footnotesize synthetic graphs ($Ins_{30\%}$)}]{                    
\scalebox{0.15}[0.15]{\includegraphics{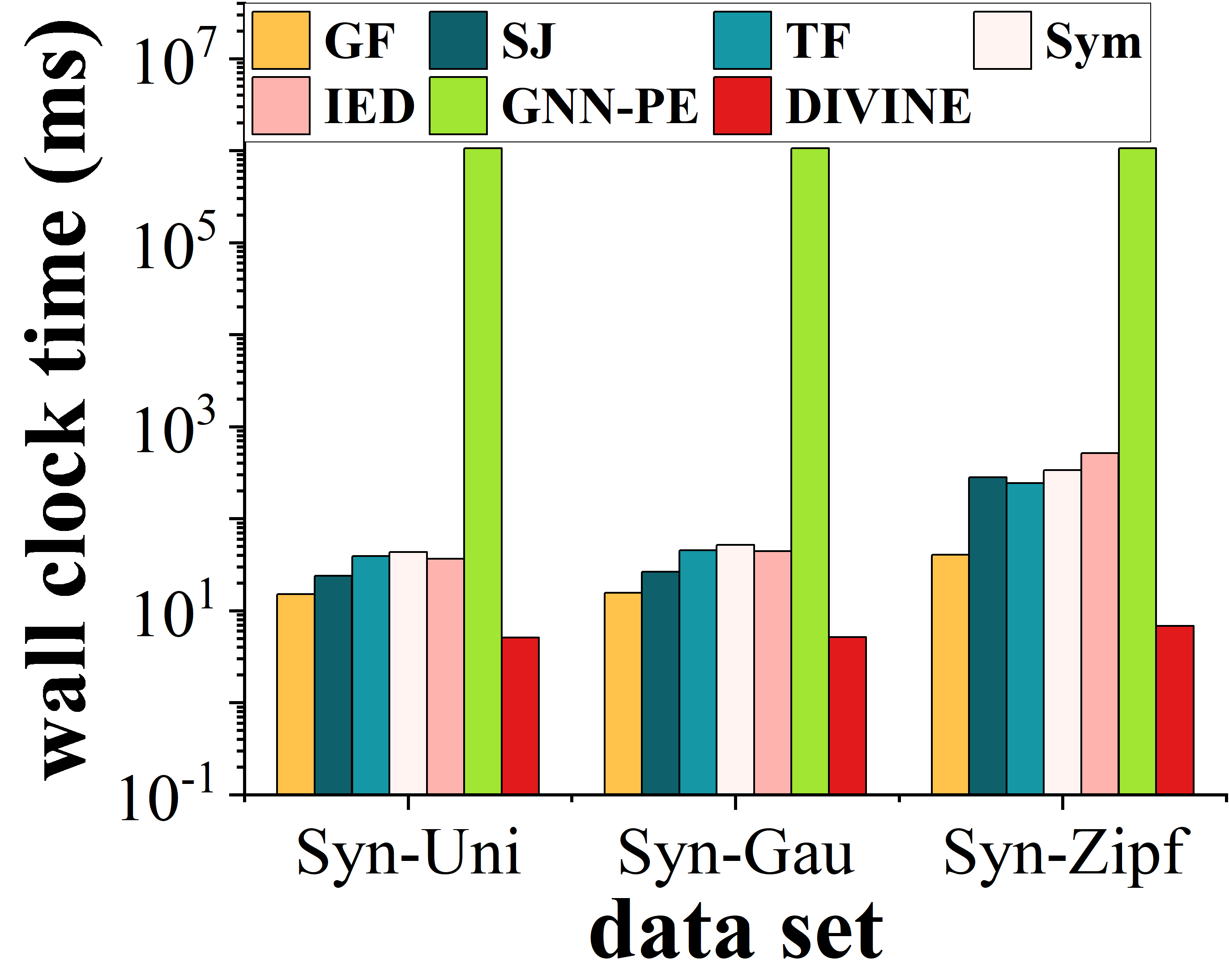}}\label{subfig:ins_syn_30}}
\\
\subfigure[][{\footnotesize synthetic graphs ($Ins_{40\%}$)}]{                    
\scalebox{0.15}[0.15]{\includegraphics{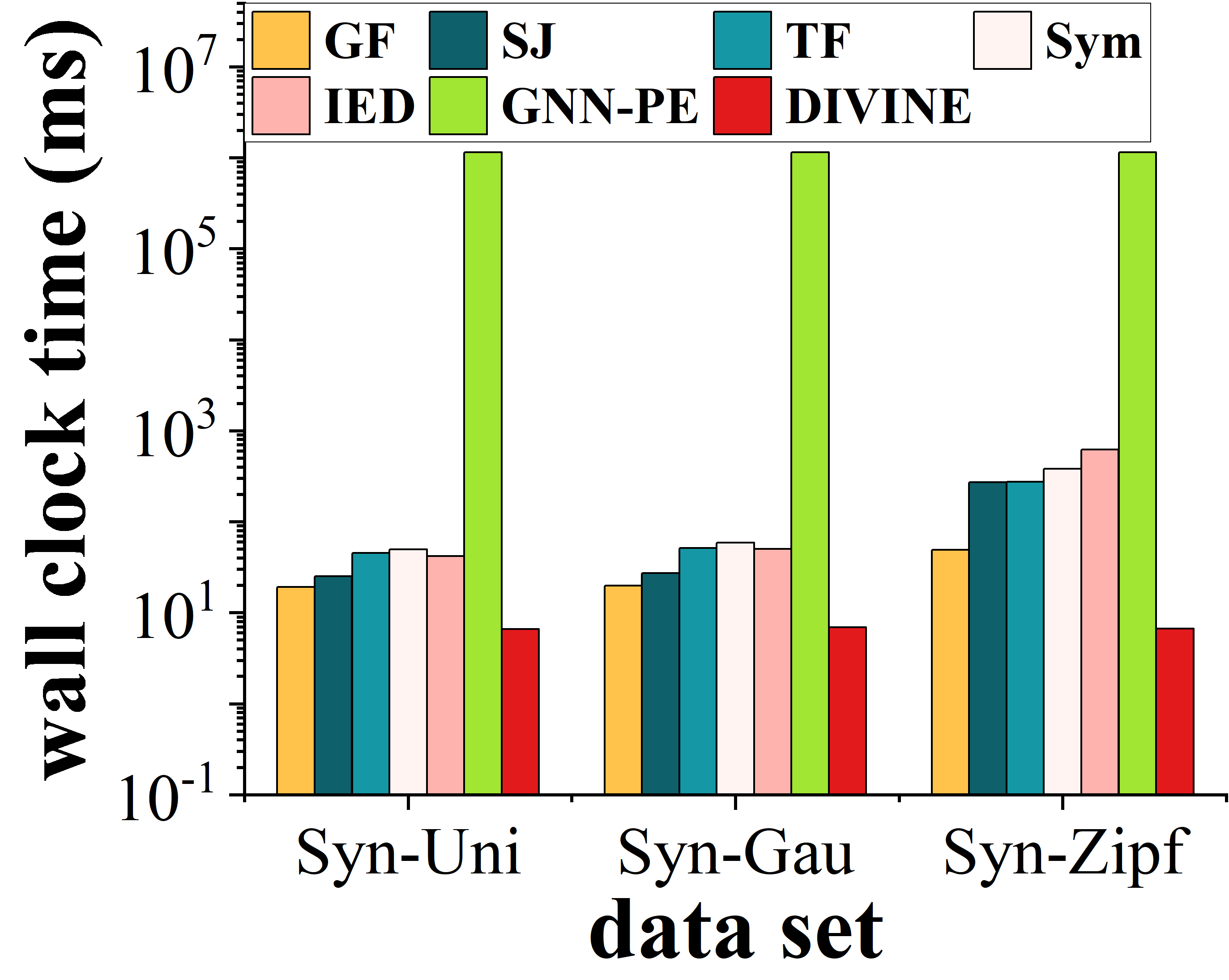}}\label{subfig:ins_syn_40}}
\qquad
\subfigure[][{\footnotesize synthetic graphs ($Ins_{50\%}$)}]{                    
\scalebox{0.15}[0.15]{\includegraphics{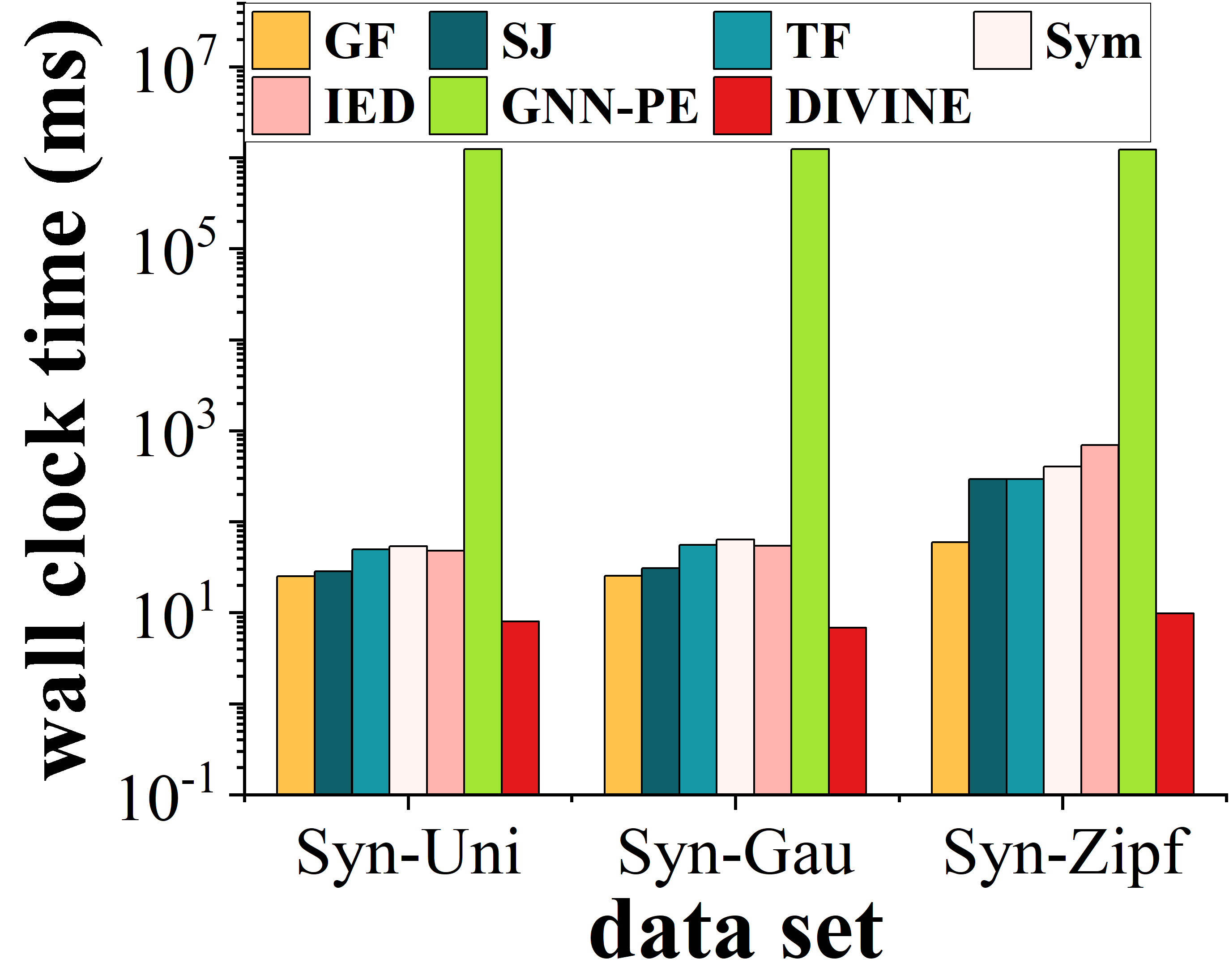}}\label{subfig:ins_syn_50}}
\caption{The DIVINE efficiency on edge-insertion-only real/synthetic graphs, compared with baseline methods.}
\label{fig:dsm_vs_graphs}
\end{figure}

\noindent {\bf The DIVINE Efficiency on Edge-Insertion-Only Real/Synthetic Graphs:}
Figure~\ref{fig:dsm_vs_graphs} compares the efficiency of our DIVINE approach with that of 5 state-of-the-art baseline methods by varying the insertion ratio from 10\% to 50\% over both real and synthetic graphs, where all other parameters are set to default values.
For GNN-PE over static graphs \cite{ye2024efficient}, due to the extremely high cost of GNN re-training, we cannot run the entire edge update stream, thus, we estimate the expected time cost for a period of consecutive (sampled) timestamps (with GNN re-training at least once).
In Figures \ref{subfig:dsm_real} and \ref{subfig:dsm_syn} with the default insertion ratio 10\%, we can see that our DIVINE approach outperforms baseline methods mostly by 1-5 orders of magnitude. 
Note that DIVINE and all baselines require a higher time cost on Reddit (re) than other real/synthetic graphs, due to its high degree (i.e., $avg\_deg(G)=491.98$). Moreover, DIVINE performs better than GNN-PE by 3-5 orders of magnitude, since even one-time GNN re-training in GNN-PE incurs a much higher time cost than our DIVINE approach.

For other subfigures with different insertion ratios from 20\% to 50\%, we can see similar experimental results over both real and synthetic graphs. Although the total query time will increase as the number of inserted edges increases, our DIVINE approach can achieve performance that is up to 1-5 orders of magnitude higher than baseline methods.

\begin{figure}[t]
\centering
\subfigure[][{\footnotesize real-world graphs ($Del_{10\%}$)}]{                    
\scalebox{0.15}[0.15]{\includegraphics{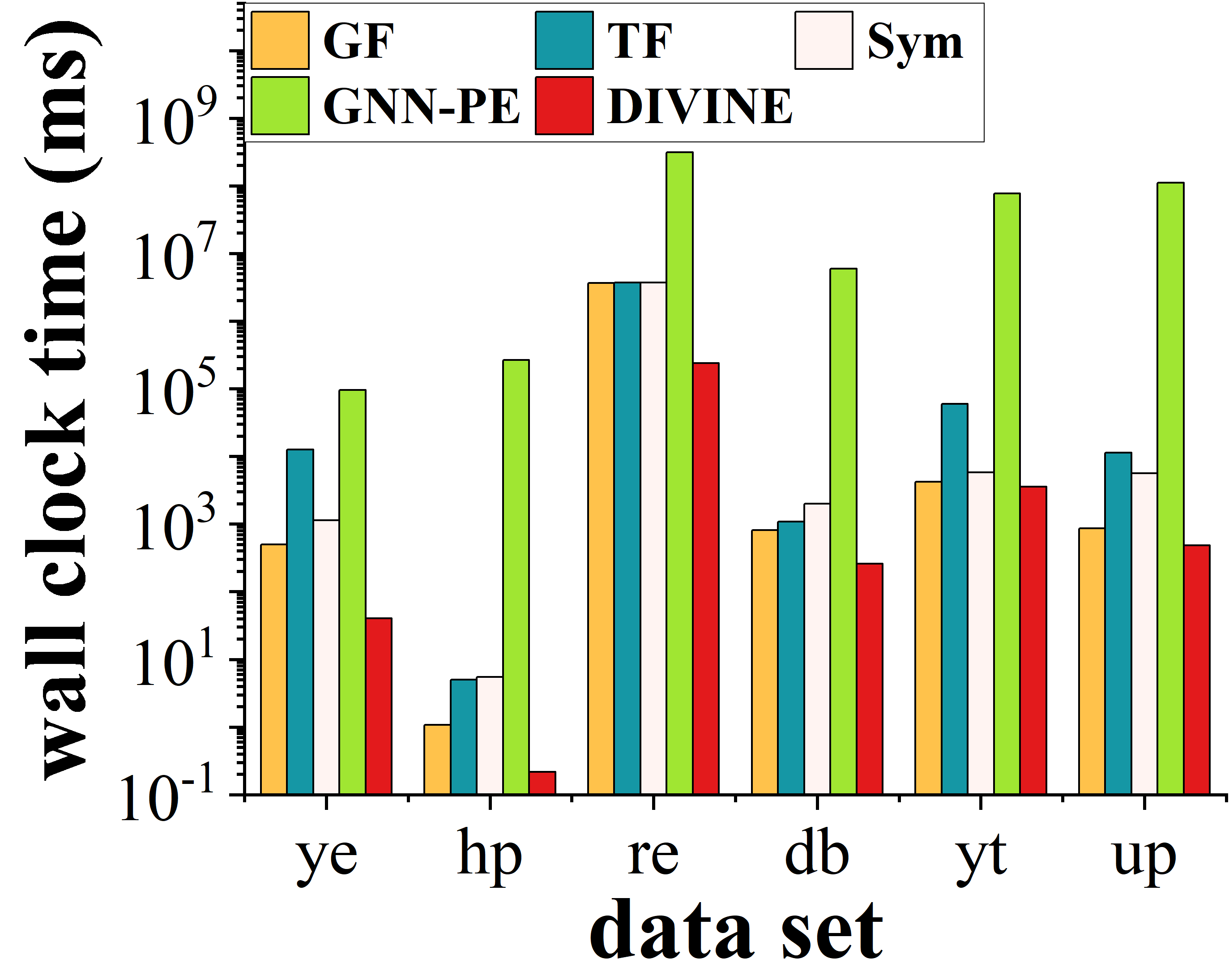}}\label{subfig:dsm_del_real}}
\qquad
\subfigure[][{\footnotesize real-world graphs ($Del_{20\%}$)}]{                    
\scalebox{0.15}[0.15]{\includegraphics{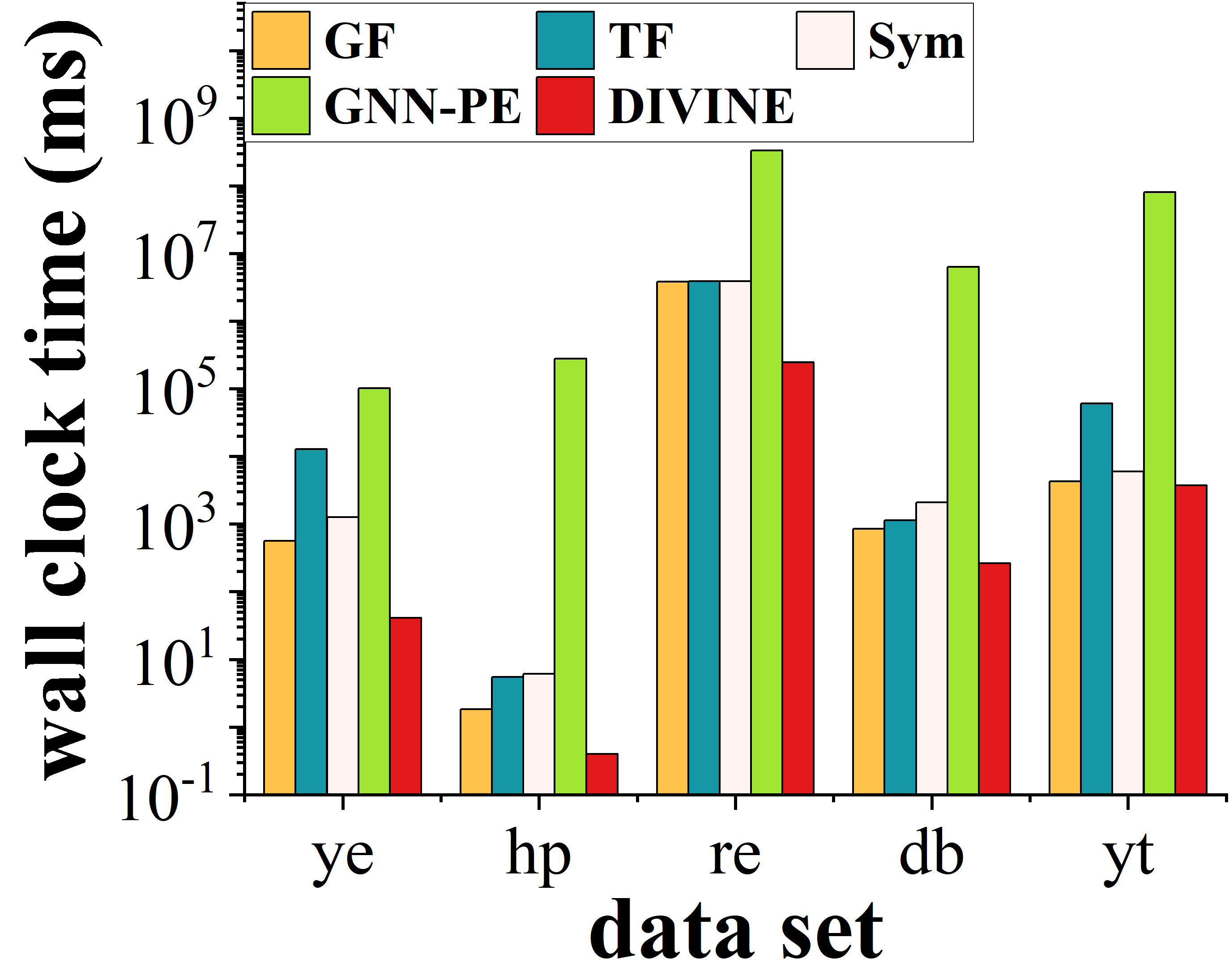}}\label{subfig:del_real_20}}
\qquad
\subfigure[][{\footnotesize real-world graphs ($Del_{30\%}$)}]{                    
\scalebox{0.15}[0.15]{\includegraphics{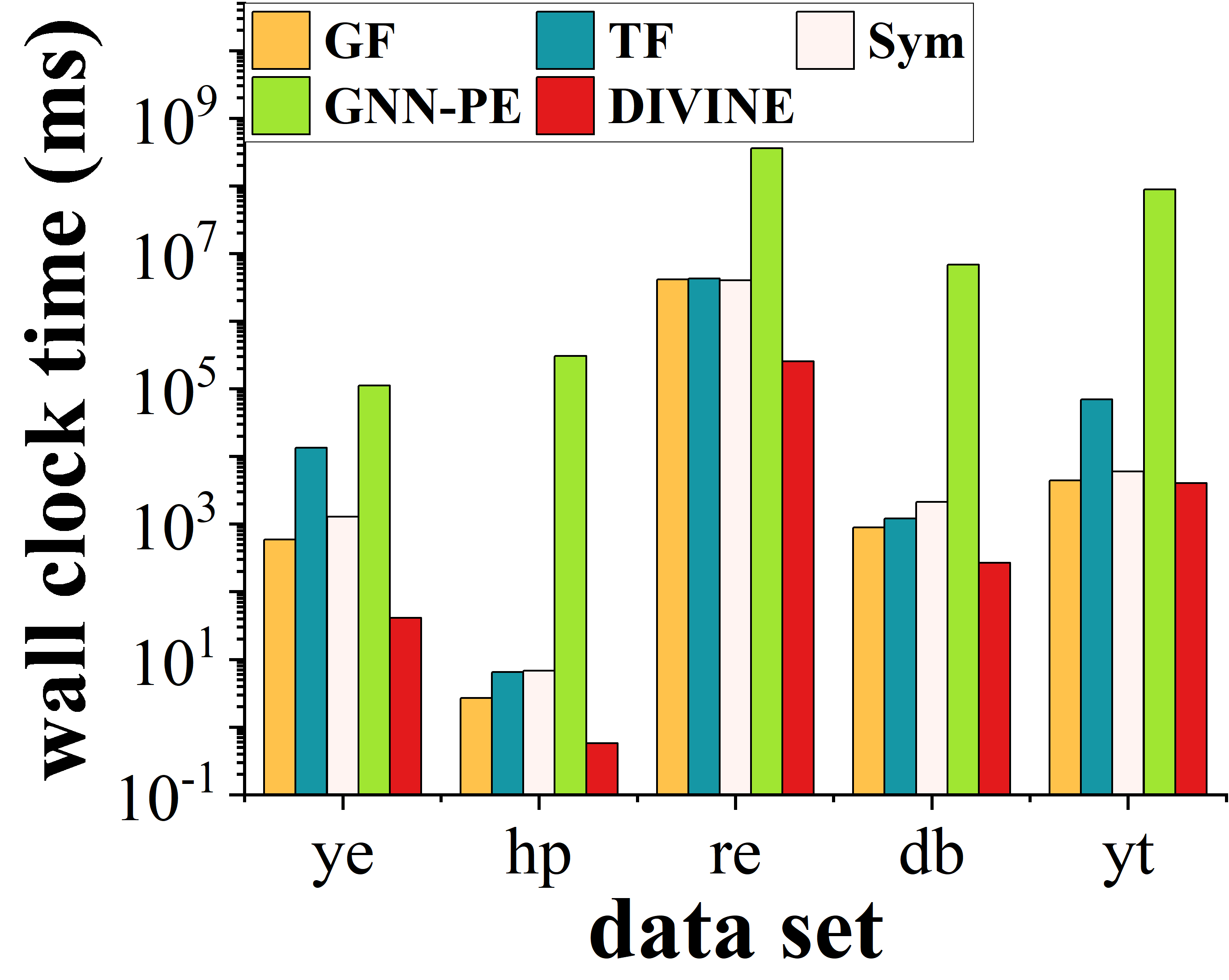}}\label{subfig:del_real_30}}
\\
\subfigure[][{\footnotesize real-world graphs ($Del_{40\%}$)}]{                    
\scalebox{0.15}[0.15]{\includegraphics{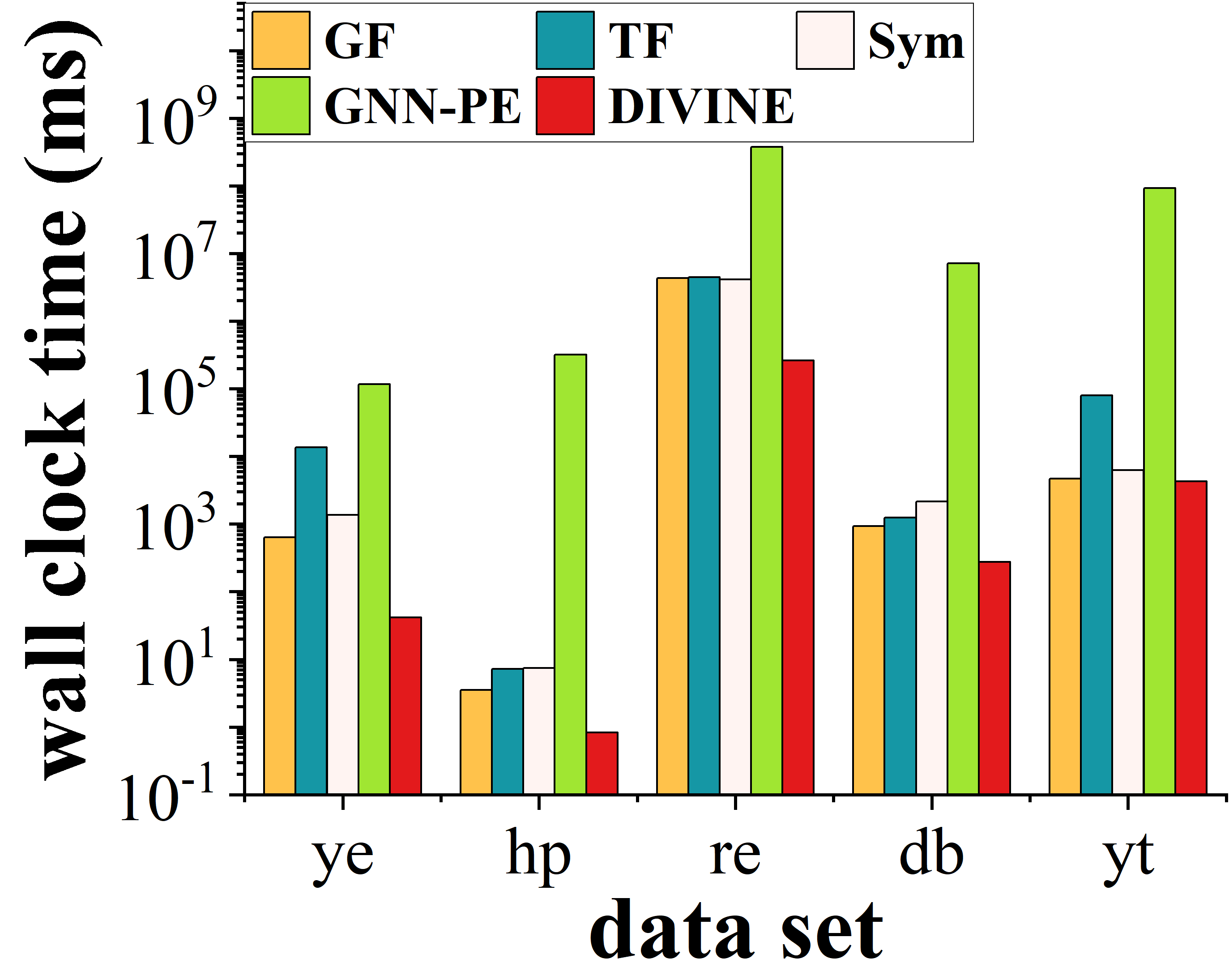}}\label{subfig:del_real_40}}
\qquad
\subfigure[][{\footnotesize real-world graphs ($Del_{50\%}$)}]{                    
\scalebox{0.15}[0.15]{\includegraphics{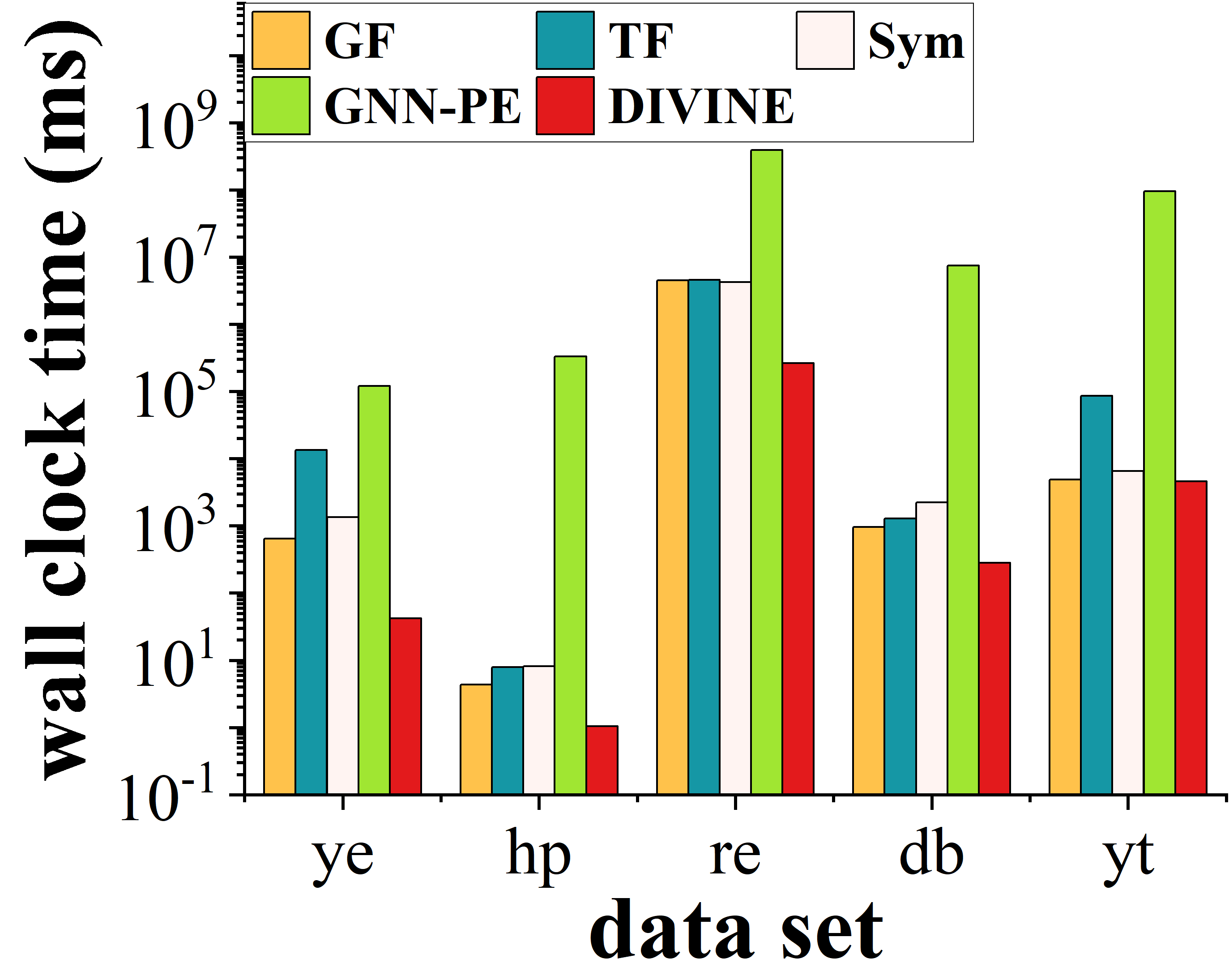}}\label{subfig:del_real_50}}
\\
\subfigure[][{\footnotesize synthetic graphs ($Del_{10\%}$)}]{
\scalebox{0.15}[0.15]{\includegraphics{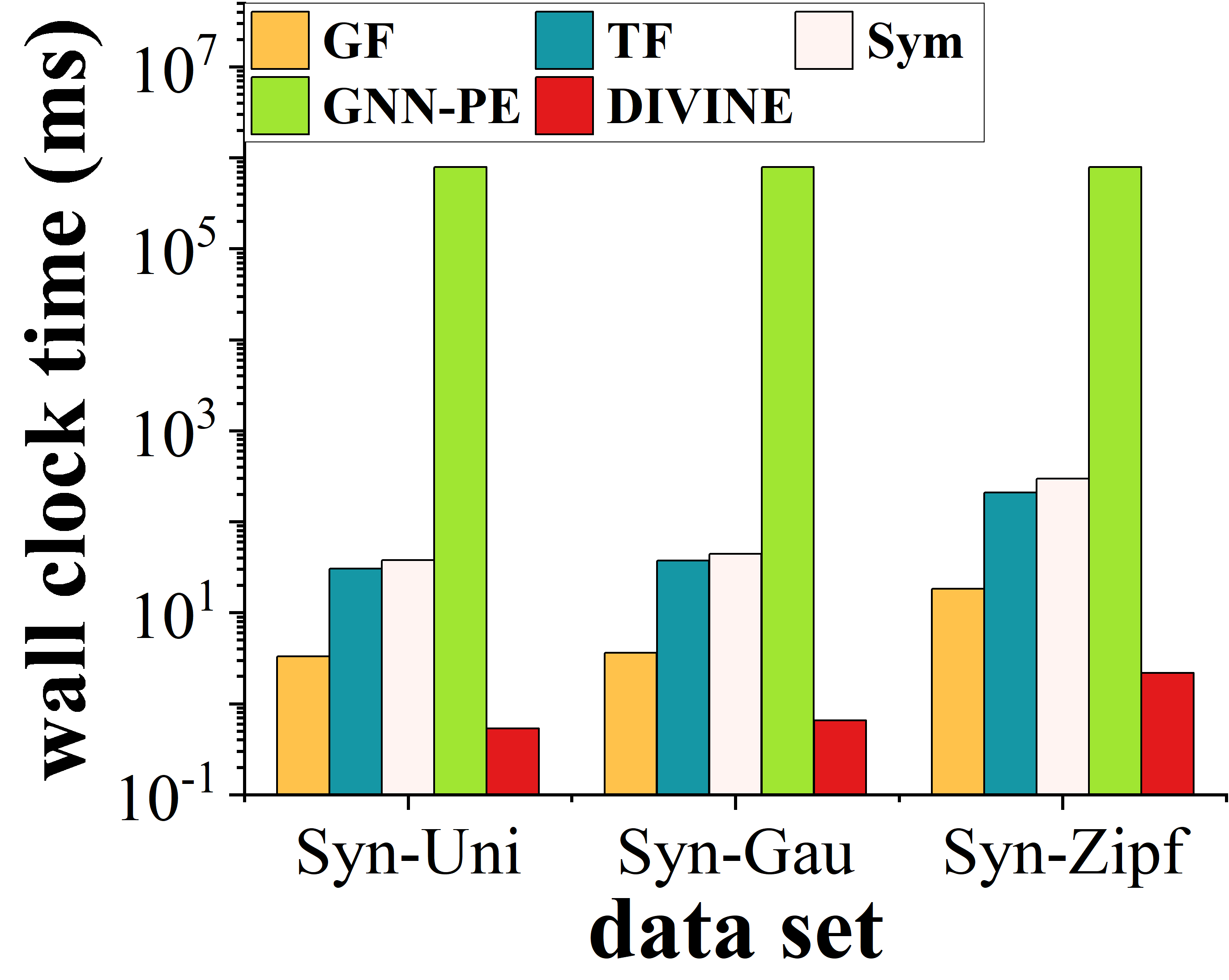}}\label{subfig:dsm_del_syn}}
\qquad
\subfigure[][{\footnotesize synthetic graphs ($Del_{20\%}$)}]{                    
\scalebox{0.15}[0.15]{\includegraphics{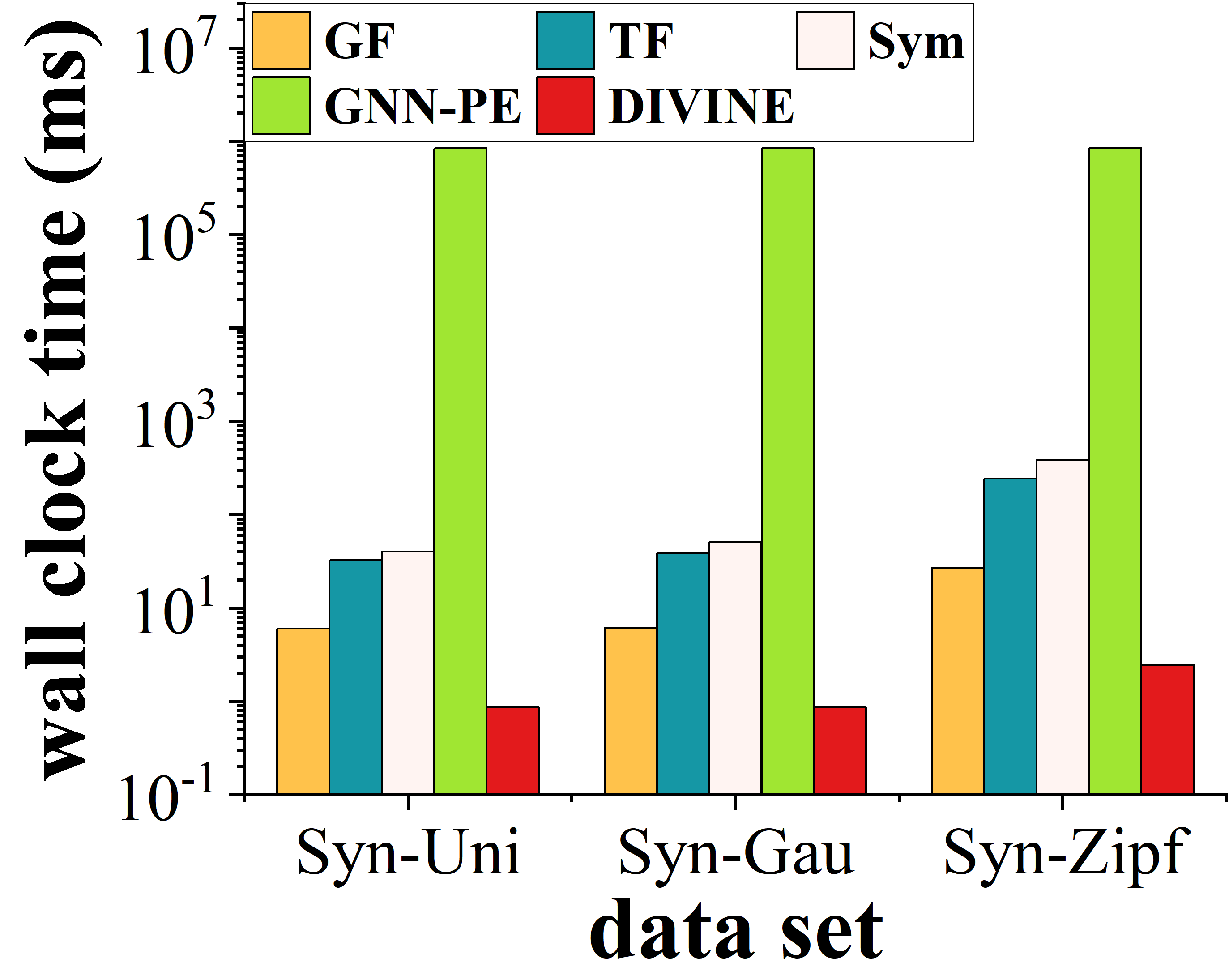}}\label{subfig:del_syn_20}}
\qquad
\subfigure[][{\footnotesize synthetic graphs ($Del_{30\%}$)}]{                    
\scalebox{0.15}[0.15]{\includegraphics{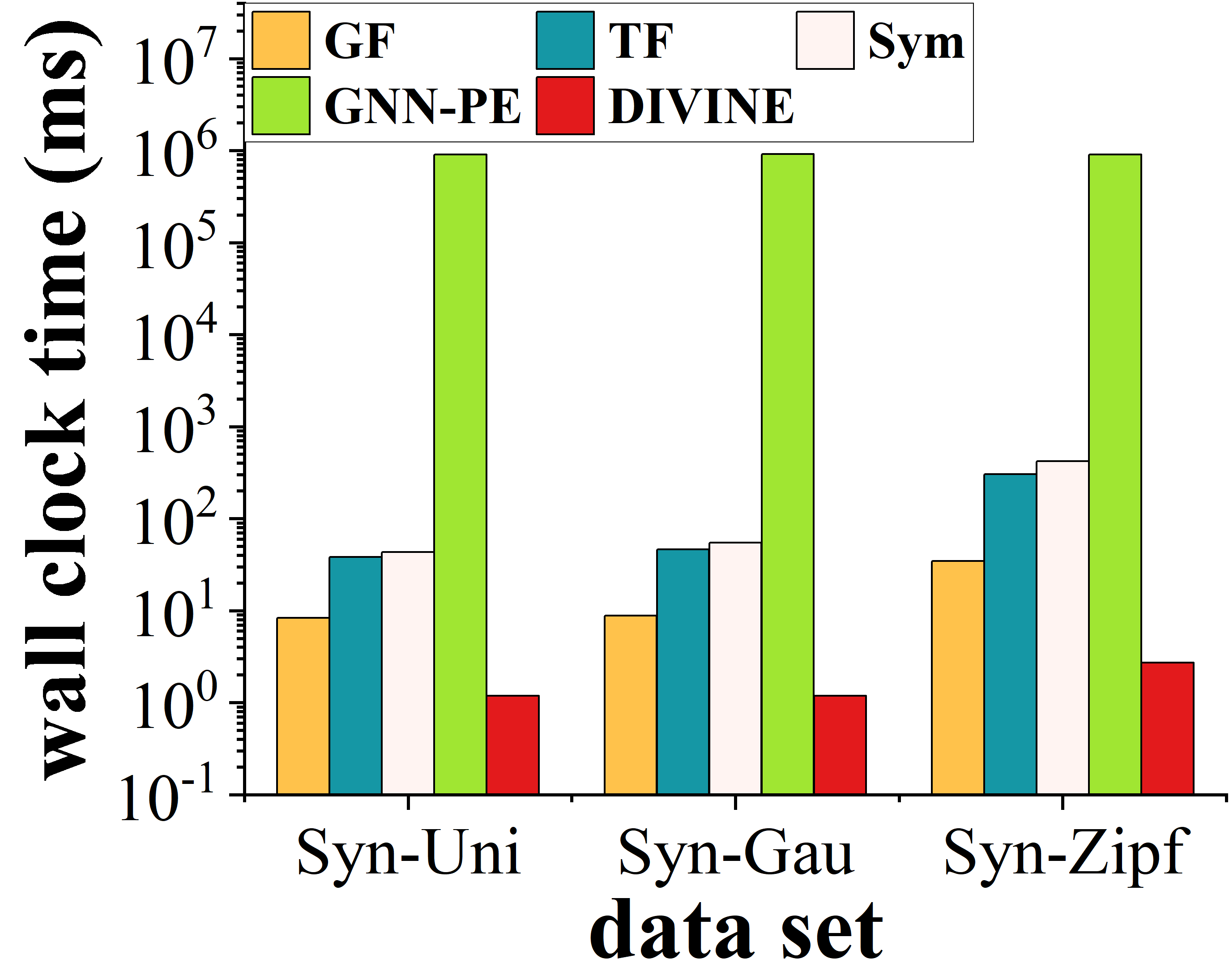}}\label{subfig:del_syn_30}}
\\
\subfigure[][{\footnotesize synthetic graphs ($Del_{40\%}$)}]{                    
\scalebox{0.15}[0.15]{\includegraphics{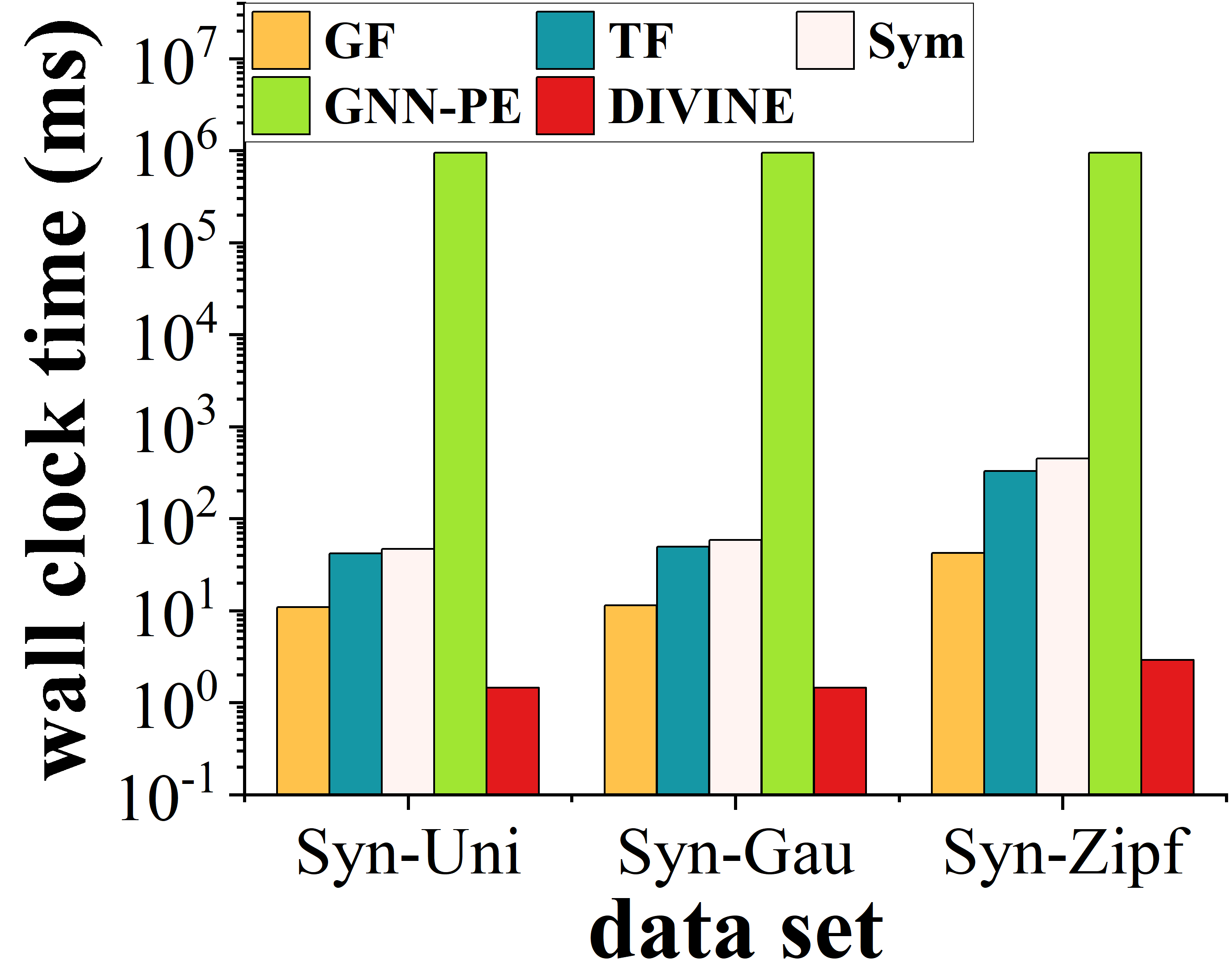}}\label{subfig:del_syn_40}}
\qquad
\subfigure[][{\footnotesize synthetic graphs ($Del_{50\%}$)}]{                    
\scalebox{0.15}[0.15]{\includegraphics{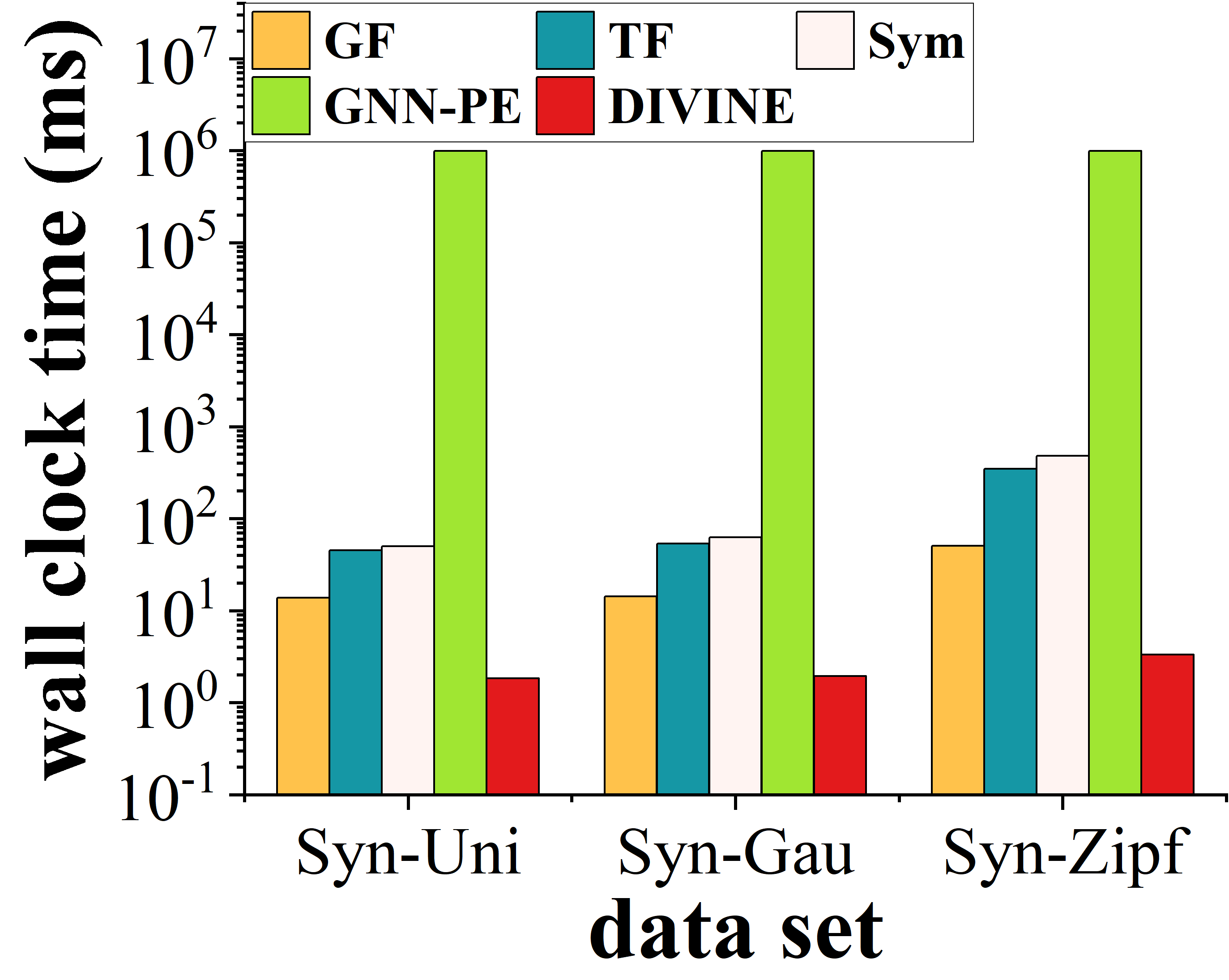}}\label{subfig:del_syn_50}}
\caption{The DIVINE efficiency on edge-deletion-only real/synthetic graphs, compared with baseline methods.}
\label{fig:dsm_deletion}
\end{figure}

\noindent{\bf The DIVINE Efficiency on Edge-Deletion-Only Real/Synthetic Graphs:}
Since SJ and IED baselines do not support the edge deletion \cite{sun2022depth}, Figure~\ref{fig:dsm_deletion} only compares the efficiency of our DIVINE approach with that of GF, TF, SYM, and GNN-PE over real/synthetic graphs, where the deletion ratio varies from 10\% to 50\% and all other parameters are set to their default values.
From Figures \ref{subfig:dsm_del_real} and \ref{subfig:dsm_del_syn} with the default deletion ratio 10\%, our DIVINE approach can always outperform the four baseline methods. 
Overall, for all real/synthetic data sets (even for $re$ with 5.73$M$ edge deletions), the wall clock time of our DIVINE approach remains the lowest, compared with baselines.

From other subfigures with different deletion ratios 20\% $\sim$ 50\%, we can find similar experimental results over both real and synthetic graphs, where our DIVINE approach can always outperform the four baseline methods by up to 1-5 orders of magnitude.

\begin{figure}[t]
\centering
\subfigure[][{\footnotesize real-world graphs ($Ins_{25\%}$)}]{                    
\scalebox{0.15}[0.15]{\includegraphics{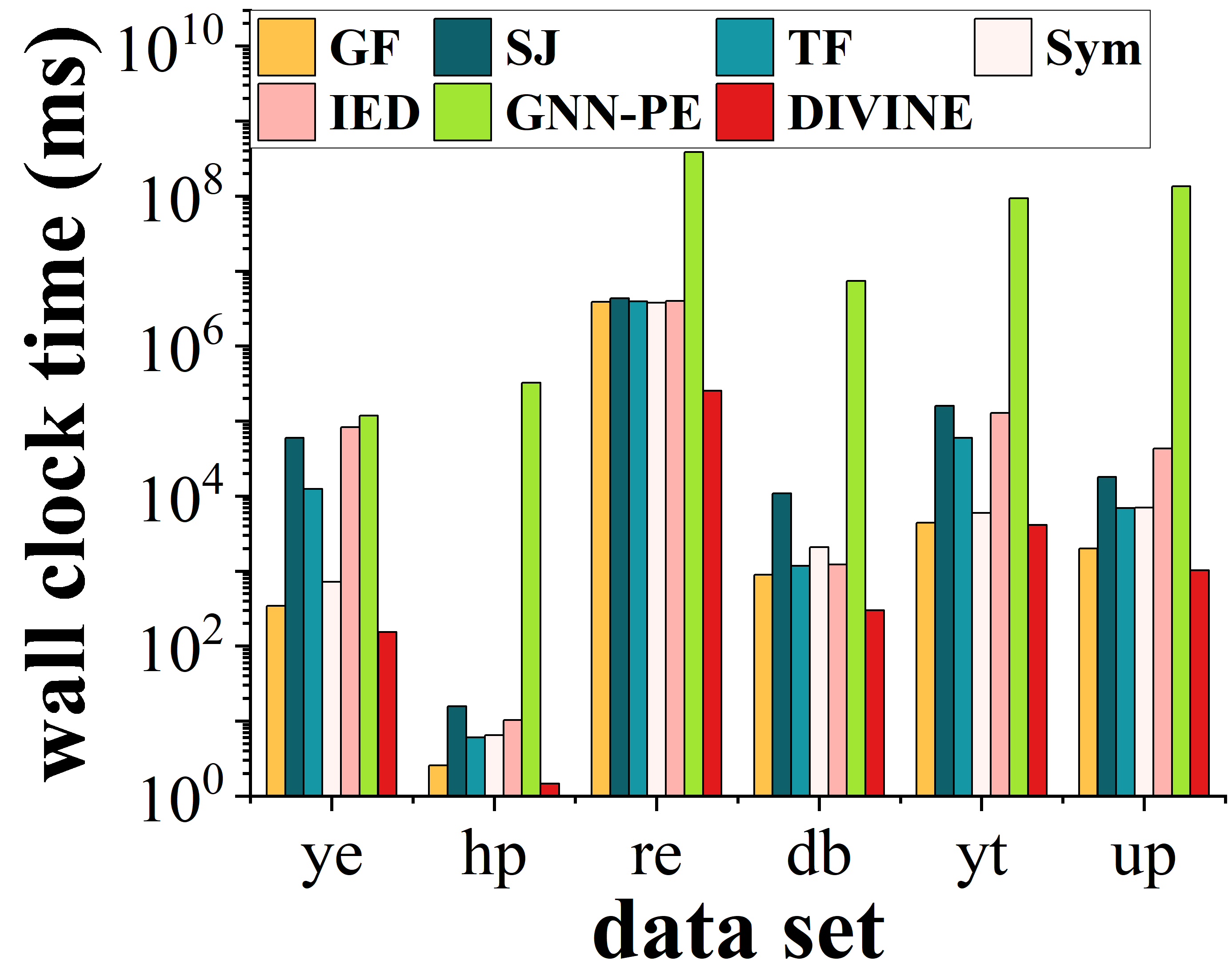}}\label{subfig:ins_real_25}}
\qquad
\subfigure[][{\footnotesize real-world graphs ($Ins_{50\%}$)}]{                    
\scalebox{0.15}[0.15]{\includegraphics{figures/ins-real-50.png}}\label{subfig:ins_real_50_1}}
\qquad
\subfigure[][{\footnotesize real-world graphs ($Del_{25\%}$)}]{                    
\scalebox{0.15}[0.15]{\includegraphics{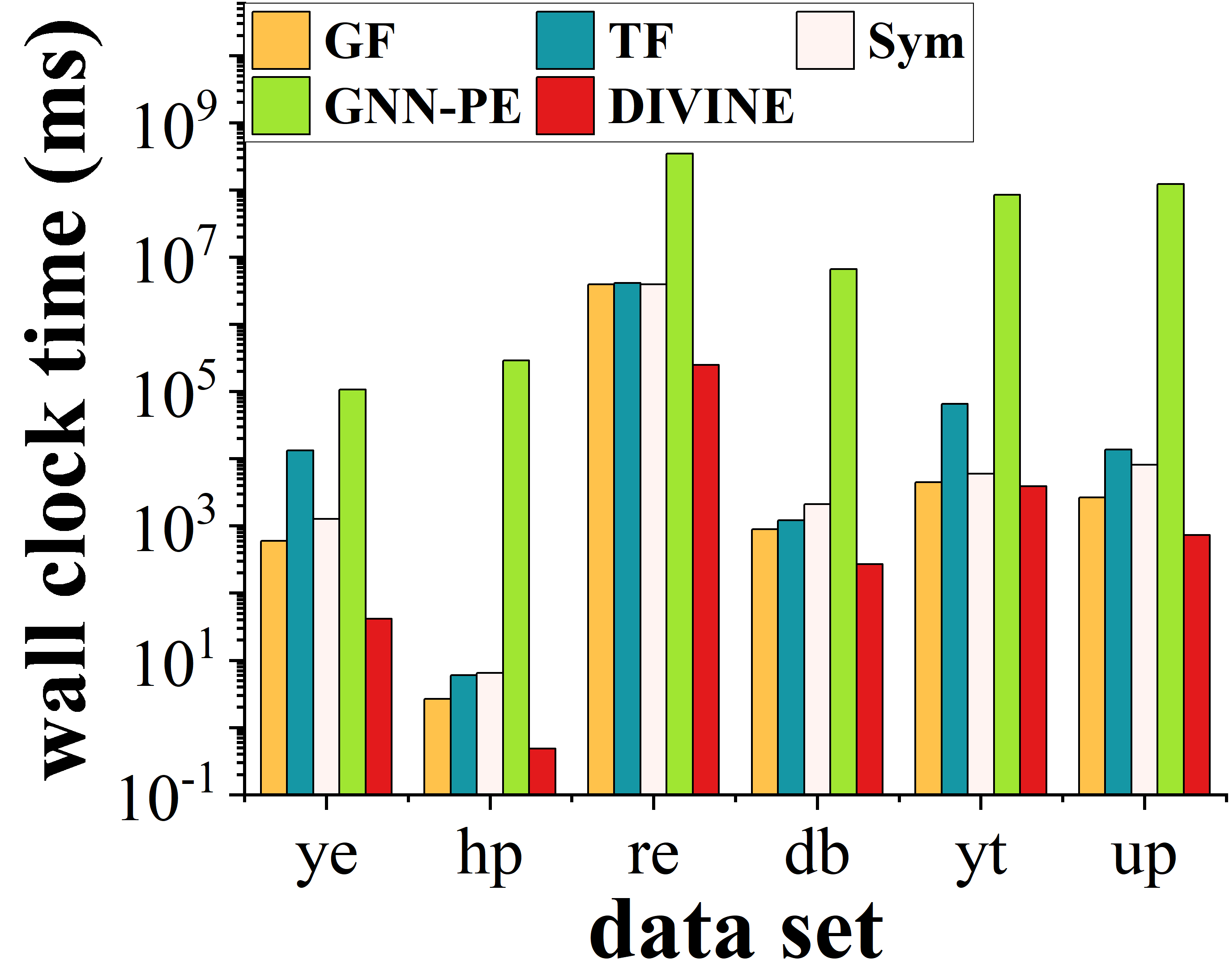}}\label{subfig:des_real_25}}
\\
\subfigure[][{\footnotesize real-world graphs ($Del_{50\%}$)}]{                    
\scalebox{0.15}[0.15]{\includegraphics{figures/des-real-50.png}}\label{subfig:des_real_50_1}}
\qquad
\subfigure[][{\footnotesize real-world graphs ($Ins_{300K}$)}]{                    
\scalebox{0.15}[0.15]{\includegraphics{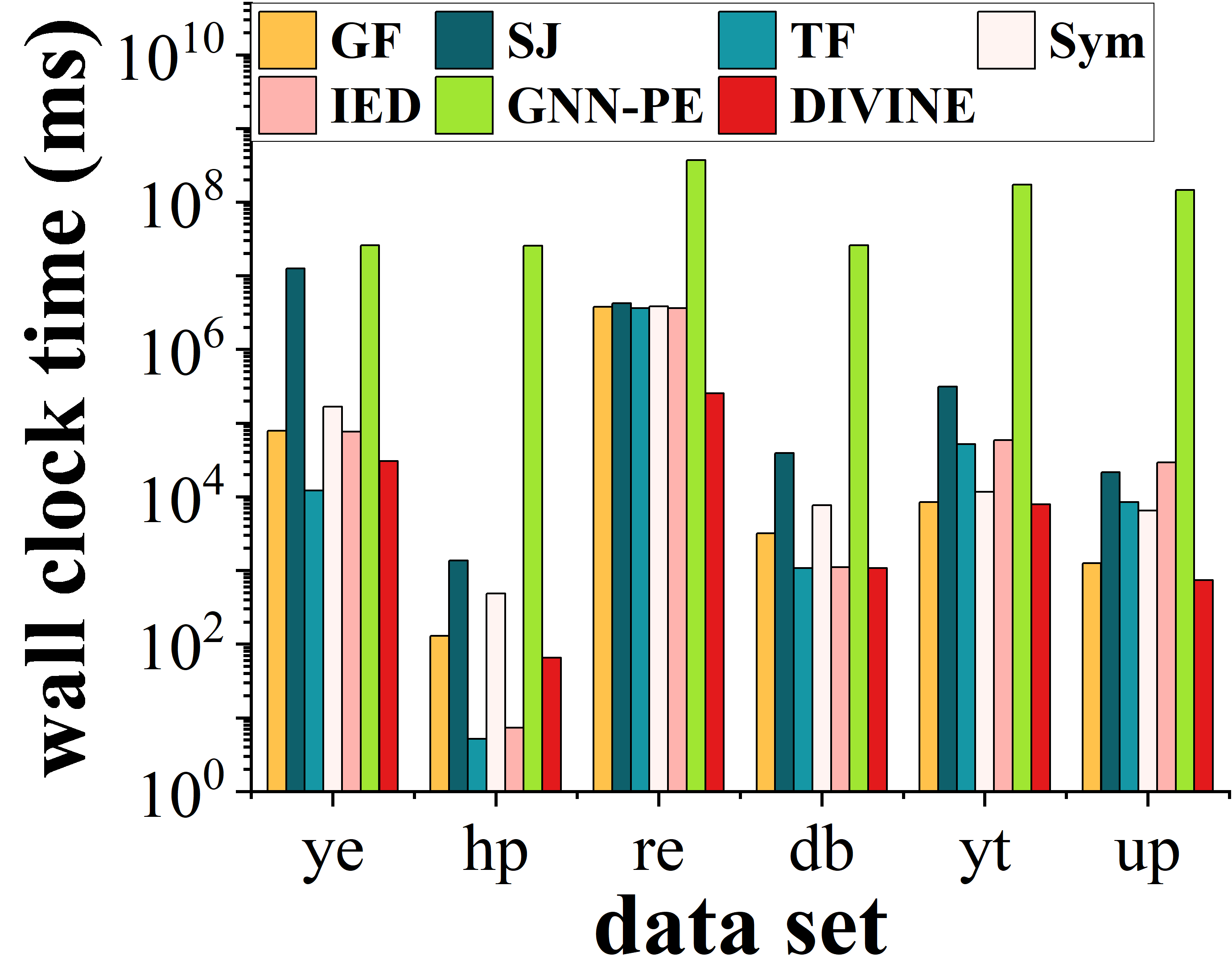}}\label{subfig:ins_real_300K}}
\\
\subfigure[][{\footnotesize synthetic graphs ($Ins_{25\%}$)}]{
\scalebox{0.15}[0.15]{\includegraphics{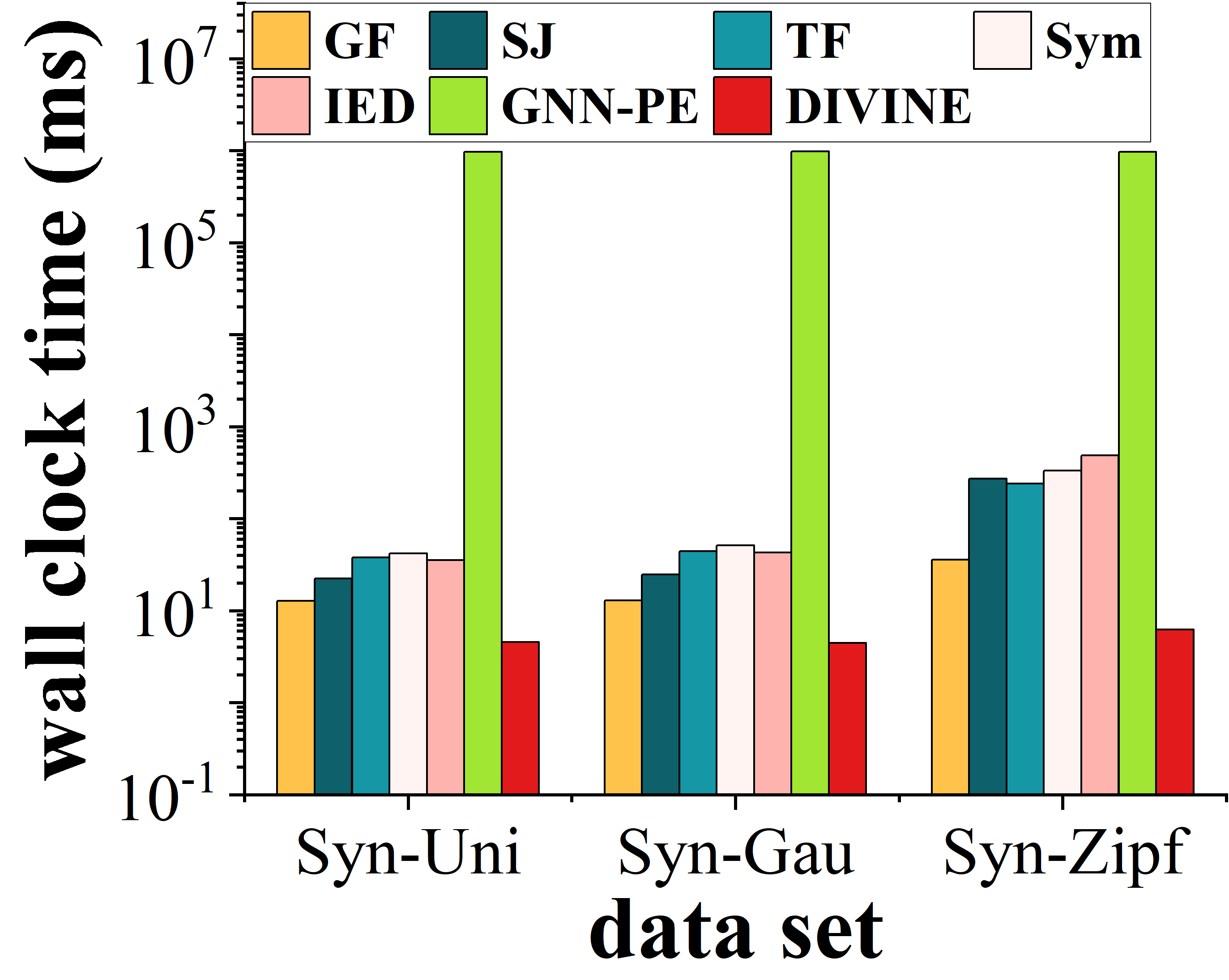}}\label{subfig:ins_syn_25}}
\qquad
\subfigure[][{\footnotesize synthetic graphs ($Ins_{50\%}$)}]{
\scalebox{0.15}[0.15]{\includegraphics{figures/ins-syn-50.png}}\label{subfig:ins_syn_50_1}}
\qquad
\subfigure[][{\footnotesize synthetic graphs ($Del_{25\%}$)}]{
\scalebox{0.15}[0.15]{\includegraphics{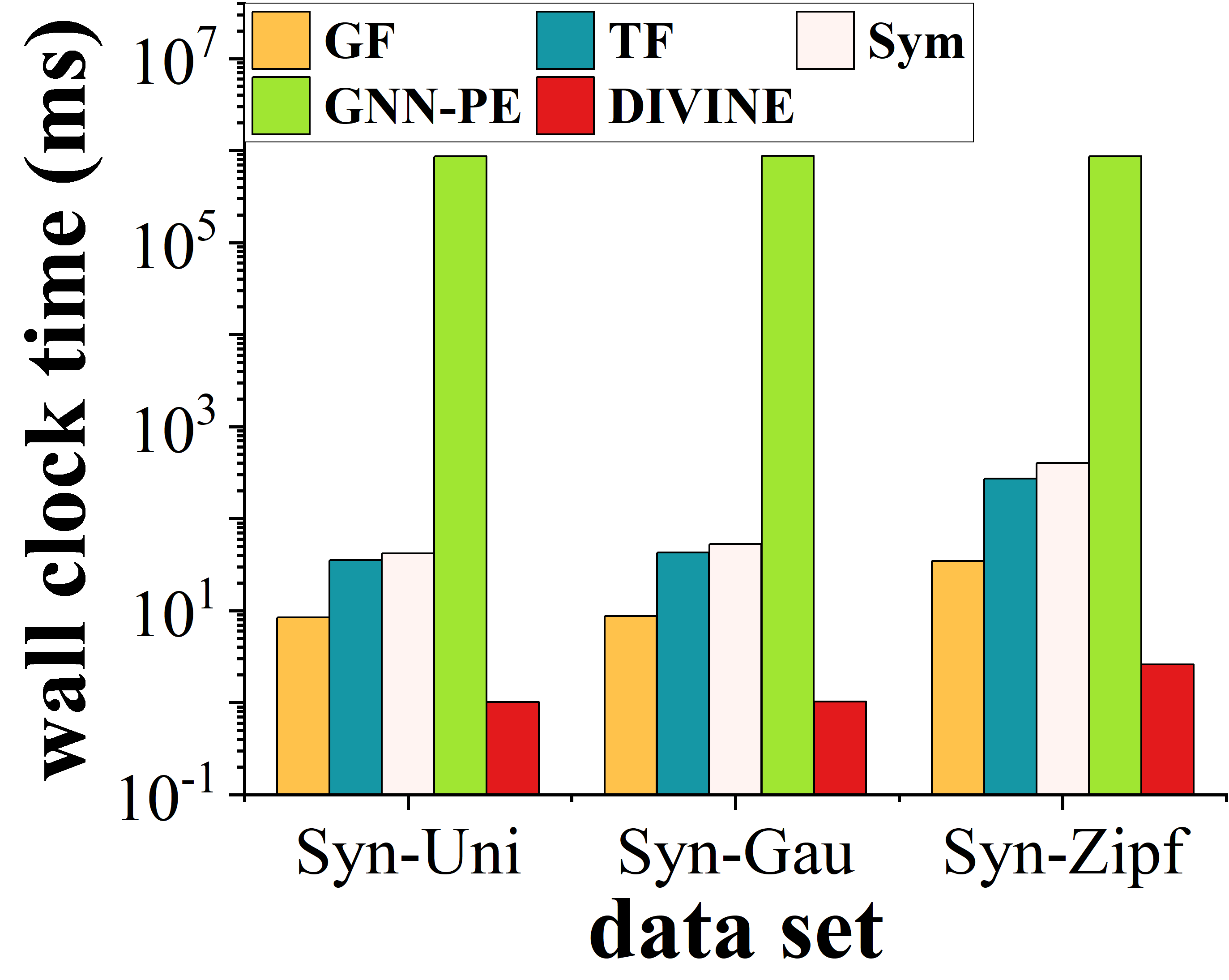}}\label{subfig:des_syn_25}}
\\
\subfigure[][{\footnotesize synthetic graphs ($Del_{50\%}$)}]{
\scalebox{0.15}[0.15]{\includegraphics{figures/des-syn-50.png}}\label{subfig:des_syn_50_1}}
\qquad
\subfigure[][{\footnotesize synthetic graphs ($Ins_{300K}$)}]{
\scalebox{0.15}[0.15]{\includegraphics{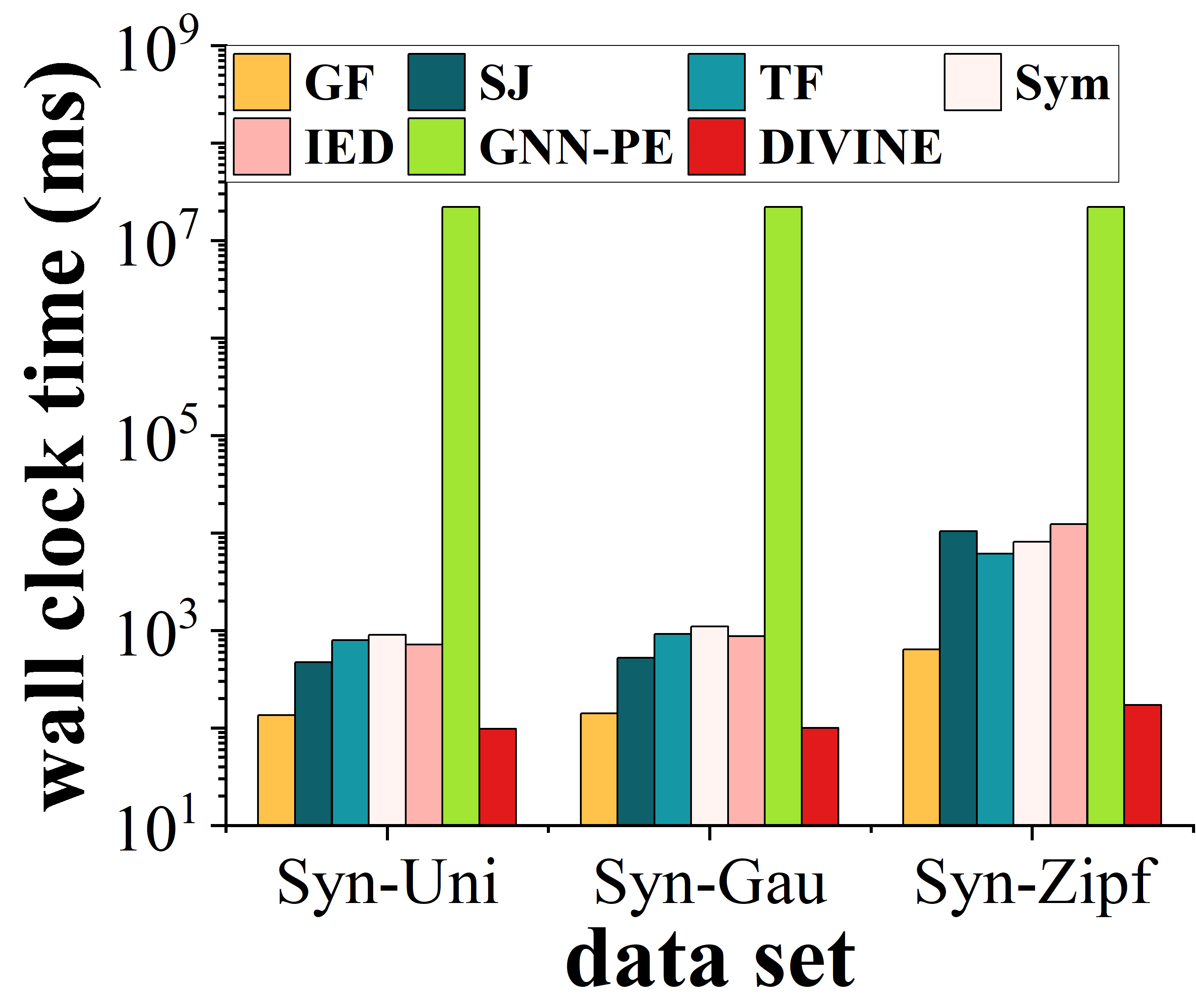}}\label{subfig:ins_syn_300K}}
\caption{The DIVINE efficiency with high/bursty edge update ratios, compared with baseline methods.}
\label{fig:dsm_update_ratio}
\end{figure}

\noindent{\bf The DIVINE Efficiency with Different Edge Update Ratios:}
Figure \ref{fig:dsm_update_ratio} shows the efficiency of our DIVINE approach and baselines over real/synthetic graphs, for high edge insertion rates ($Ins$): 25\% and 50\%, high edge deletion rates ($Del$): 25\% and 50\%, and high/bursty updates of $300K$ edge insertions ($Ins_{300K}$), where other parameters are set to default values. Similar to results with default 10\% edge insertion and deletion rates, respectively, our DIVINE approach always outperforms baselines by up to 1-5 orders of magnitude, for high or bursty edge update ratios.

\begin{figure}[t]
\centering
\subfigure[][{real-world graphs}]{                    
\scalebox{0.15}[0.15]{\includegraphics{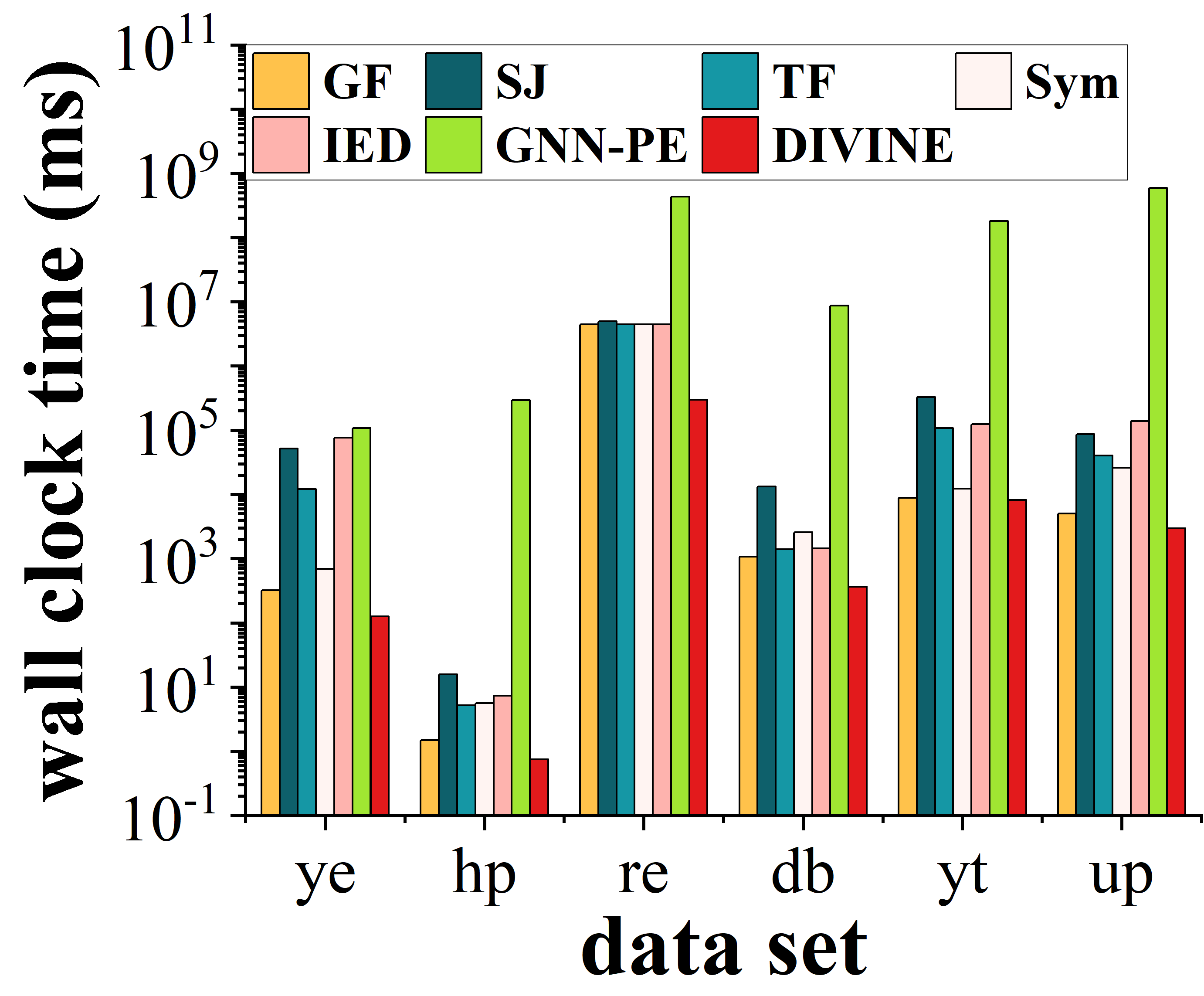}}\label{subfig:ins_real_lc}}
\qquad
\subfigure[][{synthetic graphs}]{
\scalebox{0.15}[0.15]{\includegraphics{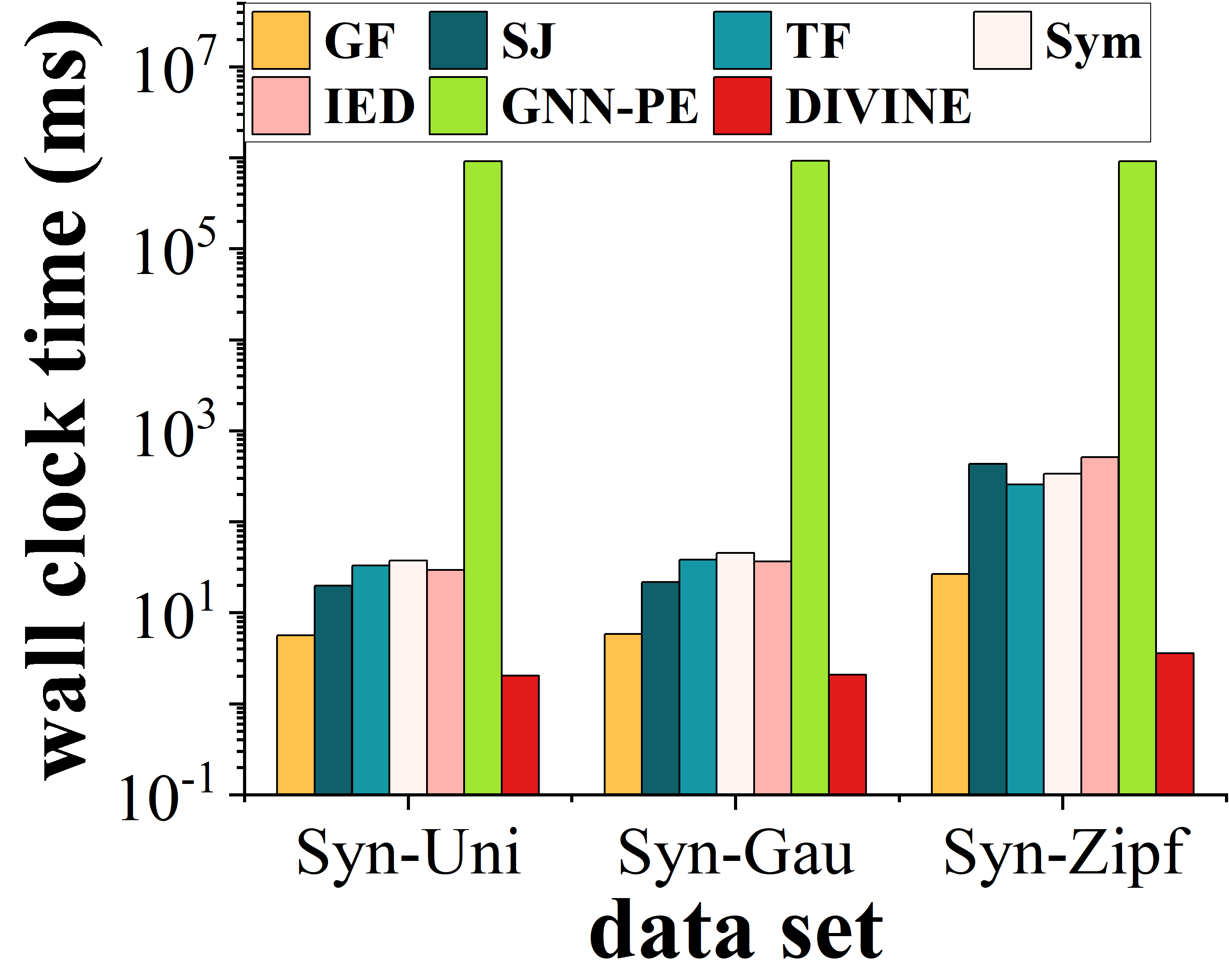}}\label{subfig:ins_syn_lc}}
\caption{The DIVINE efficiency on real/synthetic graphs with vertex label changes, compared with baseline methods.}
\label{fig:dsm_efficiency_vlabel}
\end{figure}

\noindent{\bf The DIVINE Efficiency on Dynamic Graphs with Vertex Label Changes:} Figure \ref{fig:dsm_efficiency_vlabel} evaluates our DIVINE approach over real and synthetic graphs with vertex label changes, where we change 1\% vertex labels upon every $10K$ updates, and default values are used for other parameters. As discussed in Section \ref{subsec:csm_query_answering}, our proposed DIVINE framework can naturally support dynamic vertex label changes by updating the embeddings of the vertex and its 1-hop neighbors. Similar to previous results (as shown in Figures \ref{fig:dsm_vs_graphs} and \ref{fig:dsm_deletion}), our DIVINE approach outperforms the baseline methods for all real/synthetic graphs.

\begin{figure}[t]
\centering
\subfigure[][{\footnotesize real-world graphs ($A_1$ mode)}]{                    
\scalebox{0.15}[0.15]{\includegraphics{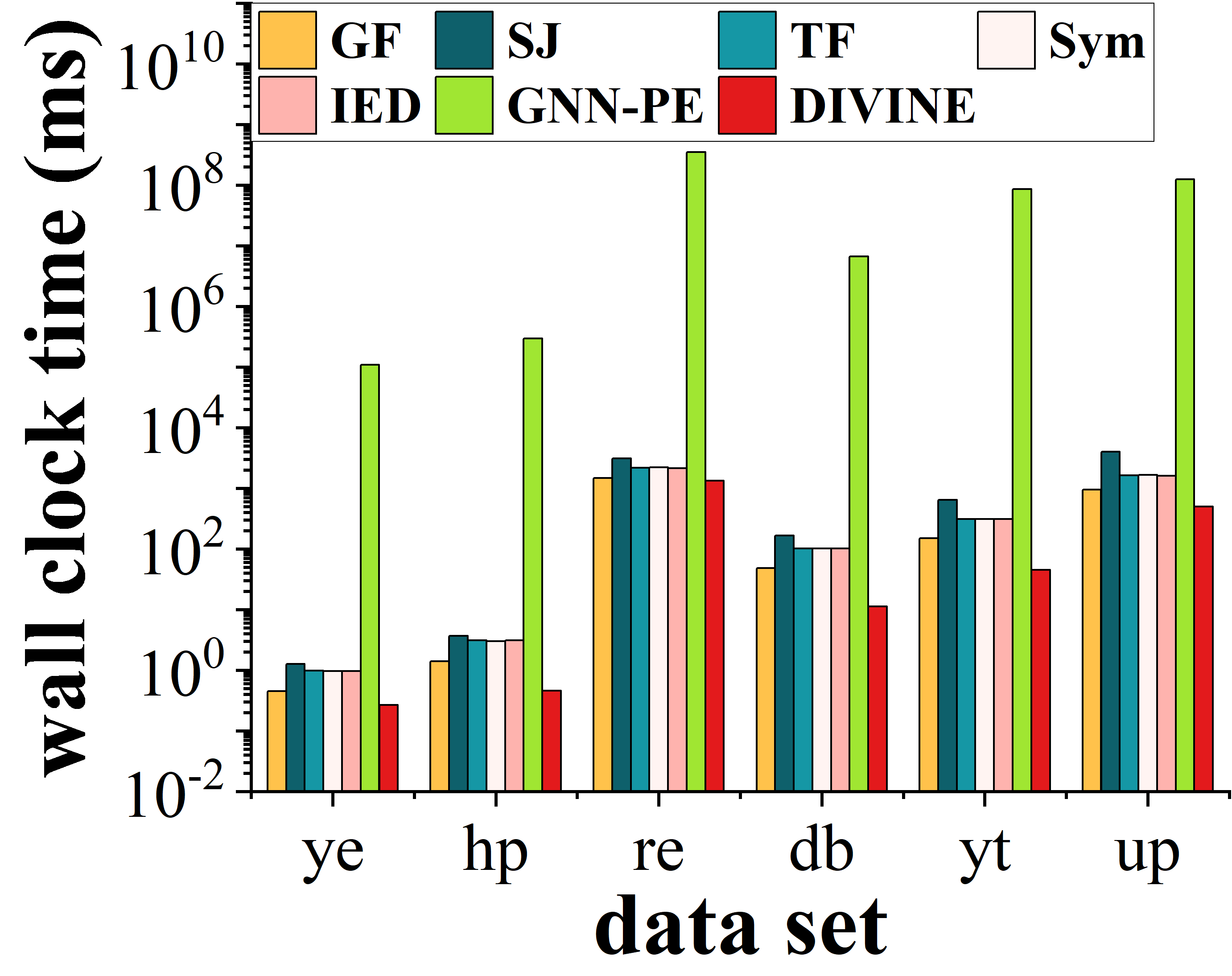}}\label{subfig:dsm_easy_real}}
\qquad
\subfigure[][{\footnotesize synthetic graphs ($A_1$ mode)}]{
\scalebox{0.15}[0.15]{\includegraphics{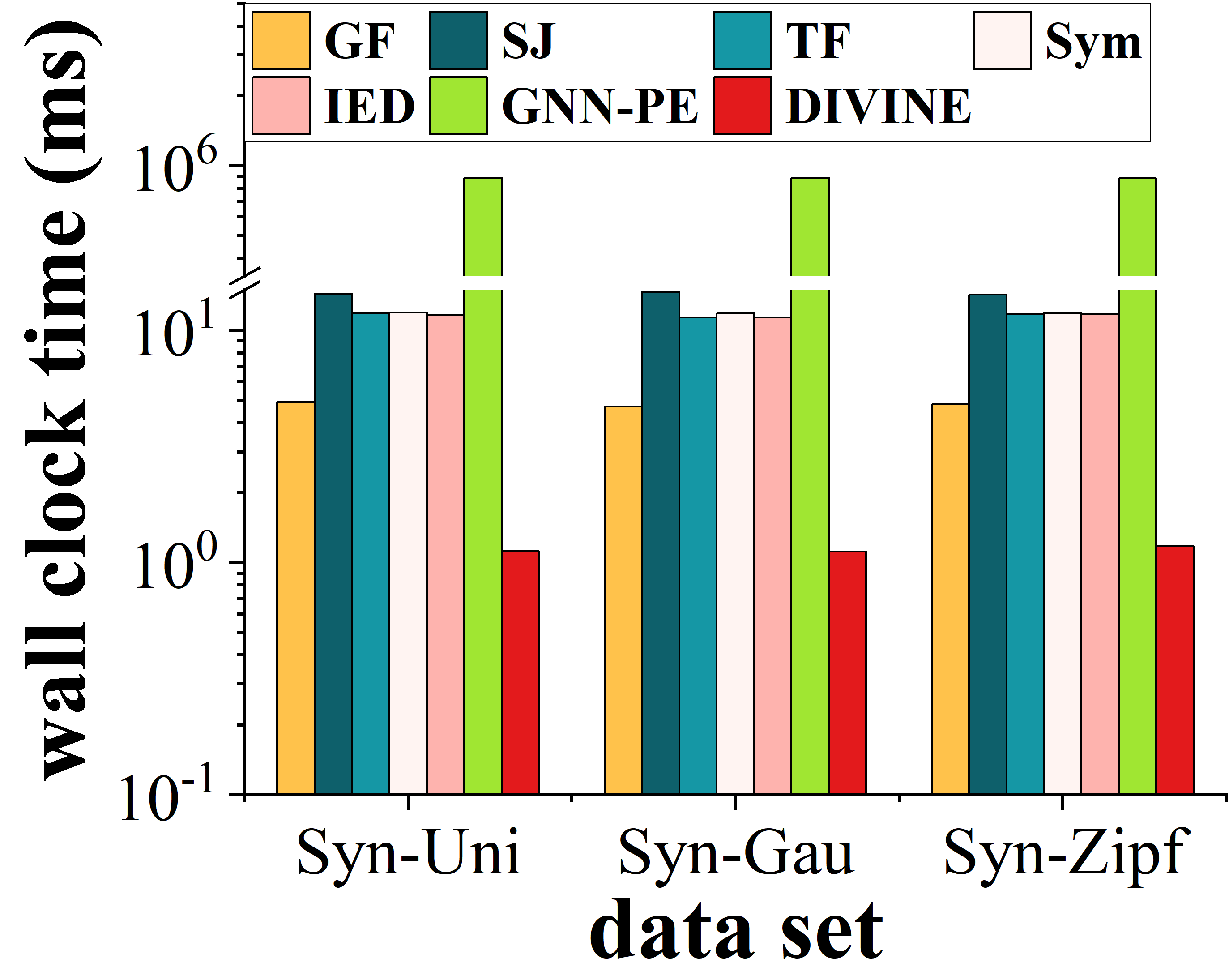}}\label{subfig:dsm_easy_syn}}
\\
\subfigure[][{\footnotesize real-world graphs ($A_2$ mode)}]{                    
\scalebox{0.15}[0.15]{\includegraphics{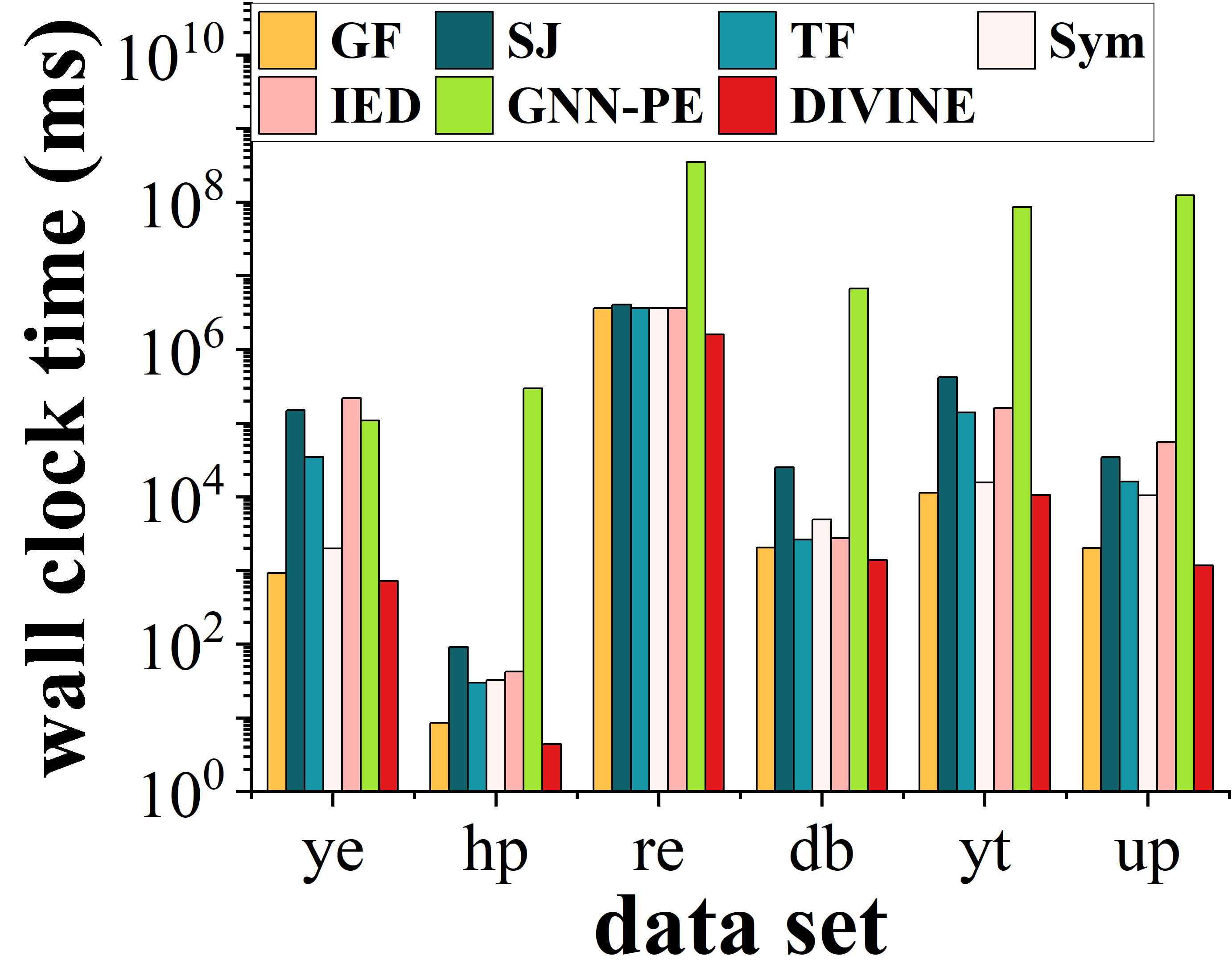}}\label{subfig:dsm_hard_real}}
\qquad
\subfigure[][{\footnotesize synthetic graphs ($A_2$ mode)}]{
\scalebox{0.15}[0.15]{\includegraphics{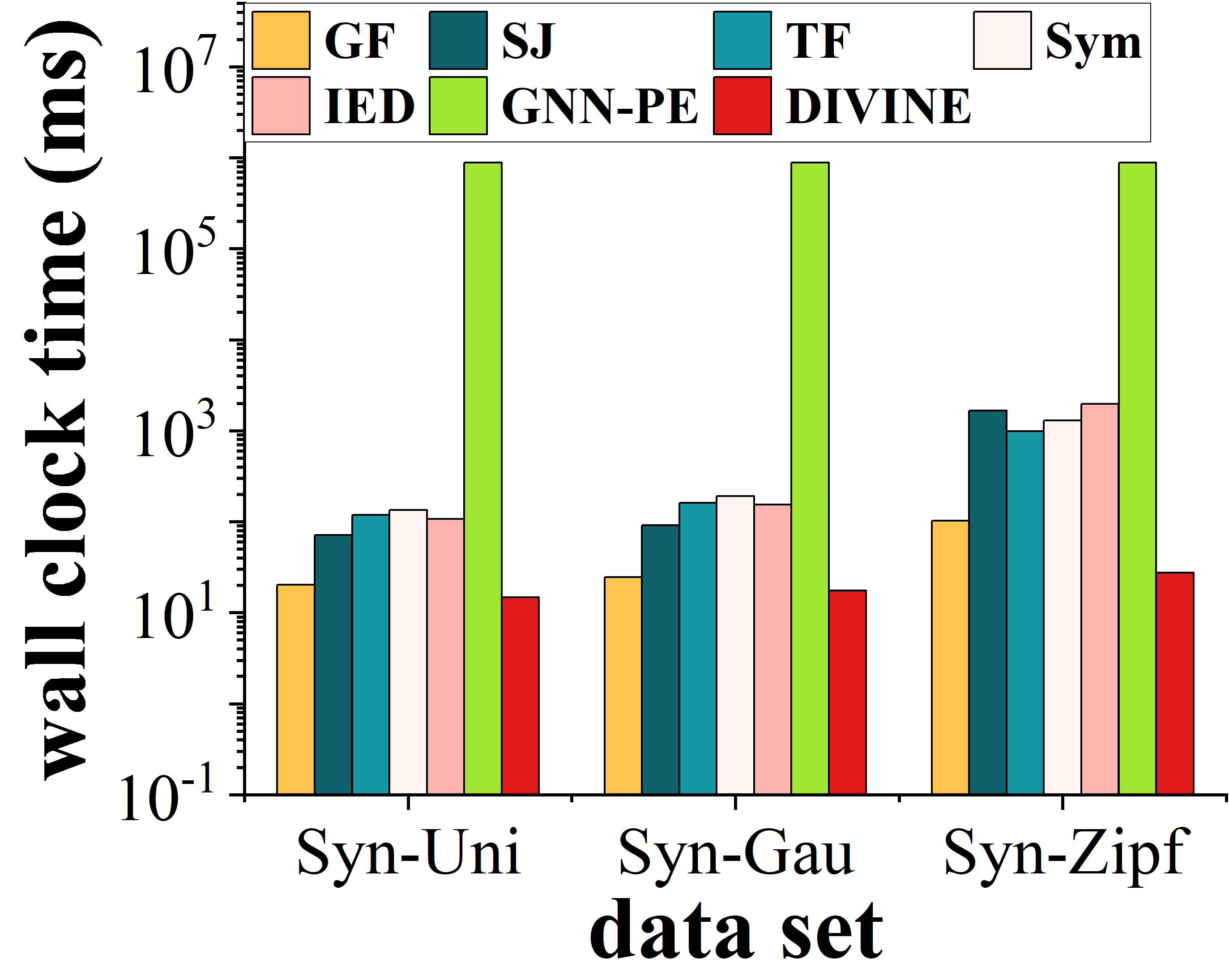}}\label{subfig:dsm_hard_syn}}
\caption{The DIVINE efficiency with adversarial queries ($A_1$ and $A_2$ modes), compared with baseline methods.}
\label{fig:dsm_efficiency_aquery}
\end{figure}

\noindent{\bf The DIVINE Efficiency with Adversarial Queries ($\bm{A_1}$ and $\bm{A_2}$ Modes):}
Figure \ref{fig:dsm_efficiency_aquery} shows our DIVINE query performance with adversarial queries over real/synthetic graphs, which have two modes, ``$A_1$'' and ``$A_2$''. The ``$A_1$'' mode tests 100 query graphs that do not occur in the data graph, whereas the ``$A_2$'' mode evaluates 100 query graphs that lead to large candidate vertex sets. From subfigures, we can see that our DIVINE approach can always outperform baseline methods under both modes.

\begin{figure}[t]
\centering
\subfigure[][{real-world graphs}]{                    
\scalebox{0.15}[0.15]{\includegraphics{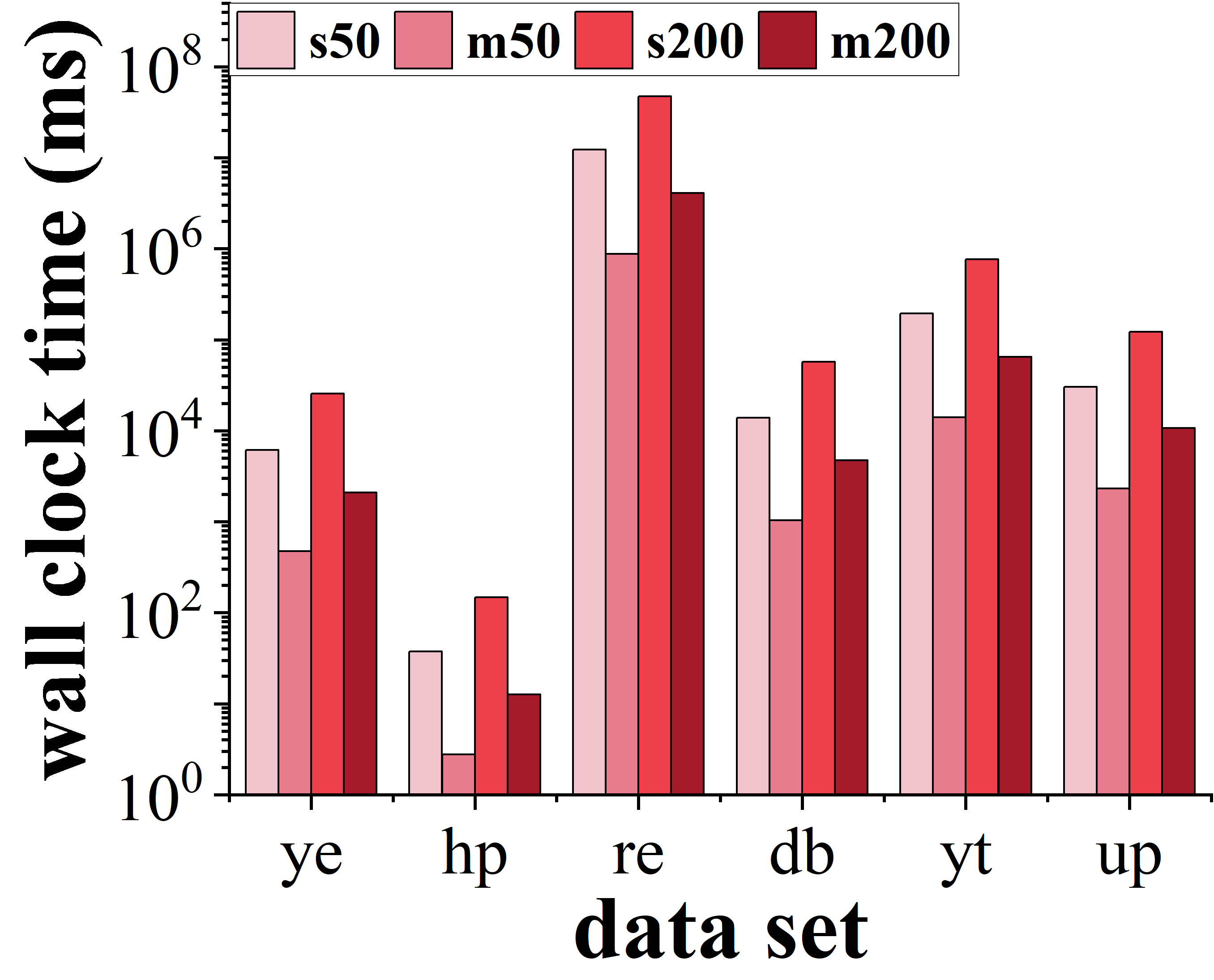}}\label{subfig:ins_real_pp}}
\qquad
\subfigure[][{synthetic graphs}]{
\scalebox{0.15}[0.15]{\includegraphics{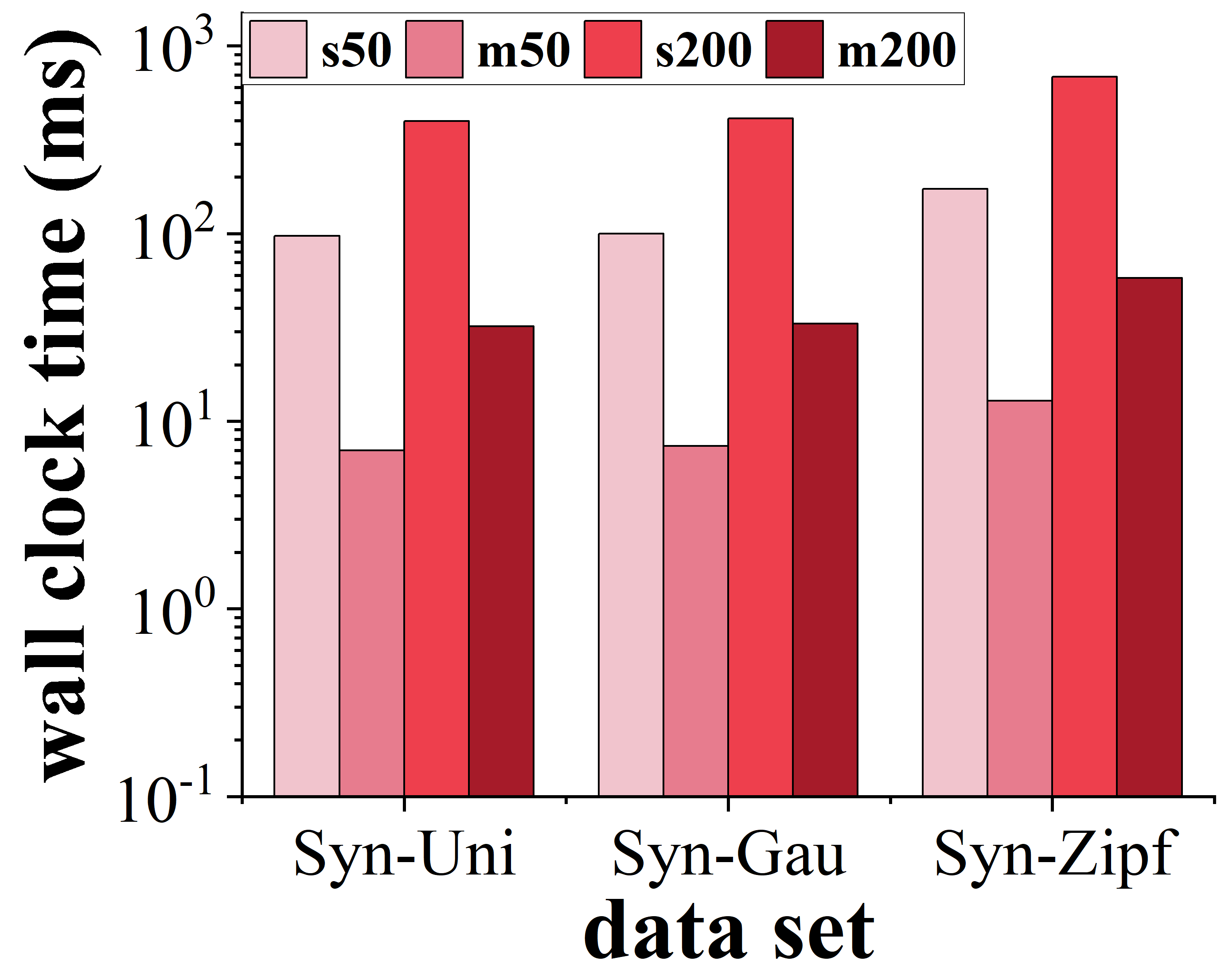}}\label{subfig:ins_syn_pp}}
\caption{The DIVINE efficiency with concurrent queries.}
\label{fig:dsm_efficiency_cquery}
\end{figure}

\noindent{\bf The DIVINE Efficiency with Concurrent Queries:} Figure \ref{fig:dsm_efficiency_cquery} illustrates the performance of our DIVINE approach, by using either a single thread (denoted as ``s'') or multiple (16) threads in parallel (denoted as ``m''), for 50 or 200 concurrent queries, which leads to four scenarios, ``s50'', ``m50'', ``s200'', and ``m200'', where default values are used for other parameters. Here, each thread is responsible for processing a query graph, including query embedding generation, synopsis search, refinement, and answer set update, such that CSM queries are independently answered in parallel.
From subfigures, we can see that concurrent CSM queries can achieve lower time cost with multithreading for our DIVINE  
approach (instead of using a single thread), confirming the effectiveness of our DAS$^3$ synopses to support parallel processing.

\begin{figure}[t]
\centering
\subfigure[][{real-world graphs}]{                    
\scalebox{0.15}[0.15]{\includegraphics{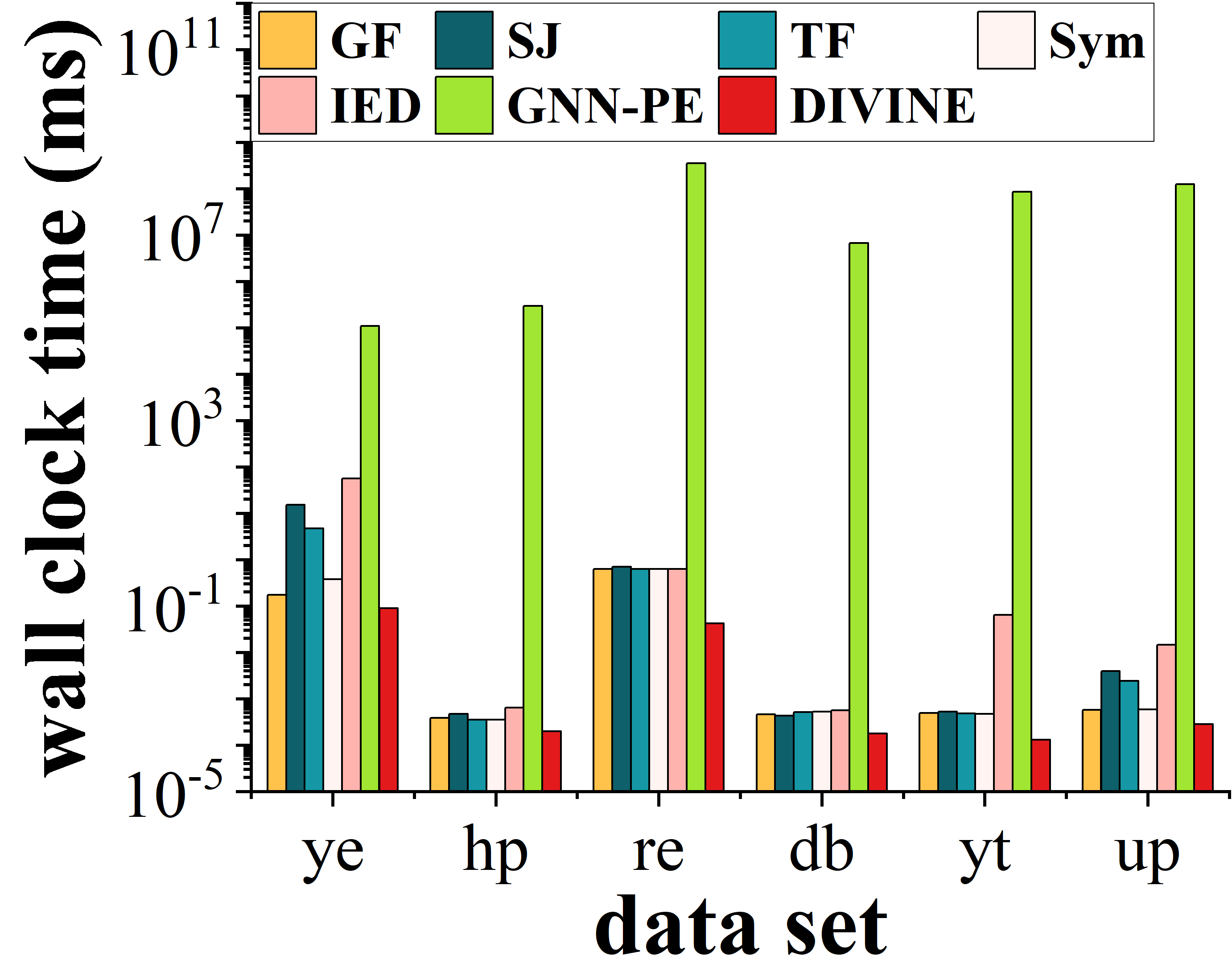}}\label{subfig:dsm_edge_real}}
\qquad
\subfigure[][{synthetic graphs}]{
\scalebox{0.15}[0.15]{\includegraphics{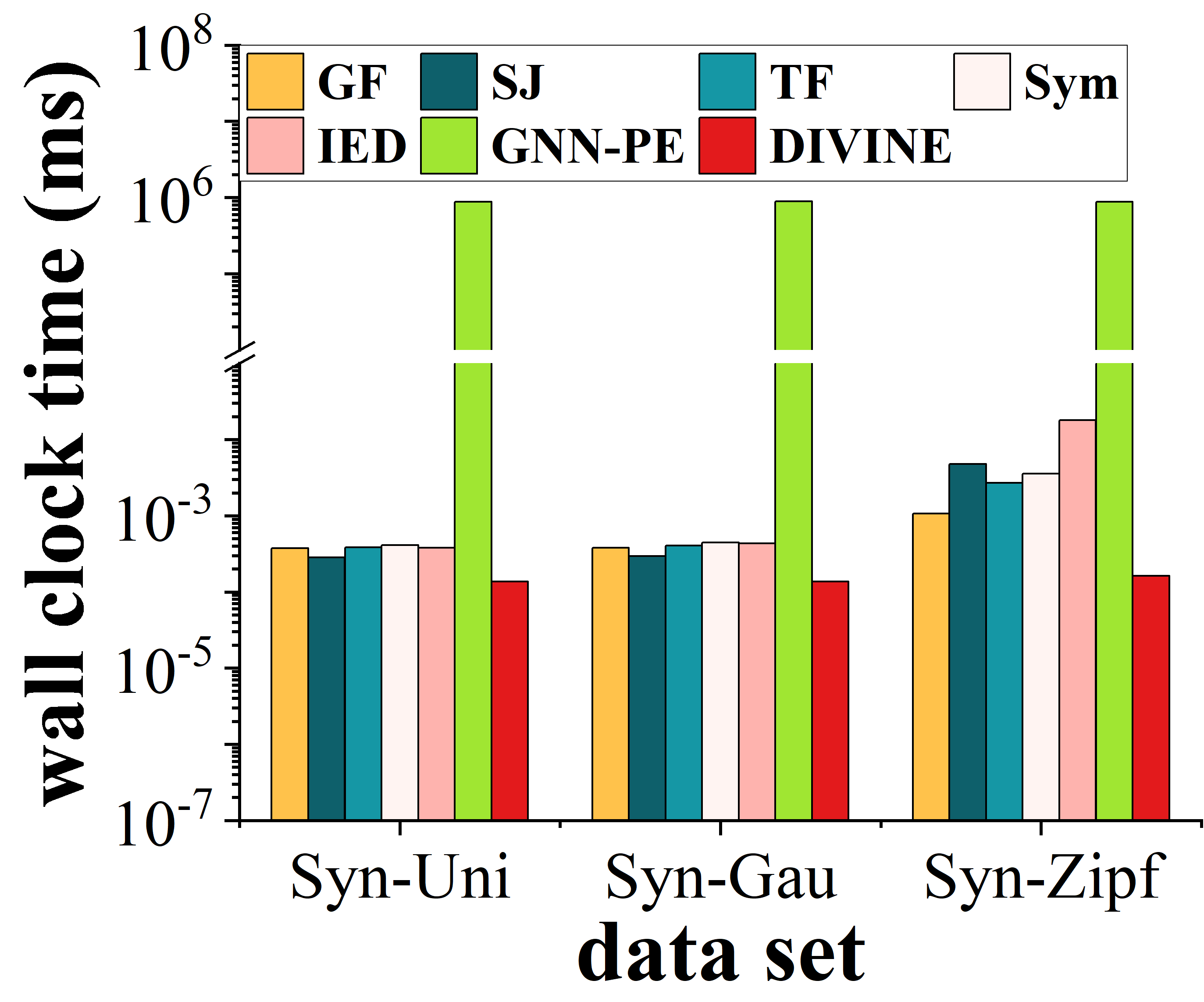}}\label{subfig:dsm_edge_syn}}
\caption{The DIVINE efficiency of per-edge update, compared with baseline methods.}
\label{fig:dsm_efficiency_edge}
\end{figure}

\noindent{\bf The DIVINE Efficiency of Per-Edge Update:}
Figure \ref{fig:dsm_efficiency_edge} illustrates our DIVINE per-edge update (insertion) cost, compared with baselines, where all parameters are set by default. From subfigures, we can observe that the per-edge update cost of our DIVINE approach is consistently lower than that of baselines over real/synthetic graphs, and remains low (i.e., 0.13 $\sim$ 89.11 $\mu s$).

To evaluate the robustness of our DIVINE approach, in the sequel, we test different scenarios, such as dynamic graph with edge deletions, vertex label changes, adversarial queries, and concurrent queries. 
We also vary parameter values on synthetic graphs (e.g., $|\Sigma|$, $avg\_deg(q)$, $|V(q)|$, $avg\_deg(G_D)$, and $|V(G_D)|$), and omit the results of baselines to better illustrate the trends of curves.

\begin{figure}[t]
\centering
\subfigure[][{\small DIVINE vs. $|\Sigma|$}]{                    
\scalebox{0.15}[0.15]{\includegraphics{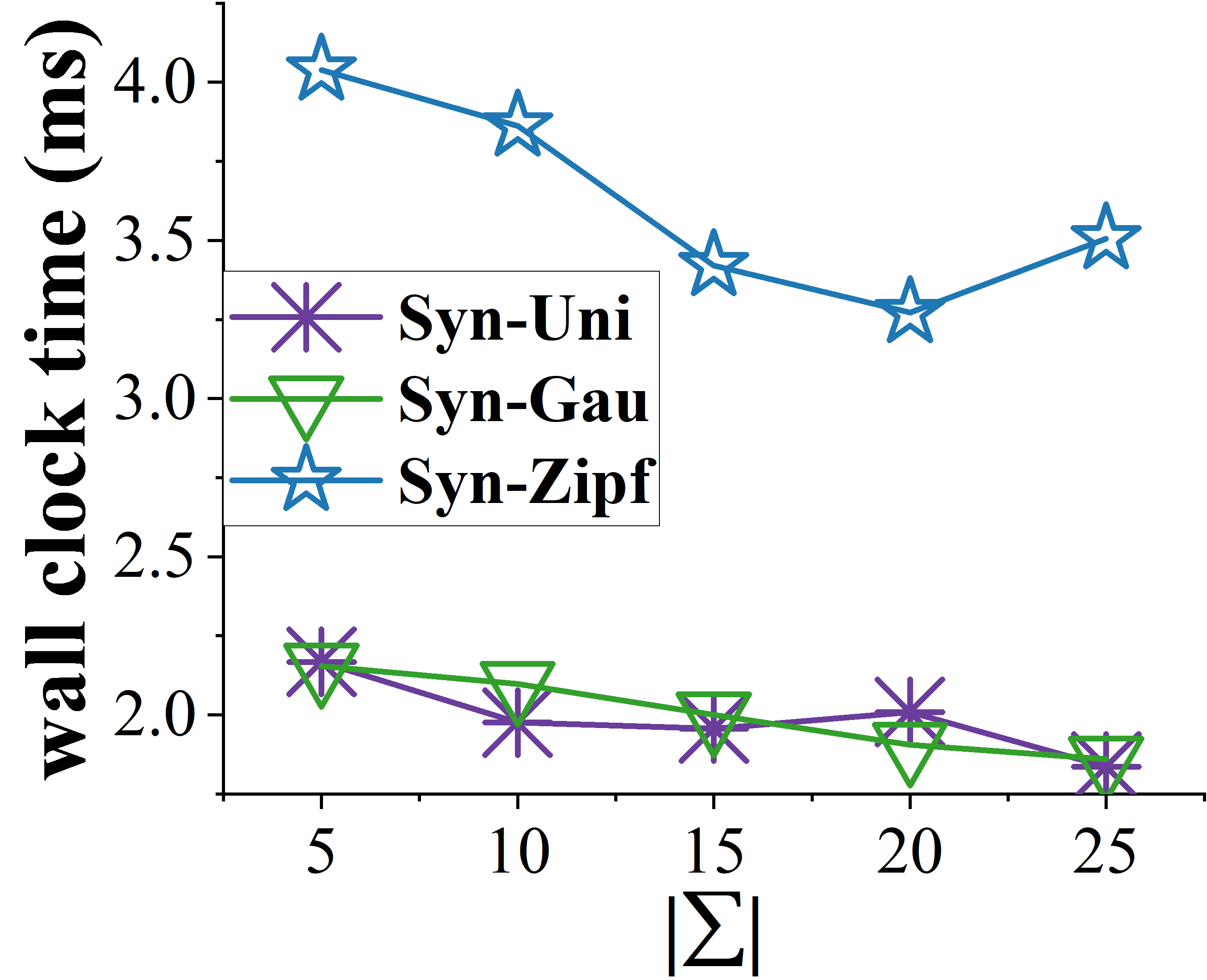}}\label{subfig:label_syn}}
\qquad
\subfigure[][{\small DIVINE vs. $avg\_deg(q)$}]{                    
\scalebox{0.15}[0.15]{\includegraphics{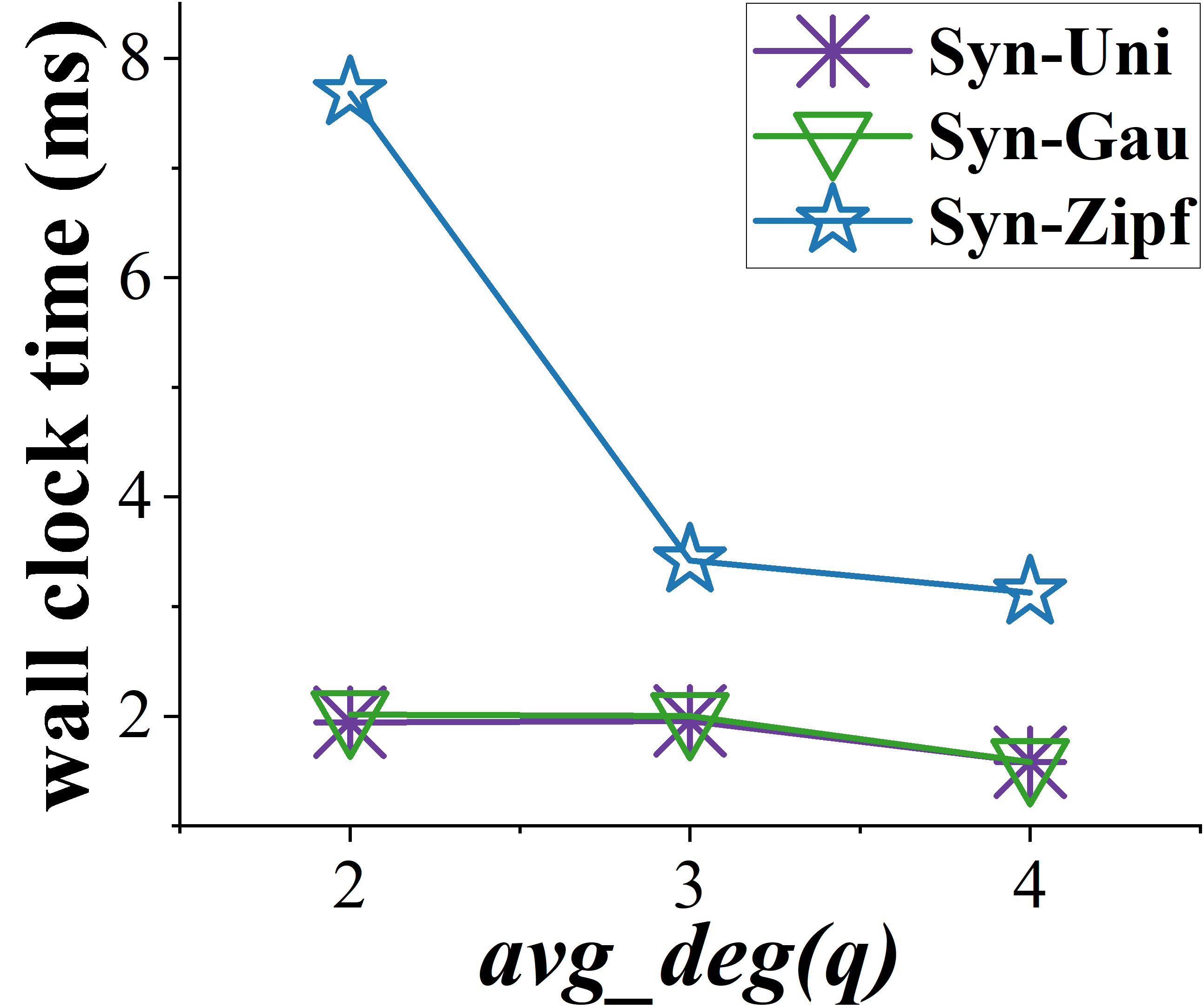}}\label{subfig:qdeg_syn}}
\qquad
\subfigure[][{\small DIVINE vs. $|V(q)|$}]{                    
\scalebox{0.15}[0.15]{\includegraphics{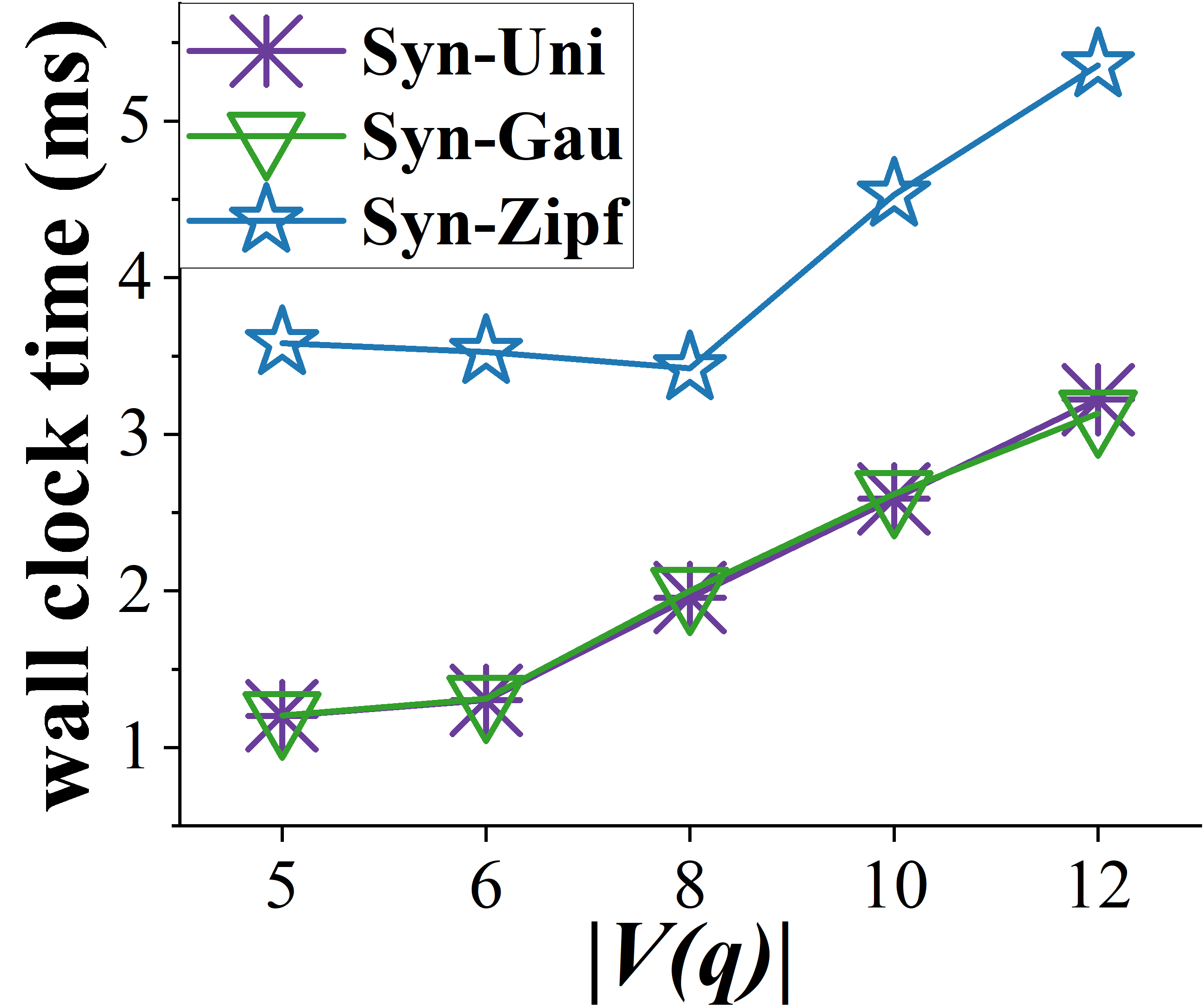}}\label{subfig:qsize_syn}}
\\
\subfigure[][{\footnotesize DIVINE vs. $avg\_deg(G_D)$}]{                    
\scalebox{0.15}[0.15]{\includegraphics{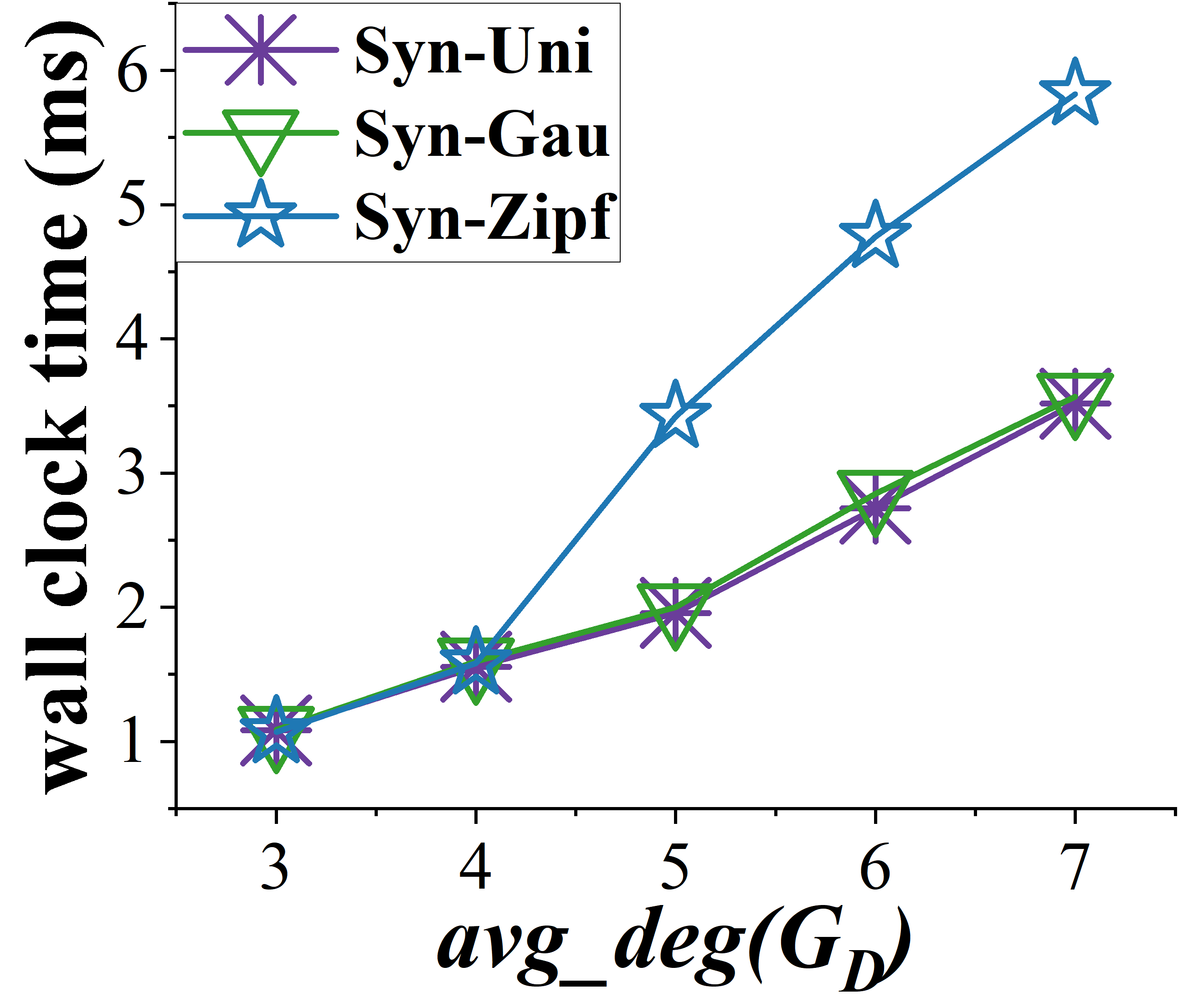}}\label{subfig:ddeg_syn}}
\qquad
\subfigure[][\small DIVINE vs. $|V(G_D)|$]{                    
\scalebox{0.15}[0.15]{\includegraphics{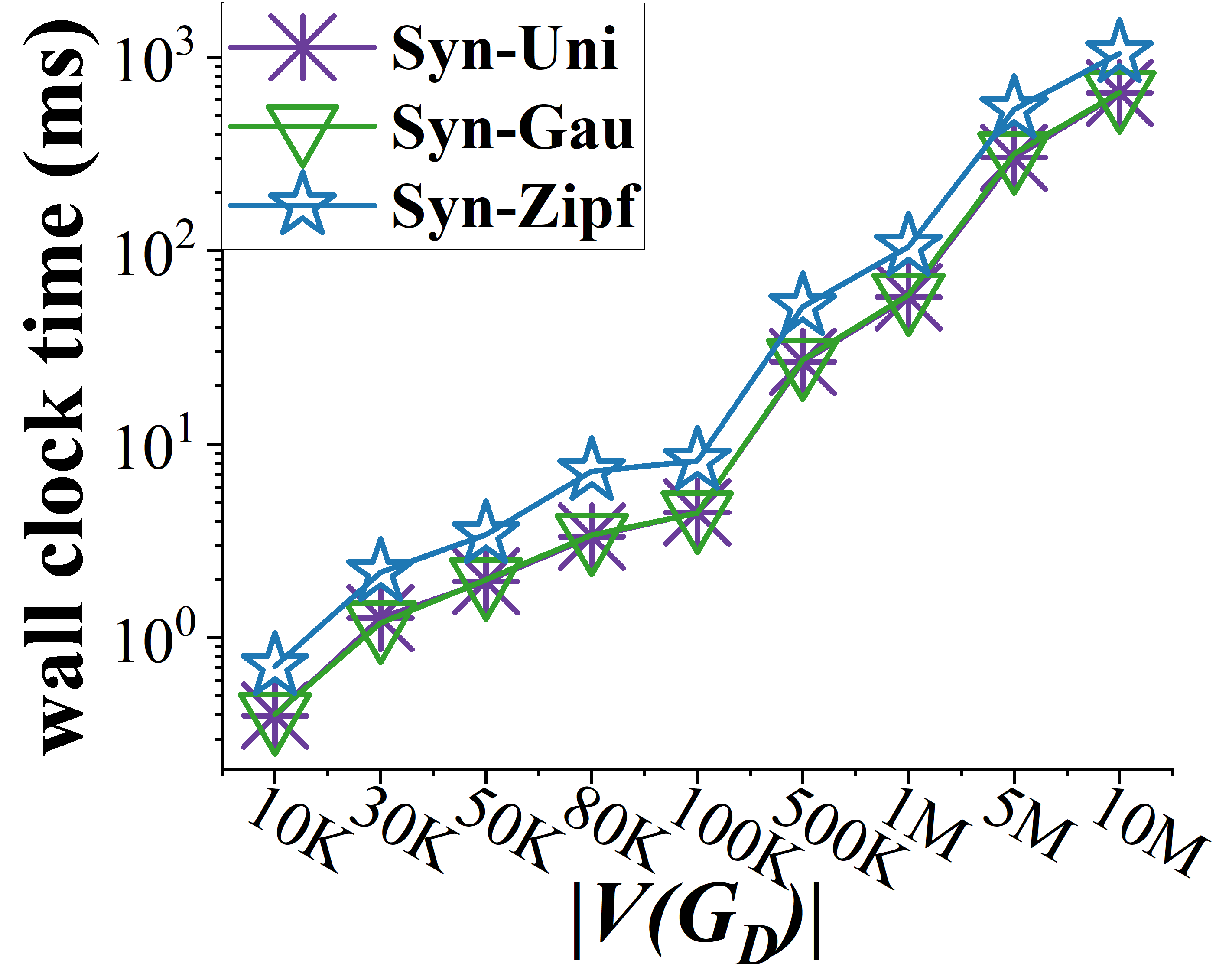}}\label{subfig:dsize_syn}}
\caption{The DIVINE efficiency w.r.t. parameters $|\Sigma|$, $avg\_deg(q)$, $|V(q)|$, $avg\_deg(G_D)$, and $|V(G_D)|$.}
\label{fig:label_syn}
\end{figure}

\noindent {\bf The DIVINE Efficiency w.r.t. $\bm{\#}$ of Distinct Vertex Labels, $\bm{|\Sigma|}$:}
Figure~\ref{subfig:label_syn} shows the wall clock time of our DIVINE approach, where $|\Sigma|$ varies from 5 to 25, and other parameters are set to default values. When the number, $|\Sigma|$, of distinct vertex labels increases, the pruning power also increases (i.e., with fewer candidate vertices). Moreover, the query cost is also affected by vertex label distributions. Overall, the DIVINE query cost remains low for different $|\Sigma|$ values (i.e., 1.83 $\sim$ 4.04 $ms$).

\noindent {\bf The DIVINE Efficiency w.r.t. Average Degree, $\bm{avg\_deg(q)}$, of the Query Graph $\bm{q}$:}
Figure~\ref{subfig:qdeg_syn} examines the DIVINE performance by varying the average degree, $avg\_deg(q)$, of the query graph $q$ from 2 to 4, where other parameters are set to default values. 
Higher degree $avg\_deg(q)$ of $q$ incurs higher pruning power of query vertices. Therefore, when $avg\_deg(q)$ increases, the DIVINE time cost decreases. For different $avg\_deg(q)$ values, the DIVINE query cost remains low (i.e., 1.58 $\sim$ 7.68 $ms$).

\noindent {\bf The DIVINE Efficiency w.r.t. Query Graph Size $\bm{|V(q)|}$:} Figure~\ref{subfig:qsize_syn} illustrates our DIVINE performance, by varying the query graph size, $|V(q)|$, from 5 to 12, where default values are used for other parameters. When the number, $|V(q)|$, of vertices in query graph $q$ increases, fewer candidate subgraphs are expected to match with larger query graph $q$. On the other hand, larger query graph size $|V(q)|$ will cause higher query costs in finding candidates for more query vertices, through synopsis traversal and refinement. Therefore, the query time is influenced by these two factors. Nevertheless, the time cost remains low for different query graph sizes (i.e., $<$$5.36$$ms$).

\noindent {\bf The DIVINE Efficiency w.r.t. Avg. Degree, $\bm{avg\_deg(G_D)}$, of Dynamic Graph $\bm{G_D}$:}
Figure~\ref{subfig:ddeg_syn} presents our DIVINE performance with different average degrees, $avg\_deg$ $(G_D)$, of dynamic graph $G_D$, where $avg\_deg(G_D)=3 \sim 7$, and default values are used for other parameters. Intuitively, a higher degree $avg\_deg(G_D)$ in the data graph $G_D$ incurs lower pruning power and more candidate vertices. Thus, when $avg\_deg(G_D)$ becomes higher, the wall clock time also increases, especially for $Syn\text{-}Zipf$ (due to its skewed vertex label distribution). Nevertheless, the DIVINE query time remains small for different $avg\_deg(G_D)$ values (i.e., 1.08 $\sim$ 5.82 $ms$).

\noindent {\bf The DIVINE Scalability Test w.r.t. Dynamic Graph Size $\bm{|V(G_D)|}$:}
Figure~\ref{subfig:dsize_syn} tests the scalability of our DIVINE approach with different dynamic graph sizes, $|V(G_D)|$, from $10K$ to $10M$, where default parameter values are used. A larger dynamic graph incurs more matching candidate vertices (and, in turn, candidate subgraphs). Thus, the time cost of our DIVINE approach increases with the increase of graph size $|V(G_D)|$, nonetheless, remains the lowest (i.e., 0.4 $\sim$ 1.05 $sec$ for graph sizes up to $10M$), which confirms the scalability of our DIVINE approach for large graph sizes.

\subsection{DAS$^3$ Synopsis Initialization Cost}
In this subsection, we report the memory cost of our DIVINE approach over real/synthetic graphs.

\begin{figure}[t]
\centering
\subfigure[][{real-world graphs}]{                    
\scalebox{0.15}[0.15]{\includegraphics{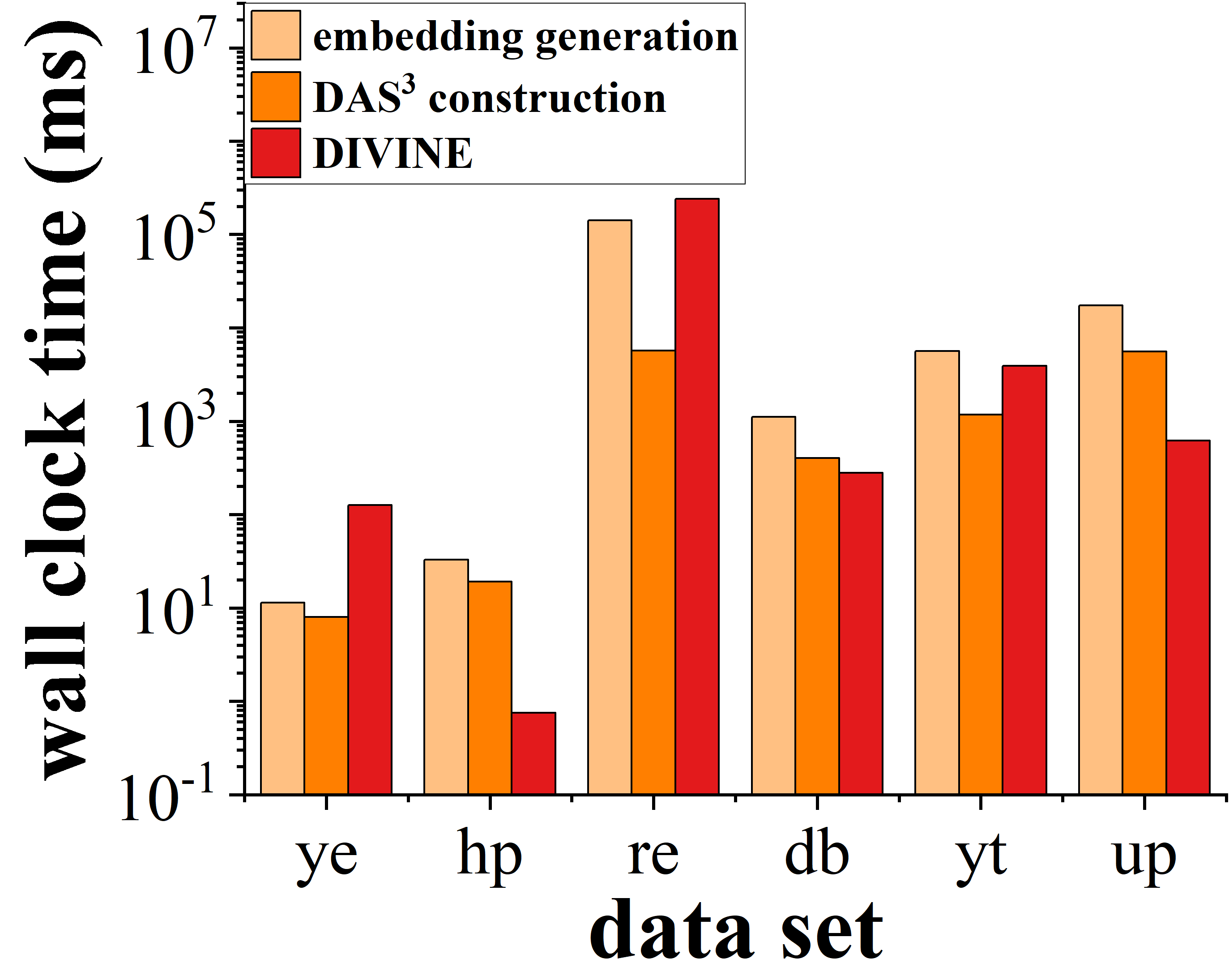}}\label{subfig:dsm_cost_real}}
\qquad
\subfigure[][{synthetic graphs}]{
\scalebox{0.15}[0.15]{\includegraphics{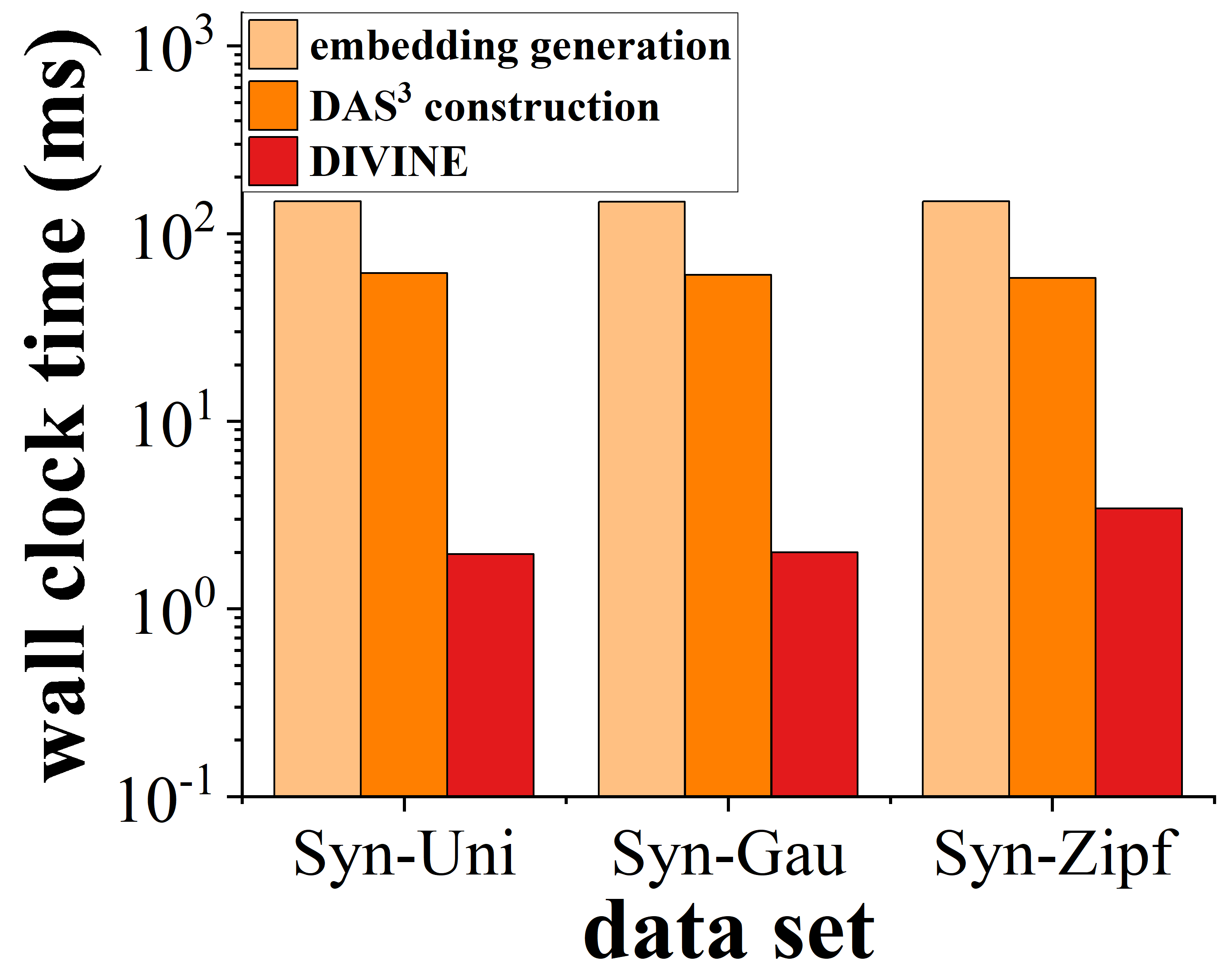}}\label{subfig:dsm_cost_syn}}
\caption{The comparison analysis of the DIVINE offline pre-computation and online query costs on real/synthetic graphs.}
\label{fig:dsm_cost}
\end{figure}

\noindent{\bf The DAS$\bm{^3}$ Synopsis Initialization Cost on Real/Synthetic Graphs.}
We compare the DAS$^3$ synopsis initialization cost of our DIVINE approach (including time costs of vertex dominance embedding generation and DAS$^3$ construction over vertex dominance embeddings) with online query time over real/synthetic graphs, where parameters are set to default values. 
In Figure~\ref{fig:dsm_cost}, for graph sizes from 3K to 3.77M, the overall offline pre-computation time varies from 19 $ms$ $\sim$ 147.52 $sec$. Specifically, the time costs of embedding generation and DAS$^3$ construction are 11 $ms$ $\sim$ 141.85 $sec$, 8 $ms$ $\sim$ 5.67 $sec$, respectively.
On the other hand, since the online query time includes the time cost of DAS$^3$ and embedding update, the maintenance cost of our DAS$^3$ is low.

\begin{figure}[t]
\centering
\subfigure[][{real-world graphs}]{                    
\scalebox{0.16}[0.15]{\includegraphics{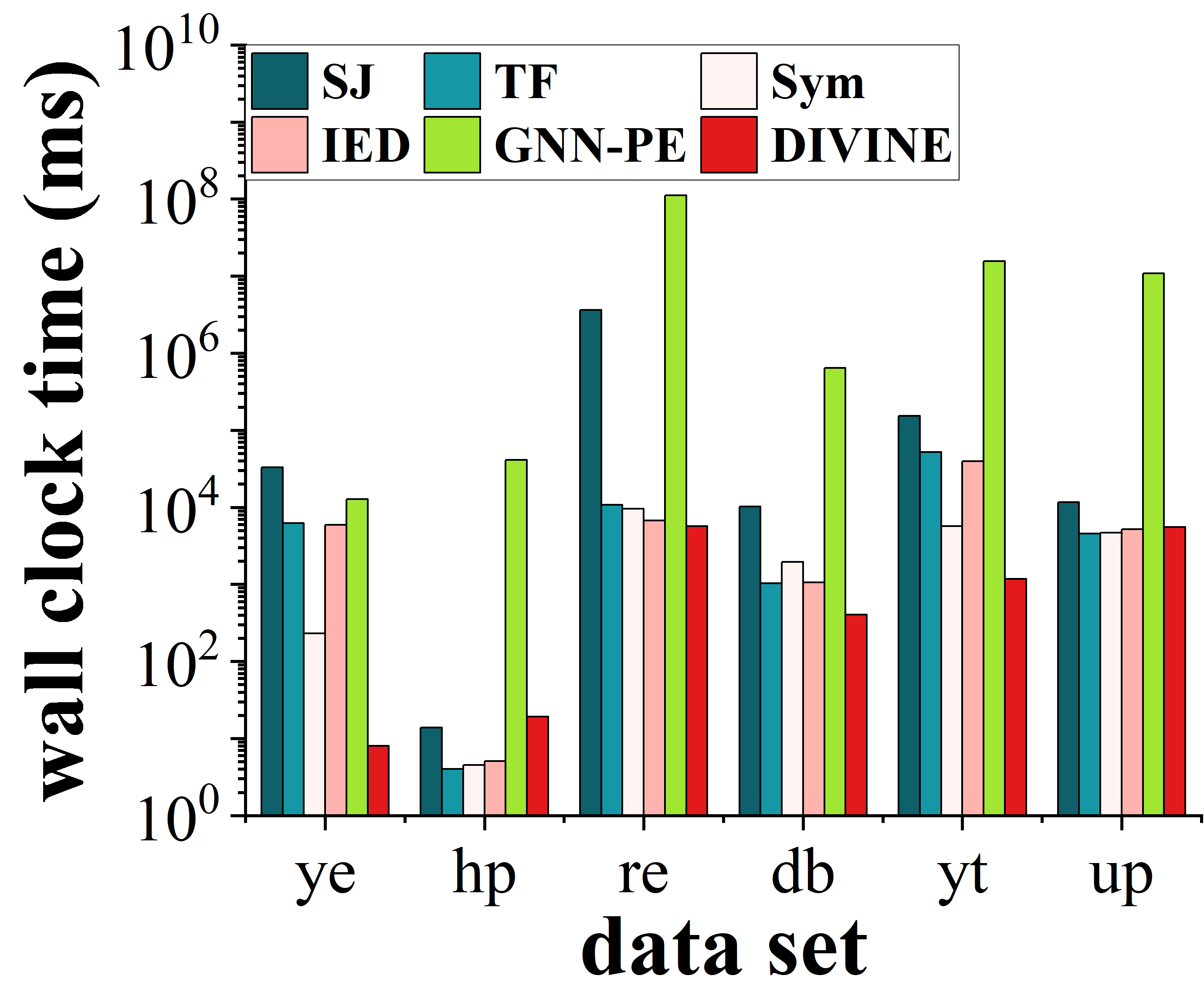}}\label{subfig:time_cost_real}}
\qquad
\subfigure[][{synthetic graphs}]{
\scalebox{0.15}[0.15]{\includegraphics{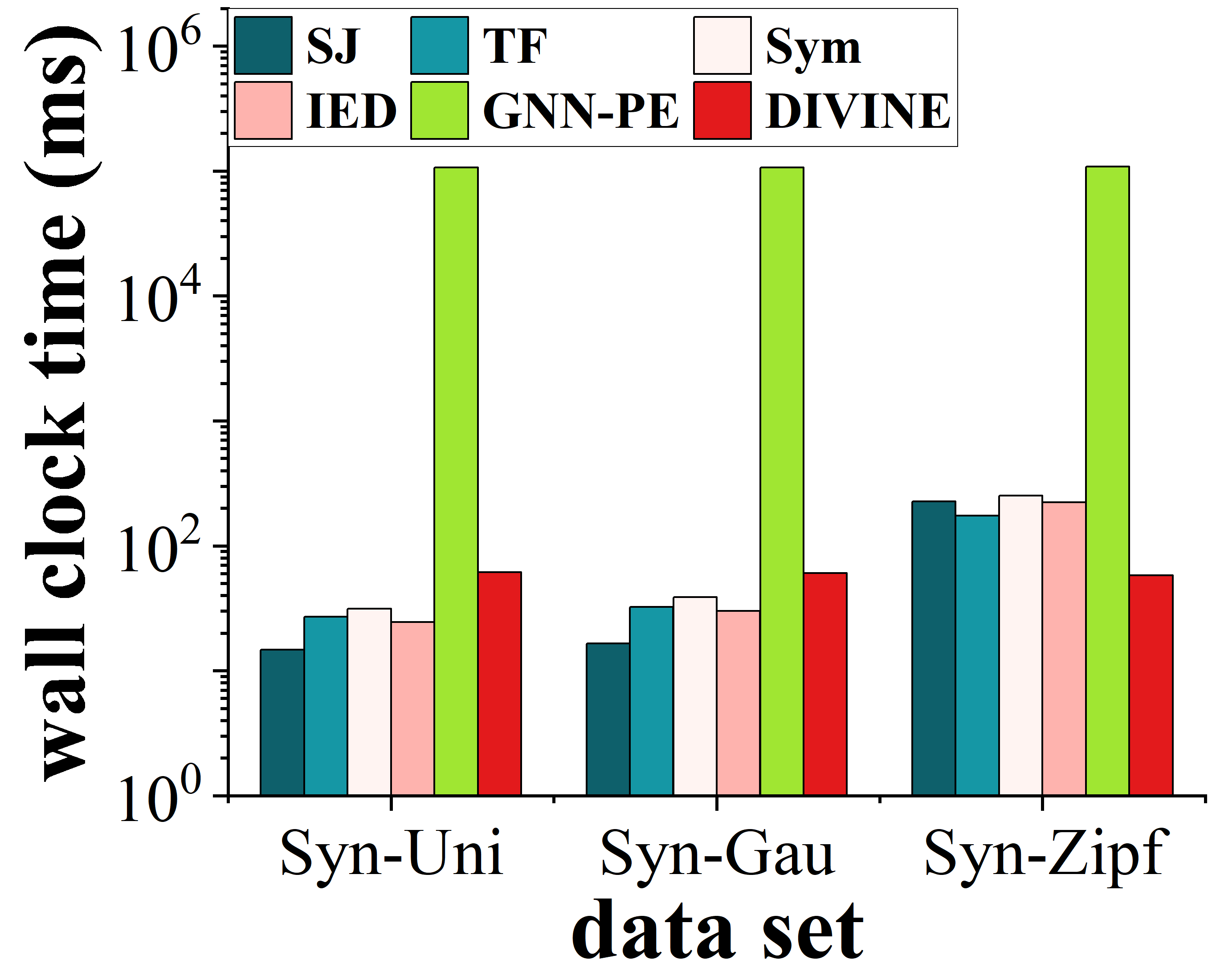}}\label{subfig:time_cost_syn}}
\caption{DAS$^3$ construction time on real/synthetic graphs, compared with baseline methods.}
\label{fig:time_cost}
\end{figure}

\noindent{\bf The DAS$\bm{^3}$ Index Time Cost on Real/Synthetic Graphs.}
Since GF directly enumerates matching answers over the original data graph without any auxiliary data structure \cite{kankanamge2017graphflow,sun2022depth}, Figure~\ref{fig:time_cost} compares the one-time construction cost of our DAS$^3$ synopses with baselines, where parameters are set to their default values.
From the figure, we can see that the time cost of our DAS$^3$ synopses is comparable to that of baselines.
Moreover, unlike baseline methods that construct an index over increasingly incoming (registered) query graphs $q$ (which may not be scalable), our DAS$^3$ construction is \textit{one-time only} over the initial data graph $G_0$. Thus, our DAS$^3$ synopses can be used to accelerate numerous online CSM query requests from users simultaneously with high throughput.

\begin{figure}[t]
\centering
\subfigure[][{real-world graphs}]{                    
\scalebox{0.15}[0.15]{\includegraphics{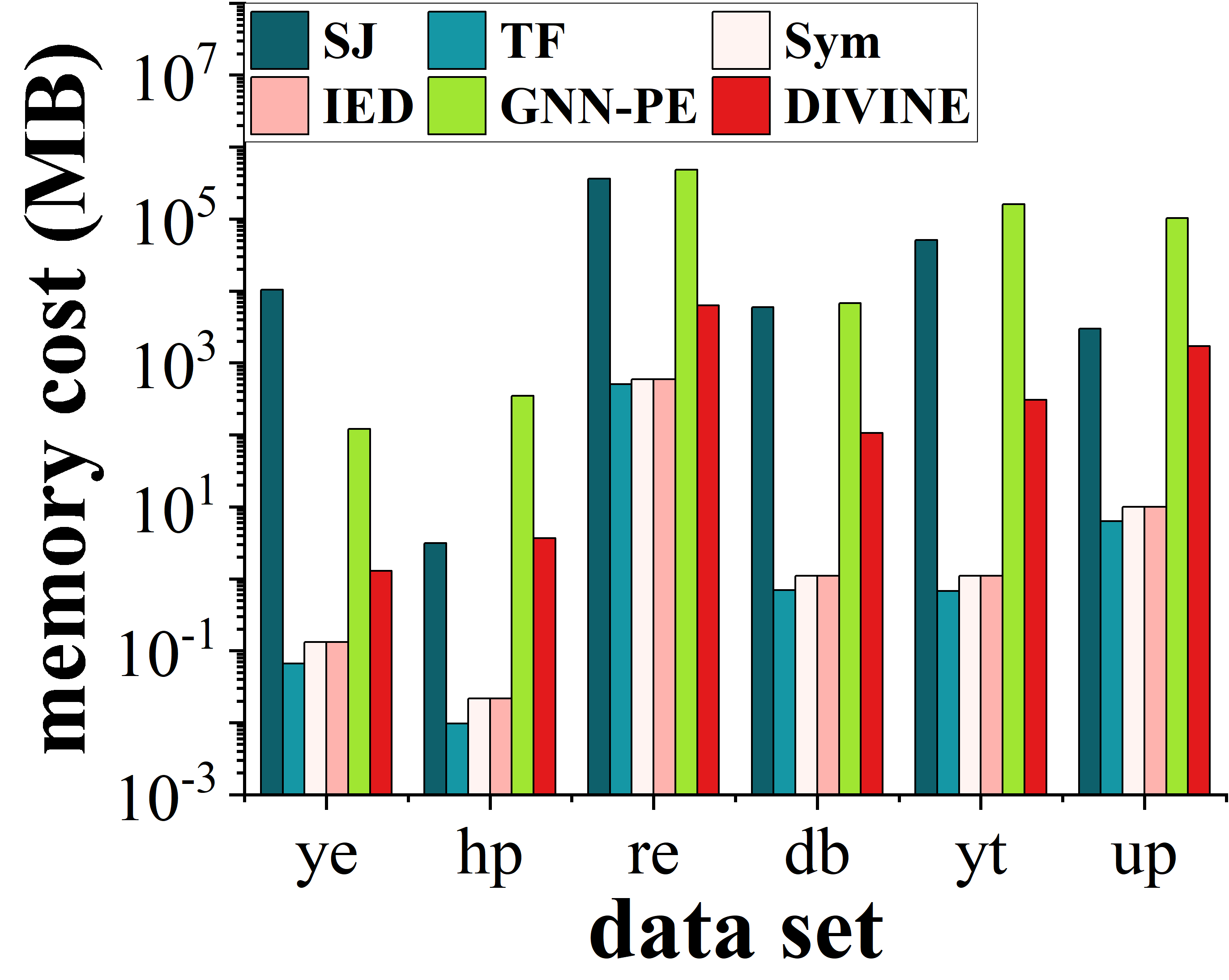}}\label{subfig:storage_cost_real}}
\qquad
\subfigure[][{synthetic graphs}]{
\scalebox{0.15}[0.15]{\includegraphics{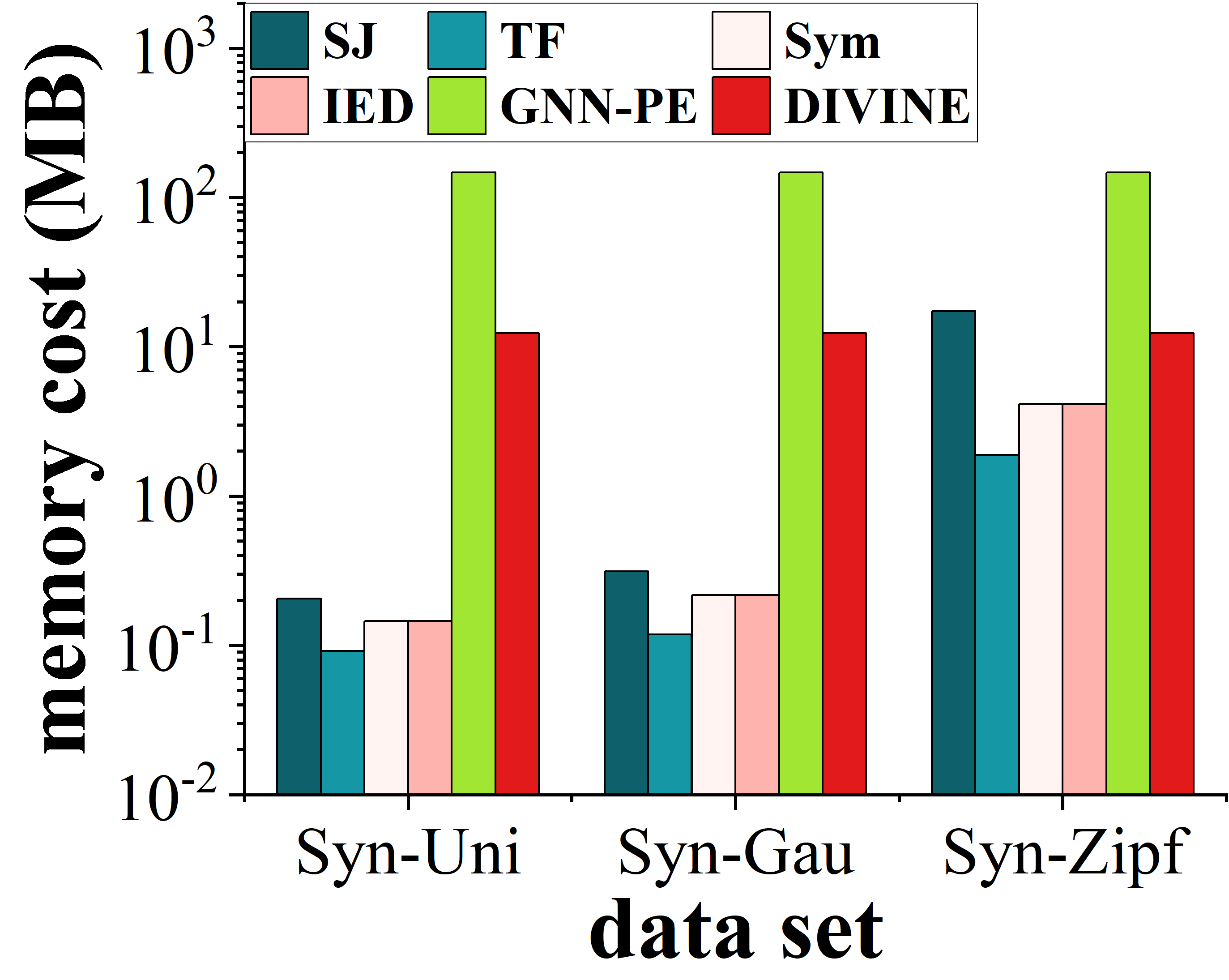}}\label{subfig:storage_cost_syn}}
\caption{The DAS$^3$ memory cost on real/synthetic graphs, compared with baseline methods.}
\label{fig:storage_cost}
\end{figure}

\noindent{\bf The DAS$\bm{^3}$ Space Cost on Real/Synthetic Graphs.}
Since GF directly enumerates matching answers over raw data graph without any auxiliary data structure \cite{kankanamge2017graphflow,sun2022depth}, Figure~\ref{fig:storage_cost} compares the storage cost of our DAS$^3$ with that of the other 5 baselines, where parameters are set to default values. Our DIVINE approach always has lower space consumption of DAS$^3$ synopses than GNN-PE over real/synthetic graphs, as well as SJ \cite{choudhury2015selectivity,sun2022depth} on most graphs. Moreover, DIVINE needs a bit higher space cost than other baselines. Nevertheless, different from these baselines that construct an index over more and more incoming (registered) query graphs $q$ (which may not be scalable), our DAS$^3$ construction is \textit{one-time only} over the initial graph $G_0$. Thus, our DAS$^3$ synopses can be used to accelerate numerous online CSM query requests from users simultaneously with high throughput.

\begin{figure}[t]
\centering
\subfigure[][{$Syn\text{-}Uni$}]{                    
\scalebox{0.15}[0.15]{\includegraphics{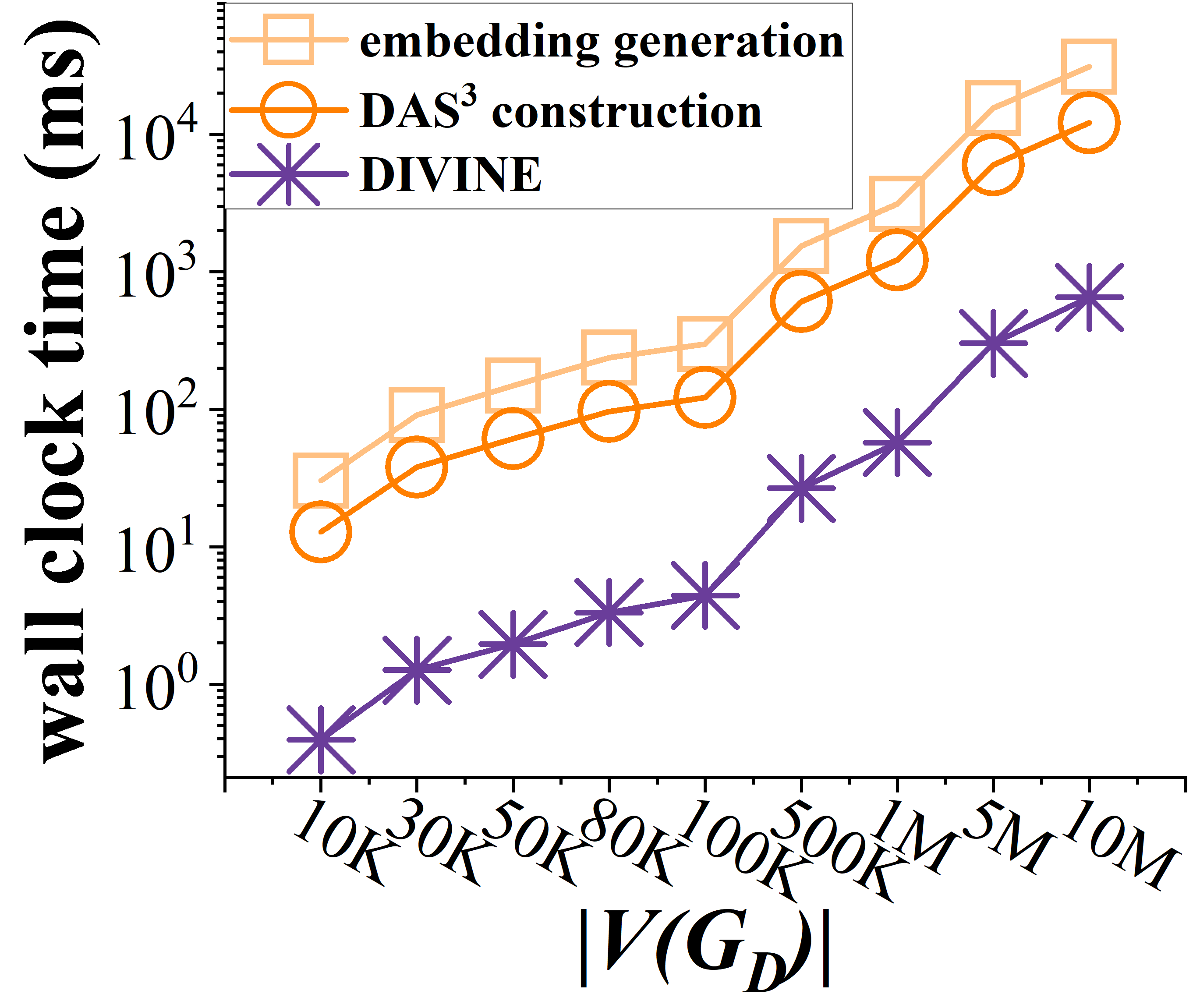}}\label{subfig:precost_uni}}
\qquad
\subfigure[][{$Syn\text{-}Gau$}]{
\scalebox{0.15}[0.15]{\includegraphics{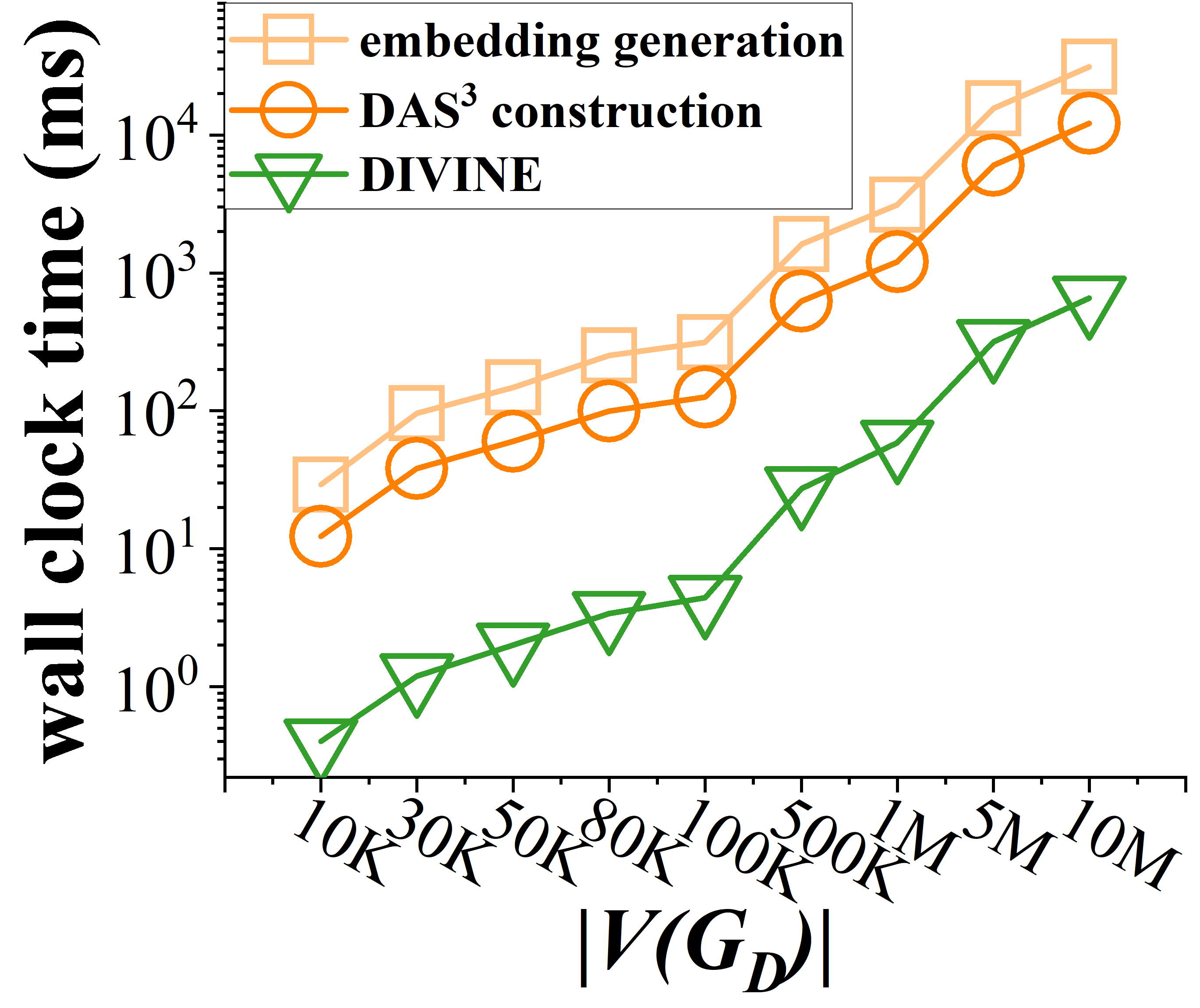}}\label{subfig:precost_gau}}
\qquad
\subfigure[][{$Syn\text{-}Zipf$}]{
\scalebox{0.15}[0.15]{\includegraphics{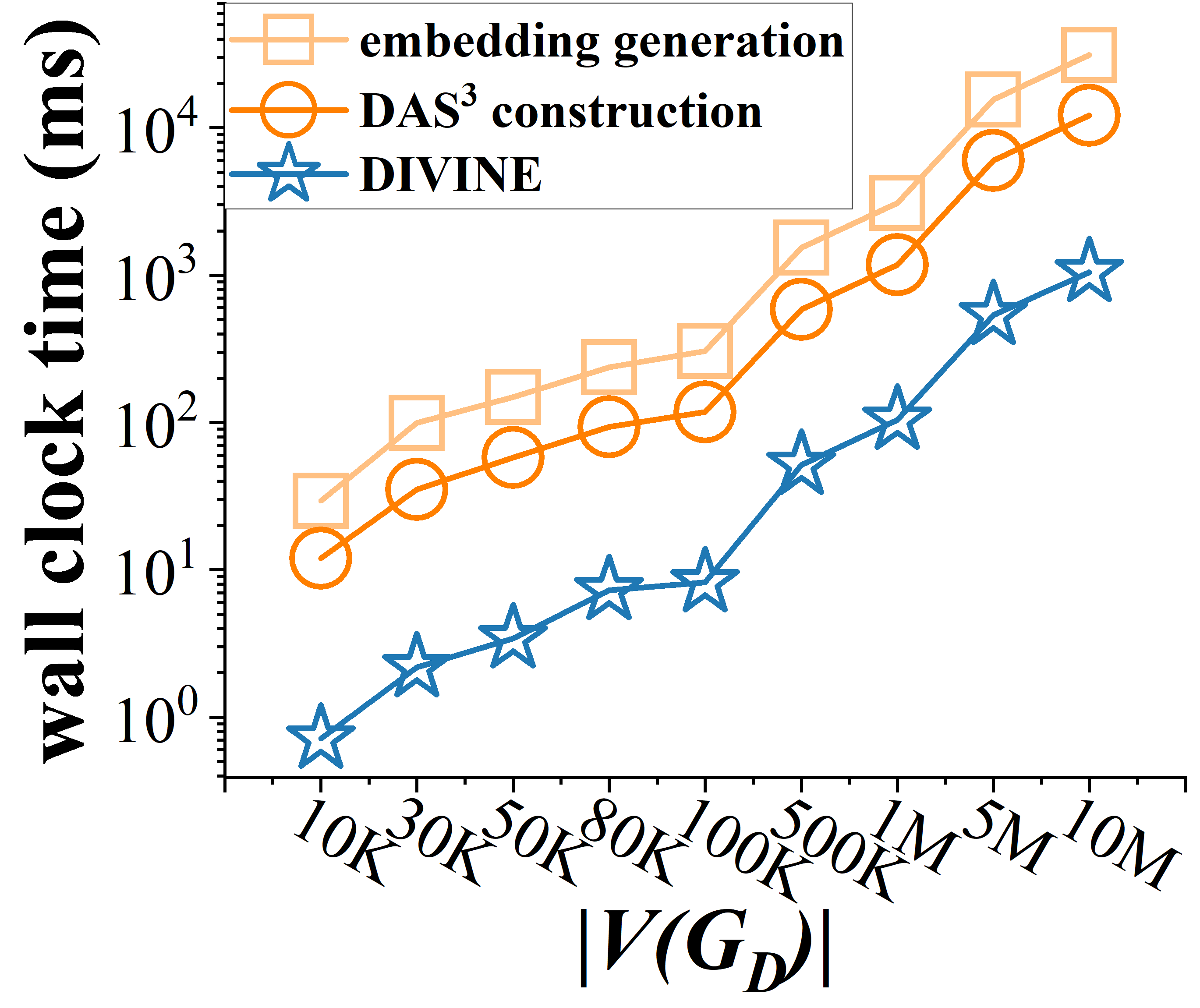}}\label{subfig:precost_zipf}}
\caption{The DIVINE pre-computation and query costs w.r.t data graph size $|V(G_D)|$.}
\label{fig:cost_vs_size}
\end{figure}

\noindent{\bf The DAS$\bm{^3}$ Synopsis Initialization Cost and Online Query Costs w.r.t. Data Graph Size $\bm{|V(G_D)|}$.}
Figure~\ref{subfig:precost_uni} evaluates the DAS$^3$ synopsis initialization cost of our DIVINE approach, including time costs of the vertex dominance embedding generation and DAS$^3$ construction over embeddings, compared with online DIVINE query time, over synthetic graphs $Syn\text{-}Uni$, where we vary the graph size $|V(G)|$ from $10K$ to $10M$ and other parameters are set to default values.
Specifically, for graph sizes from $10K$ to $10M$, the time costs of the embedding generation and DAS$^3$ construction are 0.03$\sim$31.19 $sec$ and 0.01$\sim$12.15 $sec$, respectively. 
The overall offline pre-computation time varies from 0.04 $sec$ to 43.34 $sec$, and the dynamic subgraph matching query cost is much smaller (i.e., 0.4$\sim$653.17 $ms$). 

In Figures \ref{subfig:precost_gau} and \ref{subfig:precost_zipf}, we can see similar experimental results over $Syn\text{-}Gau$ and $Syn\text{-}Zipf$ graphs, respectively, where the subgraph matching query cost is much smaller than the overall offline pre-computation time.

\section{Related Work}
\label{sec:related_work}

\noindent {\bf Continuous Subgraph Matching:}
Due to the NP completeness of the subgraph isomorphism \cite{lewis1983michael,cordella2004sub,grohe2020graph}, continuous subgraph matching (CSM) is not tractable. Several exact algorithms over dynamic graphs have been proposed \cite{fan2013incremental,choudhury2015selectivity,kankanamge2017graphflow,idris2017dynamic,idris2020general,kim2018turboflux,min2021symmetric}, which can be classified into 3 categories: \textit{recomputation-based} \cite{fan2013incremental}, \textit{direct-incremental} \cite{kankanamge2017graphflow}, and \textit{index-based incremental algorithms} \cite{choudhury2015selectivity,idris2017dynamic,idris2020general,kim2018turboflux,min2021symmetric}. The recomputation-based algorithms re-compute all the subgraph matching answers at each timestamp from scratch, the direct-incremental algorithms incrementally calculate new matching results from dynamic graphs upon updates, and index-based incremental ones incrementally maintain matching results via auxiliary indexes built over query answers.

Our proposed DIVINE approach falls into the second category of the direct-incremental algorithm. As shown in Algorithm~\ref{alg4}, for already registered query graphs, our method incrementally maintains the answer set by starting from the updated edge and leveraging vertex dominance embeddings to efficiently identify affected candidates. This update procedure is executed on-the-fly and does not rely on any query-specific or materialized index structure, which aligns with the characteristics of a direct-incremental approach. However, unlike existing works that directly use structural information (e.g., vertex label and neighbors' label set) to filter out vertex candidates, our work designs a novel and effective \textit{vertex dominance embedding} technique for candidate vertex/subgraph retrieval, which can effectively prune vertex/subgraph candidates and improve query efficiency.

Furthermore, there exist graph-based SPARQL engines such as gStore \cite{zou2014gstore}, which transform RDF graphs and SPARQL queries to labeled data and query graphs, respectively, and apply a signature-based filtering technique followed by the subgraph isomorphism verification. Although gStore supports online updates over RDF repositories, it needs to recompute the matching answer after each update for dynamic query maintenance. In contrast, our DIVINE approach is designed for general-purpose, dynamic graphs, and leverages vertex dominance embeddings to support incremental subgraph matching with provable correctness and high efficiency.

There are some studies on CSM w.r.t. specific query graph topologies (e.g., paths \cite{qin2019towards,sun2021pathenum}, cycles \cite{qiu2018real}, and cliques \cite{mondal2016casqd}) or approximate matching \cite{chen2010continuous,fan2013incremental,fan2010graph,henzinger1995computing,song2014event}. In contrast, we consider exact CSM over dynamic graphs, w.r.t. arbitrary query graph structures.

\noindent {\bf Graph Embeddings:}
Traditional heuristic-based graph embedding methods \cite{perozzi2014deepwalk,tang2015line,wang2016structural,grover2016node2vec,ribeiro2017struc2vec} are designed for static graphs and generate graph/node embeddings for fixed graph structures, which cannot be directly used for dynamic graphs. Some previous works \cite{li2019graph,bai2019simgnn,duong2021efficient,ye2024efficient} proposed to use GNNs  \cite{sun2022self,sun2024motif,wang2020gognn,hao2021ks,wang2021binarized,wang2022powerful,huang2022able} to generate graph embeddings for graph matching. However, these works either cannot guarantee the accuracy of tasks over unseen test graphs \cite{li2019graph,bai2019simgnn} (due to limitations of neural networks), or cannot efficiently and incrementally maintain embeddings in dynamic graphs with continuous updates \cite{duong2021efficient,ye2024efficient}. In contrast, our proposed vertex dominance embeddings do not use learning-based graph embedding, which can guarantee exact CSM without false negatives and enable incremental embedding updates over dynamic graphs. 

\section{Conclusion}
\label{sec:conclusion}
In this paper, we formulate and tackle the continuous subgraph matching (CSM) problem, which continuously monitors the matching subgraph answers over a large-scale dynamic graph. We propose a general framework for efficiently processing CSM queries, based on our carefully-designed \textit{\underline{D}ynam\underline{I}c \underline{V}ertex Dom\underline{IN}ance \underline{E}mbedding} (DIVINE).
We also provide an effective degree grouping technique and pruning strategies to facilitate our efficient algorithms of retrieving/maintaining CSM subgraph answers. Most importantly, we devise an effective cost model for guiding the design of dynamic vertex dominance embeddings and further enhance the pruning power of our CSM query processing. 
Extensive experiments on real/synthetic graphs confirm our DIVINE query performance. 

\section*{Acknowledgments}
This work was supported by the Natural Science Foundation of China (62272170) and the Shanghai International Joint Lab of Trustworthy Intelligent Software (22510750100). Mingsong Chen is the corresponding author.

\clearpage

\balance
\bibliographystyle{ACM-Reference-Format}
\bibliography{sample}

\end{document}